\documentclass[%
twocolumn,
nofootinbib,
 amsmath,amssymb,
]{revtex4-2}

\PassOptionsToPackage{color}{xy}

\usepackage[margin=1.6cm]{geometry}

\usepackage{mathrsfs}
\usepackage{dsfont}
\usepackage{bbm}
\usepackage{palatino}

\usepackage[usenames,dvipsnames]{xcolor}
\usepackage{graphicx, calc}
\usepackage[font=small]{caption}
\usepackage{titlesec}
\usepackage{subcaption}
\usepackage{float}
\usepackage{array}
\usepackage{multirow}
\usepackage{bm}
\usepackage{amsmath}
\usepackage{amssymb}
\usepackage{amsthm}  
\usepackage{stmaryrd}
\usepackage{ragged2e}
\usepackage{enumitem}
\usepackage[colorlinks=true,urlcolor=black,linkcolor=Mahogany,citecolor=Mahogany,plainpages=false,pdfpagelabels]{hyperref}

\usepackage{balance}

\newcolumntype{X}[1]{>{\centering\let\newline\\\arraybackslash}p{#1}}

\usepackage{subcaption}
\DeclareCaptionJustification{justified}{\justifying}
\usepackage{qcircuit}
\def\A#1{\save []="#1" \restore}
\def\qww{\qw & \qw}
\def\mX{\measure{\makebox[.9em][c]{$X$}}}
\def\mZ{\measure{\makebox[.9em][c]{$Z$}}}
\newcommand{\meterb}[1]{*=<1.8em,2.2em>{\xy 0;<0em,-.8em>:
0*{\begingroup
\everymath{\scriptstyle}
\tiny #1 \endgroup},<0em,.3em>*{\xy ="j","j"-<.778em,-.322em>;{"j"+<.778em,.322em> \ellipse ur,_{}},"j"-<0em,-.2em>;p+<.5em,.9em> **\dir{-},"j"+<2.2em,2.2em>*{},"j"-<2.2em,2.2em>*{} \endxy} 
\endxy} \POS ="i","i"+UR;"i"+UL **\dir{-};"i"+DL **\dir{-};"i"+DR **\dir{-};"i"+UR **\dir{-},"i" \qw}
\newcommand{\pone}{\mathcal{Q}_1} % notation for P1 codes
\newcommand{\qp}{\mathcal{Q}}     % notation for quantum polar codes
\newcommand{\pstab}[2]{{\tiny $(#1\!\times\!10^{#2})$}}

\newcommand{\ee}{\text{\Large\ensuremath{\mathfrak{e}}}}
\newcommand{\eee}{\text{\normalsize\ensuremath{\mathfrak{e}}}}

\makeatletter
\renewcommand{\fnum@figure}{Fig. \thefigure}
\makeatother

\newcommand{\ket}[1]{{\ensuremath{\lvert#1\rangle}}}
\newcommand{\bra}[1]{{\ensuremath{\langle#1\lvert}}}
\newcommand{\oline}{\overline}
\newcommand{\cl}{\mathcal}
\newcommand{\ident}{\mathds{1}}

\newlength\myheight
\newlength\mydepth
\settototalheight\myheight{Xygp}
\DeclareMathOperator*{\argmax}{arg\,max}
\DeclareMathOperator*{\argmin}{arg\,min}

\DeclareMathOperator{\wt}{wt}

\DeclareMathOperator{\eqdef}{:=}
\DeclareMathOperator{\xor}{XOR}
\DeclareMathOperator{\cnot}{CNOT}

\newtheorem{procedure}{Procedure}
\newtheorem{lemma}{Lemma}
\newtheorem{theorem}{Theorem}

\titlespacing*{\paragraph}{\parindent}{1.8ex plus .2ex minus .2ex}{2ex plus .5ex minus .2ex}

\titlespacing{\section}{0pt}{2.5ex plus .2ex minus .2ex}{1.5ex plus .2ex minus .2ex}

\titlespacing{\subsection}{0pt}{2.5ex plus .2ex minus .2ex}{1.5ex plus .2ex minus .2ex}

\renewcommand{\thesection}{\arabic{section}}

\titleformat*{\section}{\fontsize{10.5}{13}\selectfont\bfseries\center}

\makeatletter 
    
\renewcommand\onecolumngrid{% <<<<<<
\do@columngrid{one}{\@ne}%
\def\set@footnotewidth{\onecolumngrid}% <<<<<<<<<<<<<<<<
\def\footnoterule{\kern-6pt\hrule width 1.5in\kern6pt}%
}

\renewcommand\twocolumngrid{% <<<<<<
        \def\footnoterule{% restore rule
        \dimen@\skip\footins\divide\dimen@\thr@@
        \kern-\dimen@\hrule width.5in\kern\dimen@}
        \do@columngrid{mlt}{\tw@}
}%

\makeatother

\begin{document}

\title{Fault-Tolerant Preparation of Quantum Polar Codes Encoding One Logical Qubit}

\author{Ashutosh Goswami$^1$, \qquad Mehdi Mhalla$^2$, \qquad Valentin Savin$^1$\\[2mm]
    {\small $^1$\,Univ. Grenoble Alpes, CEA-L\'eti, F-38054 Grenoble, France}\\
    {\small $^2$\,Univ. Grenoble Alpes, CNRS, Grenoble INP, LIG, F-38000 Grenoble, France}\\
    {\small ashutosh-kumar.goswami@cea.fr, mehdi.mhalla@univ-grenoble-alpes.fr, valentin.savin@cea.fr}
}
\date{\today}

%\tableofcontents

%\begin{abstract}
%This paper explores a new approach to fault-tolerant quantum computing (FTQC), relying on quantum polar codes. We consider quantum polar codes of Calderbank-Shor-Steane type, encoding one logical qubit, which we refer to as $\pone$ codes. First, we show that a subfamily of $\pone$ codes is equivalent to the well-known family of Shor codes. Moreover, we show that  $\pone$ codes significantly outperform Shor codes, of the same length and minimum distance. Second, we consider the fault-tolerant preparation of $\pone$ code states. We give a recursive procedure to prepare a $\pone$ code state, based on two-qubit Pauli measurements only. The procedure is not by itself fault-tolerant, however, the measurement operations therein provide redundant classical bits, which can be advantageously used for error detection. Fault-tolerance is then achieved by combining the proposed recursive procedure with an error detection method. Finally, we consider the fault-tolerant error correction of $\pone$ codes. We use Steane error correction, which incorporates the proposed fault-tolerant code state preparation procedure. We provide numerical estimates of the logical error rates for $\pone$ and Shor codes of length $16$ and $64$ qubits, assuming a circuit-level depolarizing noise  model. Remarkably, the $\pone$ code of length $64$ qubits achieves a logical error rate very close to $10^{-6}$ for the physical error rate $p = 10^{-3}$, therefore, demonstrating the potential of the proposed polar codes based approach to FTQC.
%\end{abstract}

\maketitle

\onecolumngrid

\vspace*{-5mm}\,\hfill\begin{minipage}{.85\textwidth}
\small 
This paper explores a new approach to fault-tolerant quantum computing (FTQC), relying on quantum polar codes. We consider quantum polar codes of Calderbank-Shor-Steane type, encoding one logical qubit, which we refer to as $\pone$ codes. First, we show that a subfamily of $\pone$ codes is equivalent to the well-known family of Shor codes. Moreover, we show that  $\pone$ codes significantly outperform Shor codes, of the same length and minimum distance. Second, we consider the fault-tolerant preparation of $\pone$ code states. We give a recursive procedure to prepare a $\pone$ code state, based on two-qubit Pauli measurements only. The procedure is not by itself fault-tolerant, however, the measurement operations therein provide redundant classical bits, which can be advantageously used for error detection. Fault-tolerance is then achieved by combining the proposed recursive procedure with an error detection method. Finally, we consider the fault-tolerant error correction of $\pone$ codes. We use Steane error correction, which incorporates the proposed fault-tolerant code state preparation procedure. We provide numerical estimates of the logical error rates for $\pone$ and Shor codes of length $16$ and $64$ qubits, assuming a circuit-level depolarizing noise  model. Remarkably, the $\pone$ code of length $64$ qubits achieves a logical error rate very close to $10^{-6}$ for the physical error rate $p = 10^{-3}$, therefore, demonstrating the potential of the proposed polar codes based approach to FTQC.
\end{minipage}\hfill\,

\vspace*{3mm}

\twocolumngrid 

Large scale quantum computers are expected to use quantum error correcting (QEC) codes to provide resilience against noise~\cite{preskill1998fault}. A  QEC code encodes one or more logical qubits into many noisy physical qubits, so that the logical qubits are more robust against noise than the physical qubits.  However, a  QEC code alone does not provide the ability to do fault-tolerant quantum computation (FTQC). To avoid the uncontrolled propagation of errors, it must be complemented with  several \emph{fault-tolerant} procedures~\cite{gottesman2010introduction}, aimed at $(i)$~preparing logical code states, $(ii)$~operating on logical states, and $(iii)$~performing error correction. 

In this paper, we explore a new approach to fault-tolerant quantum computation (FTQC), relying on quantum polar codes. Introduced first in 2009 for classical systems~\cite{arikan2009channel}, and then generalized to the quantum case~\cite{renes2011efficient, wilde2013polar, renes2014polar, dupuis2021polarization}, polar codes arguably represent one of the most important advances of the past decade in the coding theory. They achieve the coherent information (one-shot capacity) of any quantum channel, and come equipped with an efficient decoding algorithm, known as successive cancellation (SC), whose complexity scales log-linearly with the code length. 

It is worth noticing that a low-complexity decoding algorithm is key to performing fault-tolerant error correction. Indeed, the decoding must be faster than the syndrome extraction rate, since otherwise, the latency overhead becomes exponential in the number of non-Clifford gates, hindering any quantum advantage ~\cite{holmes2020nisq}.

Yet, despite their excellent error correction properties, polar codes have been hardly explored for quantum computing, except the work in \cite{krishna2018magic} on magic state distillation. Here, we focus on two closely related ingredients of FTQC, namely fault-tolerant code state preparation and fault-tolerant error correction.

The main contributions are as follows.

We consider quantum polar codes of Calderbank-Shor-Steane (CSS) type that encode one logical qubit, which we refer to as $\pone$ codes. We show that $\pone$ codes are a natural generalization of the well-known family of Shor codes, providing improved error correction performance.

We then consider the fault-tolerant preparation of $\pone$ code states under the effect of noise. We propose a procedure to prepare $\pone$ code states, by recursively performing  Pauli $Z\otimes Z$ or Pauli $X\otimes X$ measurements. This procedure is not by itself fault-tolerant, however, the measurement operations therein provide redundant classical bits, which can be advantageously used for error detection. Hence, to achieve fault-tolerance, the proposed procedure is complemented by an error detection method. 

Finally, we consider the fault-tolerant error correction of $\pone$ codes, using Steane error correction \cite{steane1997active, steane2002fast}. We provide numerical estimates of the logical error rate (LER), assuming a circuit-level depolarizing noise  model, for $\pone$ and Shor-$\pone$ codes of length $N=16$ and $N=64$. Remarkably, the $\pone$ code of length $64$ qubits achieves an LER very close to $10^{-6}$ for the physical error rate $p = 10^{-3}$, therefore, demonstrating the potential of the proposed polar codes based approach to FTQC.

\smallskip Relevant background information, as well as proofs of the main theorems and additional numerical results are provided in the Supplemental material appended to this paper. 

\section{CSS Quantum Polar Codes} 

The quantum polar transform $Q_N$, where $N = 2^n$, with $n > 0$,  is the unitary operation on $N$ qubits that operates in the computational basis as the classical polar transform $P_N$. Precisely, for any $\bm{u} = (u_1, \dots, u_N) \in \{0,1\}^N$, we define $Q_N\ket{\bm{u}} = \ket{P_N\bm{u}}$, where $P_N = \big( \begin{smallmatrix} 0 & 1 \\ 1 & 1 \end{smallmatrix}\big)^{\otimes n} $. Hence, $Q_N$  can be realized by recursively applying the quantum CNOT gate, transversely, on subblocks of $2^k$ qubits, for $k = 0,...,n-1$ (see Fig. \ref{fig:qpolar_N8}).

\smallskip Let $ \mathcal{S} = \{1,...,N\}$ denote an $N$-qubit quantum system. For a CSS quantum polar code \cite{renes2011efficient}, the system $\mathcal{S}$ is partitioned into $\mathcal{S} = \mathcal{Z} \cup \mathcal{I} \cup \mathcal{X}$, and the input quantum state of the polar transform is taken as follows.

\smallskip For $\mathcal{Z} \subseteq \mathcal{S}$, the quantum state is frozen to a known Pauli $Z$ basis state $\ket{\bm{u}}_\mathcal{Z}$, where $\bm{u} := (u_1, \dots, u_n) \in \{0, 1\}^{|\mathcal{Z}|}$. For $\mathcal{X} \subseteq \mathcal{S}$, with $\mathcal{Z} \cap \mathcal{X} = \emptyset$, it is frozen to a known Pauli $X$ basis state $\ket{\bm{\oline{v}}}_\mathcal{X}$, where $\bm{v} \in \{0, 1\}^{|\mathcal{X}|}$ and we use the notation $\ket{\bar{0}} := \ket{+}$, and $\ket{\bar{1}} := \ket{-}$.  The remaining subset $\mathcal{I} := \mathcal{S}\setminus (\mathcal{X} \cup \mathcal{Z})$ is used to encode quantum information $\ket{\phi}_\cl{I}$. 

\smallskip Therefore, the logical code state, denoted by  $\ket{\widetilde{\phi}}_\mathcal{S}$, is given by $\ket{\widetilde{\phi}}_\mathcal{S} = Q_N (\ket{\bm{u}}_\mathcal{Z} \otimes \ket{\phi}_\mathcal{I}   \otimes \ket{\bm{\oline{v}}}_\mathcal{X})$.

\begin{figure}[H]
\,\hfill\scalebox{0.95}{\input{qpolar_N8}}\hfill\,

\vspace*{0.5em}
\caption{ \small The quantum polar code, with $N=2^3$, frozen sets $\cl{Z}=\{1,2,3\}$, $\cl{X}=\{6,7,8\}$ (chosen only for the simplicity of illustration), and frozen states $\ket{\bm{u}}_\mathcal{Z} =\ket{0,0,0}$, and $\ket{\bm{\oline{v}}}_\mathcal{X}= \ket{+,+,+}$. }
\label{fig:qpolar_N8}
\end{figure}

\section{$\pone$ codes: Quantum Polar Codes Encoding One Logical Qubit } \label{sec:q1code}

As $\pone$ codes encode only one qubit, the information set $\mathcal{I} = \{i\}$, for some $i \in \mathcal{S} = \{1, \dots, N\}$. Given the index $i$, the frozen sets $\mathcal{Z}$ consists of the set of indices preceding $i$,  $i.e$, $\mathcal{Z} \eqdef \{1, \dots, i-1\}$, and the frozen set $\mathcal{X}$ consists of the set of indices succeeding $i$,  $i.e$, $\mathcal{X} \eqdef \{i+1, \dots, N\}$.

The choice of the frozen sets $\cl{Z}$ and $\cl{X}$, given the information qubit position $i$, is simply explained by the sequential nature of the SC decoding, and the fact that only logical errors corresponding to $i$ ($i.e.$, $X$ and $Z$ errors on position $i$ at the level of uncoded qubits) need to be corrected. 
Indeed, freezing qubit $j > i$ in $Z$ basis would be of no help in decoding of logical $X$ error on $i$, while freezing it in $X$ basis amounts to ignoring $X$ error that happen on it at the level of uncoded qubits, as it corresponds to a $X$ type stabilizer. A similar observation holds for $Z$ errors.

In case the information position is a power of two, that is, $i = 2^k, 0\leq k \leq n$, the corresponding $\pone$ code is a Shor code\cite{shor1995scheme,bacon2006operator}. This follows from Theorem~\ref{thm:polar-shor} below. We refer to these codes as \emph{Shor-$\pone$ codes}, or simply Shor codes, when no confusion is possible. 

\begin{theorem} \label{thm:polar-shor}
 For a given $i = 2^k, 0\leq k \leq n$,  the logical states $\ket{\widetilde{0}}_\mathcal{S}$ and $\ket{\widetilde{1}}_\mathcal{S}$ of the $\pone(N, i)$ code are as follows (up to a normalization factor),
\begin{align}
\ket{\widetilde{0}}_\mathcal{S} &= \otimes_{r=1}^{2^k} \left( \otimes_{c=1}^{2^{n-k}} \ket{+}_{r,c} + \otimes_{c=1}^{2^{n-k}} \ket{-}_{r,c} \right), \label{eq:polar-shor-log0} \\
\ket{\widetilde{1}}_\mathcal{S}  &= \otimes_{r=1}^{2^k} \left( \otimes_{c=1}^{2^{n-k}} \ket{+}_{r,c} - \otimes_{c=1}^{2^{n-k}} \ket{-}_{r,c} \right), \label{eq:polar-shor-log1}
\end{align}
where $r$ and $c$ are row  and column indexes, with $\cl{S}$ being reshaped as a $2^k \times 2^{n-k}$ matrix of qubits. 
\end{theorem}
The \emph{construction} of a $\pone$ code refers to the choice of the information position $i$, which determines how  well the code protects the encoded quantum information. Hence, the position~$i \in \mathcal{S}$ should be chosen in a way to optimize the LER performance, depending on the specific noisy quantum channel.  For the subfamily of Shor-$\pone$ codes, the choice is restricted to positions $i = 2^k, 0\leq k \leq n$.

\smallskip For depolarizing quantum channels, we use the density evolution technique \cite{tal2013construct, arikan2009channel} to numerically estimate the LER of $\pone$ and Shor-$\pone$ codes. The numerical results are shown in 
Fig.~\ref{fig:depolar_indep_xz_dec_main} for both $\pone$ and Shor codes, for even recursion levels $n = 4 \text{ to } 12$, where the legend gives the  best index $i$ for low physical error rates.  It can be observed that the $\pone$ code in general outperforms the Shor code, for a given $n$. From the corresponding values of $i$ given in Fig.~\ref{fig:depolar_indep_xz_dec_main}, one may further note that the $\pone$ and Shor codes have the same quantum minimum distance, for a given value of $n$. The superior performance of $\pone$ codes owes to the SC decoding, which is able to decode beyond the minimum distance of the code, by effectively exploiting the  channel polarization property of polar codes.

For the odd recursion levels, our numerical results (omitted here) indicate  performance gains, depending on $n$ and the physical error rate value.

\begin{figure}[!t]
\includegraphics[width=\linewidth]{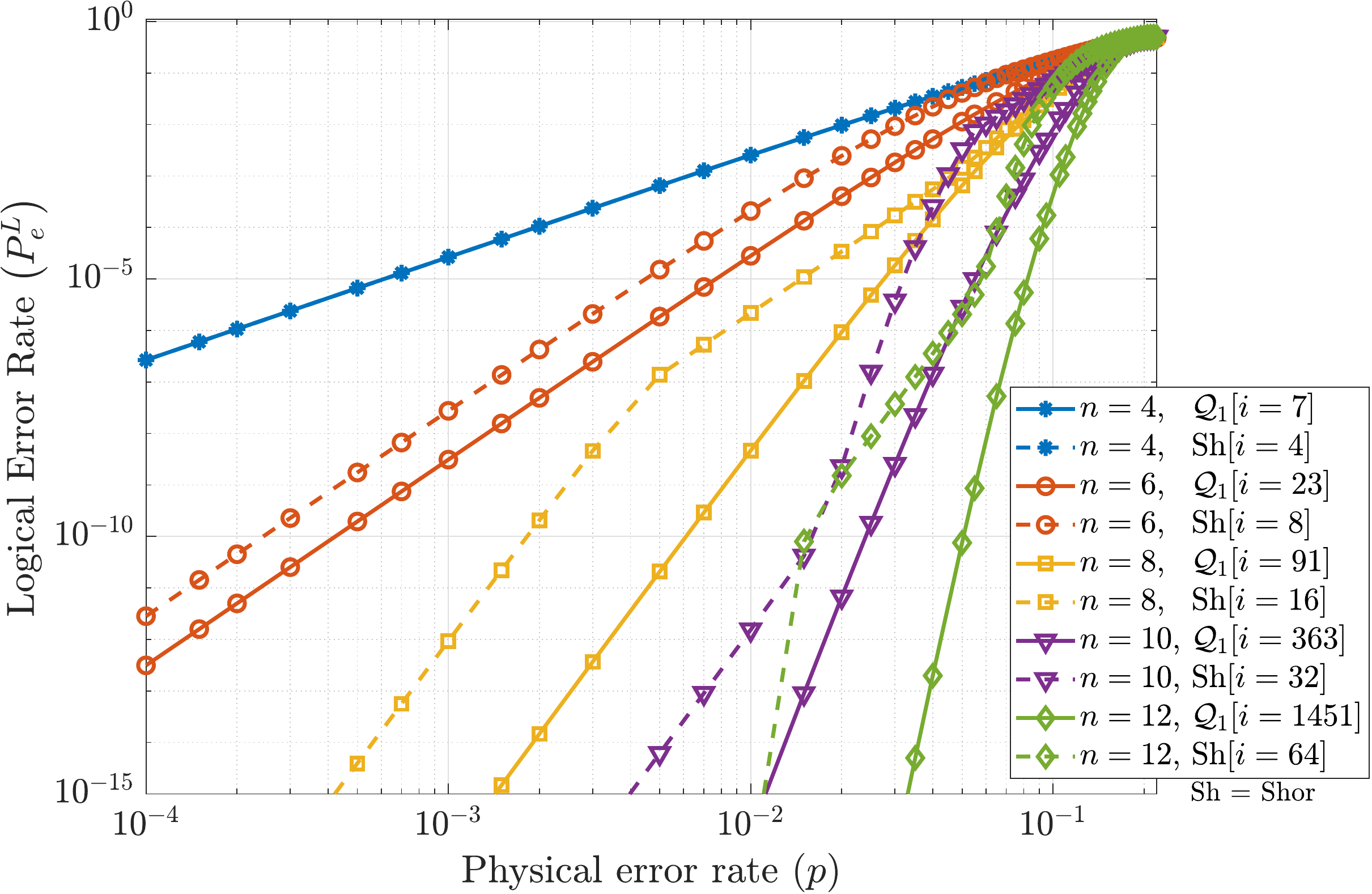}
\caption{LER of $\pone$ and Shor codes, for the depolarizing channel.}
\label{fig:depolar_indep_xz_dec_main}
\end{figure}

\section{Measurement Based Preparation of $\pone$ Code States}

The conventional encoding of quantum polar codes in Fig.~\ref{fig:qpolar_N8} is not fault-tolerant, as errors propagate  through the $\cnot$ gates. Further, measuring the stabilizer generators using the  standard ``phase kickback trick''~\cite{gottesman2010introduction}, similar to the case of quantum LDPC codes~\cite{gottesman2014fault}, is also not fault-tolerant, due to the high weights of generators.

\smallskip Hence, we propose a new procedure to prepare $\pone$ code states, based on two-qubit Pauli measurements only. We describe our procedure in two steps. First, we assume that all the operations are error free, and show that the proposed procedure does indeed prepare a $\pone$ code state. Then, we consider our procedure under the effect of errors (\emph{i.e.}, noisy gates and measurements), and show that it can be made fault-tolerant by incorporating an error detection gadget, exploiting the redundancy in the measurement outcomes.

\smallskip We consider the preparation of logical $\ket{0}$ and $\ket{+}$ states, for which all the input qubits are frozen in  either $Z$ or $X$ basis. Consequently,  we consider the preparation of general $\pone$ code states, with frozen sets $\mathcal{Z} = \{1,\dots,i\}$ and $\mathcal{X} = \{i + 1,\dots,N\}$, for some arbitrary $ 1 \leq i \leq N$, where $N=2^n$,  $n \geq 1$. Further, since our preparation procedure is recursive, to clearly indicate the length of the prepared $\pone$ state, we will use the notation $i(n) \eqdef i$, $\mathcal{Z}(n) \eqdef \mathcal{Z}$, and $\mathcal{X}(n) \eqdef \mathcal{X}$. Therefore, we want to prepare the following $N$-qubit $\pone$ state on the system $\mathcal{S} = \{1, \dots, N\}$,
\begin{equation}
\ket{q_N}_\mathcal{S} := Q_N \left( \ket{\bm{u}, \oline{\bm{v}}}_\mathcal{S}\right) = Q_N \left(\ket{\bm{u}}_{\cl{Z}(n)} \otimes \ket{ \oline{\bm{v}}}_{\cl{X}(n)}\right), \label{eq:q-prep-state}
\end{equation}
where $\bm{u} \in \{0,1\}^{i(n)}$ and $\bm{v} \in \{0,1\}^{N - i(n)}$. When no confusion is possible, we may simply write $\ket{q_N}$ instead of $\ket{q_N}_\mathcal{S}$. Here, we consider $\pone$ states defined by the same value of $i(n)$ as equivalent, regardless of the corresponding frozen values $\bm{u}, \bm{v}$. Indeed, equivalent $\pone$ states are defined by the same stabilizer generators, up to sign factors.

To prepare $\ket{q_N}_\mathcal{S}$ from (\ref{eq:q-prep-state}), we consider the following measurement based procedure.

\begin{procedure}[Measurement Based Preparation] \label{prot:prep}
Given a $n$-bit sequence $b_1\cdots b_n \in \{0, 1\}$, our measurement based procedure on $N = 2^n$ qubit system $\mathcal{S} = \{1, \dots, N\}$ is as follow.

\begin{list}{}{\setlength{\labelwidth}{2em}\setlength{\leftmargin}{1.7em}\setlength{\listparindent}{0em}}

\item[$(1)$] First, $\mathcal{S}$ is initialized in a Pauli $Z$ basis state $\ket{\bm{u}}_{\mathcal{S}}, \bm{u} \in \{0, 1\}^N$.

\item[$(2)$] Then, two-qubit Pauli measurements are recursively applied for $n$ levels. The recursion is the same as the recursion of the quantum polar transform (see Fig. \ref{fig:qpolar_N8}), except the $\cnot$ gate is replaced by either Pauli $X \otimes X$ or $Z \otimes Z$ measurement. Precisely, if $b_k = 0$ (or, $b_k = 1$), we apply Pauli $X \otimes X$ (or, $Z \otimes Z$) measurements at the $k^{th}, k = 1, \dots, n$ recursion level.
\end{list}
\end{procedure}

\begin{theorem} \label{thm:Q1_code_prep}
Consider the $\pone$ state $\ket{q_N}_\mathcal{S}$ from (\ref{eq:q-prep-state}), with $1 \leq i(n) \leq N$. Let $b_1\cdots b_n$ be the binary representation of $i(n)-1$, with $b_n$ being the most significant bit, $i.e$, $i(n)-1 = \sum_{k=1}^{n} b_k 2^{k-1}$. Then, $\ket{q_N}_\mathcal{S}$ can be prepared, using the measurement based procedure in Procedure \ref{prot:prep}, corresponding to the $n$ bit sequence $b_1\cdots b_n$.
\end{theorem}
The measurement based preparation for $N = 8, i(n) = 3$ is illustrated in Fig. \ref{fig:qpolarprep_N8_i3}.

\begin{figure}[!t]
\,\hfill\hspace*{2em}\scalebox{0.95}{\input{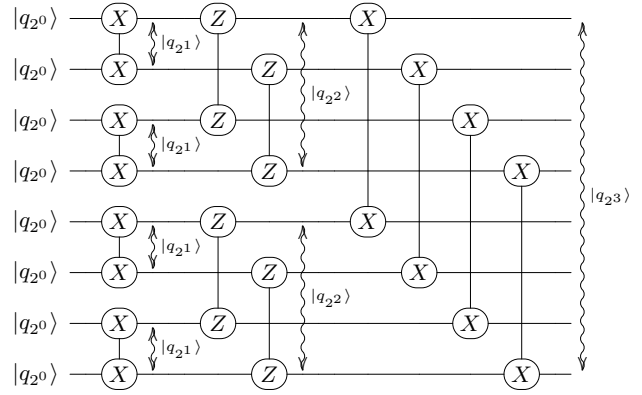}}\hfill\,

\caption{Measurement-based preparation of $\ket{q_N}_\mathcal{S}$ in (\ref{eq:q-prep-state}), with $N = 8, i(n) = 3$,  where slightly flattened circles connected by a~vertical wire denote either a $X \otimes X$ or a $Z \otimes Z$ ~measurement on the corresponding qubits, and $\ket{q_{2^k}}$ are equivalent $\pone$ states of length $2^k$ ($\ket{q_{2^0}}$ is a Pauli $Z$ basis state).}
\label{fig:qpolarprep_N8_i3}
\end{figure}

\smallskip We first show in Lemma \ref{lem:Q1_code_prep} that given two equivalent $\pone$ states of length $K/2, K = 2^k$, we can prepare a $\pone$ state of length $K$ by performing $Z \otimes Z$ or $X \otimes X$ measurements transversely on them. Further, when we apply Pauli $Z \otimes Z$ measurements, we have $i(k) = i(k-1) +  K/2 > K/2$ and when we apply Pauli $X \otimes X$ measurements, we have $i(k) = i(k-1) \leq K/2$. The proof of Theorem \ref{thm:Q1_code_prep} then simply follows from Lemma \ref{lem:Q1_code_prep}, by noting that $b_k = 1 \Leftrightarrow i(k) > K/2$ (hence, $b_k = 0 \Leftrightarrow i(k) \leq K/2$).

\begin{lemma} \label{lem:Q1_code_prep}
Consider two equivalent $\pone$ states on $K/2$-qubit systems $\mathcal{S}_1 : = \{1, \dots, K/2 \}$ and $\mathcal{S}_2 : = \{K/2+1, \dots, K\}$ as follows, $\ket{q_{\frac{K}{2}}^1}_{\mathcal{S}_1} \eqdef Q_{\frac{K}{2}} \ket{\bm{u_1}, \oline{\bm{v}}_{\bm{1}}}_{\mathcal{S}_1}$ and $\ket{q_{\frac{K}{2}}^2}_{\mathcal{S}_2} \eqdef Q_{\frac{K}{2}} \ket{\bm{u_2}, \oline{\bm{v}}_{\bm{2}}}_{\mathcal{S}_2}$, where $\bm{u_1}, \bm{u_2} \in \{0, 1\}^{i(k-1)}$ and $\bm{v_1}, \bm{v_2} \in \{0, 1\}^{\frac{K}{2}-i(k-1)}$, with $1 \leq i(k-1) \leq K/2$. Let $\mathcal{S} := \mathcal{S}_1 \cup \mathcal{S}_2$ be the joint system, then we have the following two cases.

\begin{list}{}{\setlength{\labelwidth}{2em}\setlength{\leftmargin}{0.3em}\setlength{\listparindent}{0em}}
\item Case $1$: If we apply transversal Pauli $Z \otimes Z$ measurements on the corresponding qubits of $\cl{S}_1$ and $\cl{S}_2$, we get the $K/2$ bit measurement outcome as follows,
\begin{equation}
 \bm{m} = P_\frac{K}{2}(\bm{u'}, \bm{x}) \in \{0, 1\}^{\frac{K}{2}}, \label{eq:m-out}
\end{equation}
where $\bm{u'} = \bm{u_1} \oplus \bm{u_2} \in \{0, 1\}^{i(k-1)}$ and  $\bm{x} \in \{0, 1\}^{\frac{K}{2}-i(k-1)}$ is a random vector, and $P_{\frac{K}{2}}$ is the classical polar transform. After measurements, the state of $\mathcal{S}$ is a $\pone$ state, $\ket{q_K}_\mathcal{S} =  Q_K \ket{(\bm{u'}, \bm{x}, \bm{u_2}), \overline{\bm{v_1} \oplus \bm{v_2}}}_\mathcal{S}$, with $i(k) = i(k-1) + K/2 > K/2$, and where $\bm{x}$ is determined from the measurement outcome $\bm{m}$ in (\ref{eq:m-out}) by, $\bm{x} = P_{\frac{K}{2}}(\bm{m})\lvert_{\mathcal{X}(k-1)}$, $i.e.$, the subvector of  $P_{\frac{K}{2}}(\bm{m}) \in \{0, 1\}^{K/2}$ corresponding to indices in the set $\mathcal{X}(k-1)$.

\item Case $2$: If we apply transversal Pauli $X \otimes X$ measurements on the corresponding qubits of $\cl{S}_1$ and $\cl{S}_2$, we get the $K/2$ bit measurement outcome as follows,
\begin{equation}
 \bm{m} = P_\frac{K}{2}^\top(\bm{z}, \bm{v'}) \in \{0, 1\}^{\frac{K}{2}}, \label{eq:m-out-x-main}
\end{equation}
where $\bm{z} \in \{0, 1\}^{i(k-1)}$ is a random vector, and $\bm{v'} = \bm{v_1} \oplus \bm{v_2} \in \{0, 1\}^{\frac{K}{2}-i(k-1)}$. After measurements, the state on $\mathcal{S}$ is a $\pone$ state $\ket{q_K}_\mathcal{S} = Q_K \ket{\bm{u_1} \oplus \bm{u_2}, \overline{(\bm{v_1},  \bm{z},  \bm{v'})}}_\mathcal{S}$, with $i(k) = i(k-1) \leq K/2$ and from (\ref{eq:m-out-x-main}), $\bm{z} = P_{\frac{K}{2}}^\top (\bm{m})\lvert_{\mathcal{Z}(k-1)}$.

\end{list}
\end{lemma}

\section{Fault-Tolerant Measurement Based Procedure} We now consider our measurement based procedure under the effect of Pauli noise. We assume the standard implementation of Pauli $Z \otimes Z$ and $X \otimes X$  measurements, where a bare ancilla qubit is initialized in either Pauli $Z$ or Pauli $X$ basis state, then two CNOT gates are applied between data and ancilla qubits, and finally the ancilla qubit is measured in Pauli $Z$ or $X$ basis. Therefore, any two-qubit Pauli measurement decomposes into four basic \emph{components}, namely one single-qubit initialization, two $\cnot$-gates, and one single-qubit measurement. 

\smallskip As the preparation of $\ket{q_N}_\mathcal{S}$ consists of $N$ single-qubit initializations, followed by $N/2 \log N$ two-qubit Pauli measurements, the total number of components in the preparation is equal to $N (1 + 2 \log N)$.

\smallskip We assume that each component fails independently with some probability $p$, according to a circuit level Pauli noise model as follows.

%Recall that our recursive procedure consists of two main operations; intialization in a Pauli $Z$ basis state and Pauli $Z\otimes Z$ or $X\otimes X$ measurements. 

\smallskip A failure in a CNOT gate corresponds to applying the perfect CNOT gate, followed by a two-qubit Pauli error on the output qubits of the CNOT gate. A failure in initialization in Pauli $Z$ (or $X$) basis corresponds to the perfect initialization, followed by an $X$ (or $Z$) error on the initialized qubit. A failure in Pauli $Z$ (or $X$)  measurement corresponds to first applying a Pauli $X$ (or $Z$) error on the qubit to be measured, and then doing the perfect Pauli $Z$ (or $X$) measurement.

\smallskip We note that the preparation in Procedure~\ref{prot:prep} is not fault-tolerant by itself. Due to failures in the components, the measurement outcomes of transversal Pauli $Z \otimes Z$ or $X \otimes X$ measurements are noisy. Precisely, we have $\bm{m} = P_\frac{K}{2}(\bm{u'}, \bm{x}) \oplus \bm{e}_X$ instead of (\ref{eq:m-out}), and $\bm{m} = P_\frac{K}{2}^\top(\bm{z}, \bm{v'}) \oplus \bm{e}_Z$ instead of (\ref{eq:m-out-x-main}), where $\bm{e}_X, \bm{e}_Z \in \{0, 1\}^{K/2}$ are unknown error terms. Recall from Lemma \ref{lem:Q1_code_prep} that vectors $\bm{x}$ and $\bm{z}$ are determined from $\bm{m}$, and they are necessary to know the prepared state $\ket{q_K}_\mathcal{S}$ after Pauli $Z \otimes Z$ and $X \otimes X$ measurements, respectively.  However, due to unknown error terms, the methods in Lemma \ref{lem:Q1_code_prep} may not correctly determine $\bm{x}$ and $\bm{z}$.  Accepting a wrong estimate $\hat{\bm{x}} \neq \bm{x}$ and  $\hat{\bm{z}} \neq \bm{z}$ amounts to extra $X$ and $Z$ errors on the respective prepared states, given by the vectors $P_K(0, \hat{\bm{x}} \oplus \bm{x}, 0, 0)$, and $P_K^\top(0, 0, \hat{\bm{z}} \oplus \bm{z}, 0)$, respectively.

\smallskip To make the measurement based preparation fault-tolerant, we consider the following error detection procedure.

\begin{procedure}[Measurement based Preparation with Error Detection] \label{prot:err_det}
Consider the preparation of a $\pone$ state of length $N$ from Procedure~\ref{prot:prep}.  We further incorporate an error detection gadget within each level of reccursion, $k=1,\dots, n$, consisting of the following two steps.

\begin{list}{}{\setlength{\labelwidth}{2em}\setlength{\leftmargin}{1.7em}\setlength{\listparindent}{0em}}
\item[$(1)$]  For all $2^{n-k}$ instances of prepared $\ket{q_K}$ states at the $k^{th}$ level of recursion, we first determine the syndrome of the error in the measurement outcome $\bm{m}$ as follows. When Pauli $Z \otimes Z$ measurements are performed ($i.e$, Case $1$ of Lemma \ref{lem:Q1_code_prep}), we determine the syndrome of the error term $\bm{e}_X$ in the measurement outcome $\bm{m}$ as, $P_{\frac{K}{2}}(\bm{e}_X) \lvert_{\mathcal{Z}(k-1)} = P_{\frac{K}{2}}(\bm{m})\lvert_{\mathcal{Z}(k-1)} \,\oplus\, \bm{u'}$. Similarly, when Pauli $X \otimes X$ measurements are performed ($i.e$, Case $2$ of Lemma \ref{lem:Q1_code_prep}), we determine the syndrome of the error term $\bm{e}_Z$ as, $P_{\frac{K}{2}}^\top(\bm{e}_Z) \lvert_{\mathcal{X}(k-1)} = P_{\frac{K}{2}}^\top(\bm{m})\lvert_{\mathcal{X}(k-1)} \,\oplus\, \bm{v}'$.

\item[$(2)$] If the syndrome is the zero vector for all the $2^{n-k}$ instances of $\ket{q_K}$, we determine the value of $\bm{x}$ or $\bm{z}$ for all prepared states as in Lemma~\ref{lem:Q1_code_prep}, and proceed to the next level of recursion. Otherwise, we declare a preparation failure.
 
\end{list}
\end{procedure}
In case Procedure \ref{prot:err_det} fails, we may restart the procedure from the begining, by initializing $N$ qubits in a Pauli $Z$ basis state.

\smallskip The fault-tolerance of the successfully prepared state ($i.e.$, when errors are not detected at any recursion level $k = 1 \text{ to } n$), follows from Theorem \ref{thm:fault_tolerance}.

\begin{theorem} \label{thm:fault_tolerance}
Consider the measurement based preparation with error detection from Procedure \ref{prot:err_det}. Suppose a successful preparation of $\ket{q_N}$, where $T_n$ component failures occur during the preparation.  Let $\bm{e}_X^f, \bm{e}_Z^f \in \{0, 1\}^N$ be the final Pauli $X$ and $Z$ errors in the noisy prepared state $\ket{q^\prime_N}$, due to the component failures. Then, there exist equivalent errors  $\bm{e}_X^{\prime f} \equiv \bm{e}_X^f $ and $\bm{e}_Z^{ \prime f} \equiv \bm{e}_Z^f$, so that $\wt(\bm{e}_X^{\prime f}) \leq T_n$ and $\wt(\bm{e}_Z^{\prime f}) \leq T_n$, where $\wt(\bm{u})$ denotes the Hamming weight of $\bm{u}$.
\end{theorem}

Theorem \ref{thm:fault_tolerance} implies that the weight of $X$ and $Z$ errors on a successfully prepared state remains small given a sufficiently low component failure probability $p$. In particular, it upper bounds the average weight of the final error by $N(1+ 2 \log N)p$. Our numerical simulation suggests that the average error weight is much lower than $N(1+ 2 \log N)p$, which is expected as we discard the preparations where errors are detected.

\section{Fault-tolerant error correction} For fault-tolerant error correction of $\pone$ codes, we consider Steane error correction. The $\pone$ states, needed in the Steane error correction, are prepared using Procedure \ref{prot:err_det}, assuming a circuit-level depolarizing noise model.  Further, as Procedure \ref{prot:err_det} consists of error detection, we consider polar code states of small lengths $N = 16$ and $N=64$.

\smallskip Let $p_{\text{prep}}= \lim_{R \to \infty} \frac{t}{R}$ be the preparation rate, where $t$ is the number of successful preparations out of $R$ preparation attempts. For a component failure probability $p = 10^{-3}$, % and $R = 10^{5}$, 
our numerical simulation gives $p_{\text{prep}} \approx 0.88$ and $p_{\text{prep}} \approx 0.47 $, respectively, for $\pone(N = 16, i = 7)$ and $\pone(N = 64, i = 23)$ codes.

\smallskip The numerical estimate of the LER  for $\pone$ and Shor codes of length $N=16$ and $N=64$ qubits are given in Fig.~\ref{fig:p1_vs_shor_main}. The LER has been estimated by Monte-Carlo (MC) simulation of the Steane error correction procedure, until a number of $f$ logical errors are reported, where $f$ is taken to be between $50$ and $200$. We also provide a theoretical upper-bound of the LER based on density evolution (DE), providing a trustworthy extrapolation of the LER for smaller values of $p$, which is a unique feature (among quantum codes) of great interest in practical applications. 

\smallskip Finally, it can be  observed that the pseudothreshold (crossing point between the LER curve and the diagonal line \cite{svore2006flow, tomita2014low}) of the $\pone(N=16, i=7)$ code is $p_\text{th} \approx  0.001$, while for the  $\pone(N=64, i=23)$ code,  we get $p_\text{th} \approx  0.01$.  If the
physical error rate $p$ falls below the pseudothreshold, then the code is guaranteed to lower the logical error rate below $p$.

\begin{figure}[!t]
\includegraphics[width = \linewidth]{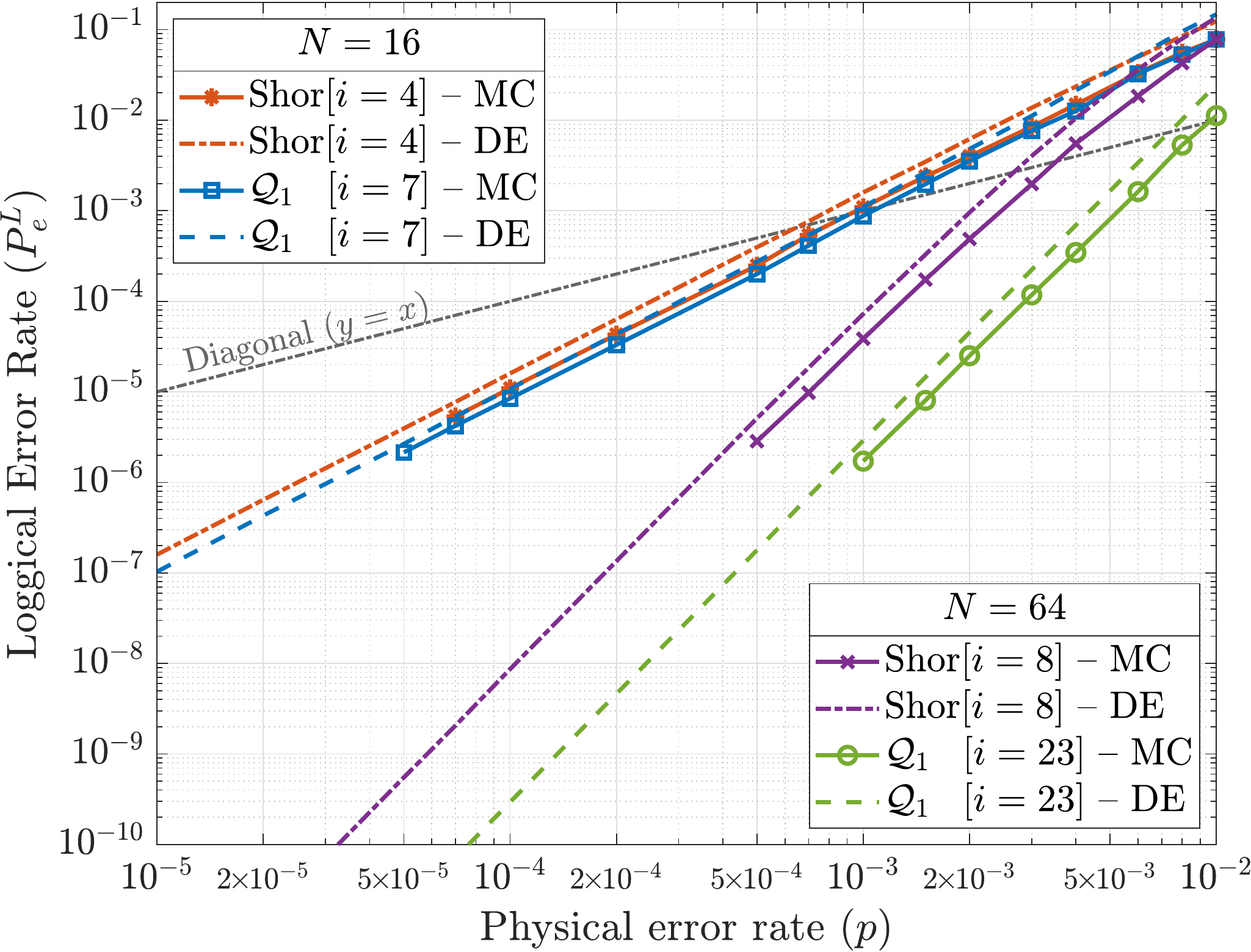}
\caption{Monte-Carlo (MC) and Density-evolution (DE)  based evaluation of the LER, for $\pone$ and Shor codes of length $N=16$ and $N=64$.}
\label{fig:p1_vs_shor_main}
\end{figure}

%\smallskip A natural extension of this work is to prepare longer $\pone$ states with high success probability. One possible approach is to replace the error detection mechanism incorporated in the preparation procedure, by an error correction one (ongoing work), allowing correcting errors on the fly, thus avoiding restarting the preparation procedure from the beginning. 

\section*{Discussion}
Polar encoded quantum computation may be seen as an error correction centric approach to FTQC.  It exploits a family of codes that have met with remarkable success, fueled by their excellent error correction performance, under practical, low complexity decoding.  In this paper, we focused on two closely related ingredients of FTQC, namely fault-tolerant code state preparation and fault-tolerant error correction. We considered polar codes encoding one logical qubit, referred to as $\pone$ codes, constituting a case of practical interest to FTQC (they support transversal logical $\cnot$ gate), and which we showed to be a meaningful generalization of the well-known Shor codes.  We may also notice that the preparation procedure presented here naturally extends to the case of quantum polar codes encoding several logical qubits, provided  they are encoded in consecutive information positions. 

\smallskip We provided here numerical estimates of the logical error rates for $\pone$ and Shor codes of length $16$ and $64$ qubits, assuming a circuit-level depolarizing noise  model. Our numerical results showed that  $\pone$ codes notably outperforms Shor codes, for the same code length and  minimum distance. To prepare longer $\pone$ code states, a possible approach is to replace the error detection mechanism incorporated in the preparation procedure, by an error correction one. Moreover, the gap between $\pone$ and Shor codes is expected to increase with the codelength, due to the $\pone$ codes construction that effectively exploits the channel polarization property. 

\smallskip Finally, we note that our preparation procedure requires distant interactions between qubits, which may be possible on some quantum technologies only. Yet, for quantum systems with local interaction constraints, it is possible to reduce the number of distant operations required to prepare small $\pone$ codes, by arranging the qubits in an appropriate manner (taking advantage of the polar code structure).  Distant operations may then be implemented through the use of swap gates, or by physically moving the qubits around. For longer codes, solutions have to be sought that may heavily depend on the specific quantum technology.
%, \emph{e.g.}, by arranging the qubits into smaller patches that can interact at a distance. 

\section*{Acknowledgment}
The authors would like to thank Davide Orsucci for useful discussions.  This work was supported by the QuantERA grant EQUIP, by the French Agence Nationale de la Recherche, ANR-22-QUA2-0005-01.

%\bibliographystyle{unsrt}
%\bibliography{biblio_database}

%\clearpage
%\setboolean{@twoside}{false}
%\includepdf[pages=-]{sup_mat.pdf}

\onecolumngrid
 
%\begin{widetext}
%long equation goes here
%\end{widetext}

\cleardoublepage

\renewcommand{\thesection}{\Roman{section}}
\renewcommand{\thesubsection}{\Roman{section}.\Alph{subsection}}

\setcounter{section}{0}
\setcounter{figure}{0}
\setcounter{table}{0}
\setcounter{equation}{0}

\onecolumngrid
\newpage

\titleformat*{\section}{\fontsize{11}{13}\selectfont\bfseries\center}
\titleformat*{\subsection}{\fontsize{11}{13}\selectfont\bfseries\center}
\titleformat*{\paragraph}{\fontsize{10.5}{12}\selectfont\itshape\bfseries}

%{hformati}{hlabeli}{hsepi}{hbeforei}[hafter i]

%\title{\vspace{-20mm} \emph{Supplemental Material for:} \\
%\vspace{2mm}
% Fault-Tolerant Preparation of Quantum Polar Codes Encoding One Logical Qubit}
%
%\author{Ashutosh Goswami$^1$, \qquad Mehdi Mhalla$^2$, \qquad Valentin Savin$^1$\\[2mm]
%    {\small $^1$\,Univ. Grenoble Alpes, CEA-L\'eti, F-38054 Grenoble, France}\\
%    {\small $^2$\,Univ. Grenoble Alpes, CNRS, Grenoble INP, LIG, F-38000 Grenoble, France}\\
%    {\small ashutosh-kumar.goswami@cea.fr, mehdi.mhalla@univ-grenoble-alpes.fr, valentin.savin@cea.fr}
%}
%\date{}

\begin{center}
 \large\textbf{\em Supplemental Material for:} 
 
 \vspace*{2mm}{\bf Fault-Tolerant Preparation of Quantum Polar Codes Encoding One Logical Qubit}
\end{center}

In Sections \ref{sec:c-pol-code} and \ref{sec:q-pole-code}, we provide introductions to classical and quantum polar codes, respectively. For quantum polar codes, we detail Steane's error correction procedure in Section \ref{subsec:steane-css}. In Section~\ref{sec:p-theorem-1}, we first provide a proof of Theorem $1$ from the main text, and then provide further details on the construction of $\pone$ codes. In Section \ref{sec:prep_lemma}, we first consider the measurement-based preparation without noise and provide the proof of Lemma $1$ from the main text.  We then consider the measurement-based preparation with noise and provide the proof of fault tolerance, according to Thoerem $3$ from the main text.  Finally, in Section \ref{sec:num_sim-prep}, we discuss numerical methods to estimate logical error rates of $\pone$ codes, using Steane error correction, incorporating our measurement based preparation. In this context, we also provide some extra numerical results.

%\tableofcontents

%\begin{abstract}
%In Sections \ref{sec:c-pol-code} and \ref{sec:q-pole-code}, we provide introductions to classical and quantum polar codes, respectively. For quantum polar codes, we detail Steane's error correction procedure in Section \ref{subsec:steane-css}. In Section~\ref{sec:p-theorem-1}, we first provide a proof of Theorem $1$ from the main text, and then provide further details on the construction of $\pone$ codes. In Section \ref{sec:prep_lemma}, we first consider the measurement-based preparation without noise and provide the proof of Lemma $1$ from the main text.  We then consider the measurement-based preparation with noise and provide the proof of fault tolerance, according to Thoerem $3$ from the main text.  Finally, in Section \ref{sec:num_sim-prep}, we discuss numerical methods to estimate logical error rates of $\pone$ codes, using Steane error correction, incorporating our measurement based preparation. In this context, we also provide some extra numerical results.
%\end{abstract}

%\maketitle

\section{Classical Polar Codes} 
\label{sec:c-pol-code}

\subsection{Encoding} The encoding of classical polar codes is done by applying the reversible $\xor$ gate recursively on an $N$ bit input $\bm{u} = (u_1, \dots, u_N) \in \{0, 1\}^N$, where $N=2^n$, with $n>0$ 
(see Fig.~\ref{fig:cpolar_recursion_and_N8}). For a set of positions $ \mathcal{F} \subseteq \{1,\dots,N\}$, the corresponding component $\bm{u}\lvert_\mathcal{F} \in \{0,1\}^{\lvert \mathcal{F} \lvert}$ of the input vector $\bm{u}$ is frozen. We may take $\bm{u}\lvert_\mathcal{F}$ to be any vector in $\{0,1\}^{\lvert \mathcal{F} \lvert}$, but it should be known to both the encoder and the decoder. The set $\mathcal{F}$ is called the \emph{frozen set}. The remaining positions $\mathcal{I} \eqdef \{1,\dots,N\}\setminus \mathcal{F}$ are used to encode information bits. The set $\mathcal{I}$ is called the \emph{information set}.

\smallskip  In the following, we denote by $\mathcal{P}(N,\mathcal{F},\bm{u}\lvert_\mathcal{F})$, the classical polar code of length $N$, frozen positions $\mathcal{F}$,  and frozen vector $ \bm{u}\lvert_\mathcal{F} \in \{0,1\}^{\lvert \mathcal{F} \lvert}$. 

\smallskip The action of the   $\xor_{2 \to 1}$ gate on input $\bm{u} = (u_1, u_2) \in \{0, 1\}^2$ gives $\bm{x} = (u_1 \oplus u_2, u_2)$. In matrix form, we may write $\bm{x} = P_2 \bm{u}$, where 
\begin{equation}
P_2 = \begin{bmatrix}
1 & 1 \\
0 & 1 \\
\end{bmatrix}
\end{equation}

\smallskip  The classical polar transform, that is, the recursive application of $\xor_{2 \to 1}$ on $N=2^n$ qubits, is thus given by the matrix 
\begin{equation} \label{eq:rec_pol}
P_N=P_2^{\otimes n}.
\end{equation}
For any value of the information bits $\bm{u}_\mathcal{I} \in \{0,1\}^{\lvert \mathcal{I} \lvert}$, $\bm{x} = P_N (\bm{u}\lvert_\mathcal{F},  \bm{u}\lvert_\mathcal{I} ) \in \{0,1\}^N$ is a codeword of the polar code $\mathcal{P}(N,F,\bm{u}\lvert_\mathcal{F})$. If $\bm{u}\lvert_\mathcal{F}$ is the all zero vector, the polar code $\mathcal{P}(N,\mathcal{F},\bm{u}\lvert_\mathcal{F})$ is generated by the columns of $P_N$ corresponding to the information set $\mathcal{I}$.

\begin{figure}[!b]
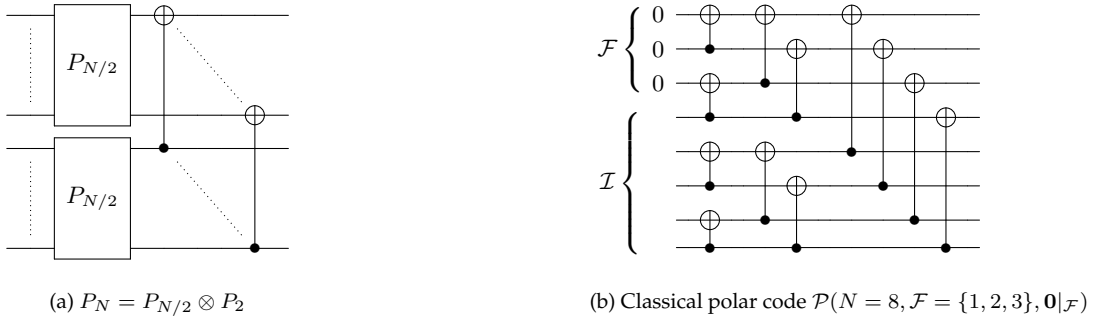

\begin{subfigure}[b]{.45\textwidth}
\,\hfill\input{cpolar_recursion}\hfill\,
\captionsetup{justification=centering}
\caption{$P_N = P_{N/2}\otimes P_2$}
\label{fig:cpolar_recursion}
\end{subfigure}\hfill%
\begin{subfigure}[b]{.55\textwidth}
\,\hfill\input{cpolar_N8}\hfill\hspace*{3mm}\,

\vspace*{1.8ex}
\captionsetup{justification=centering}
\caption{Classical polar code $\mathcal{P}(N=8,\mathcal{F}= \{1,2,3\},\bm{0}\lvert_\mathcal{F})$}
\label{fig:cpolar_N8}
\end{subfigure}
\caption{$(a)$ Polar transform recursion: $P_N$ in terms of $P_{N/2}$. $(b)$ Example of a classical polar code encoding $|\mathcal{I}| = 5$ bits into $N = 2^3$ bits, with frozen set $\mathcal{F}= \{1,2,3\}$, and frozen vector $u_\mathcal{F}=(0,0,0)$. Here, the set $\mathcal{F}$ (thus, $\mathcal{I}$)  is chosen only for the purpose and the simplicity of the illustration. In general, it needs not consist of consecutive positions. }
\label{fig:cpolar_recursion_and_N8}
\end{figure}

\subsection{Construction} 
The \emph{construction} of a classical polar code refers to the choice of the information set $\mathcal{I}$ (or, equivalently, the frozen set $\mathcal{F}$). This is done in a \emph{channel specific way} as follows. 

\smallskip Consider a discrete, memoryless, classical channel $W(y \mid x)$, with binary input $x\in \{0,1\}$, and output $y \in \cl{Y}$.  For $\bm{u}\in\{0,1\}^N$, let $\bm{x} = P_N(\bm{u}) \in\{0,1\}^N$ (note that here $\mathcal{F}$ and $\mathcal{I}$ need not be defined) and $\bm{y} \in \cl{Y}^N$ be the output corresponding to $N$ uses of the $W$ channel, with inputs $x_1,\dots,x_N$. For $i=1,\dots,N$, let $W^{(i)}\left(\bm{y}, \bm{u}_{1:i-1} \mid u_i\right)$ denote the so-called \emph{virtual channel}, with  input $u_i \in \{0,1\}$ and output $(\bm{y}, \bm{u}_{1:i-1}) \in \cl{Y} \times \{0,1\}^{i-1}$, where $\bm{u}_{1:i-1} := (u_1,\dots,u_{i-1})$.

\smallskip Informally, the \emph{channel polarization theorem}~\cite{arikan2009channel} states that, for sufficiently large $N$, almost all the virtual channels become arbitrarily close to either the noiseless (perfect) channel or to the completely noisy (useless) channel. The \emph{closeness} to the either perfect or useless channel can be expressed in terms of different parameters, such as the mutual information $I\left(W^{(i)}\right)$, the Bhattacharyya parameter $Z\left(W^{(i)}\right)$, or the error probability $P_e\left(W^{(i)}\right)$, which can be computed analytically for some channels (\emph{e.g.}, binary erasure channels), or estimated numerically through \emph{density evolution}, for more general channels~\cite{tal2013construct}. Once one of these parameters is computed for all the virtual channels, they are sorted from the most reliable (closest to the perfect channel) to the least reliable (closest to the useless channel) one.

%\smallskip For sufficiently large $N$, almost all the virtual channels become arbitrarily close to either the noiseless (perfect) channel or close to the completely noisy (useless) channel. This phenomenon is called {\bf channel polarization} \cite{arikan09}. The set $\mathcal{F}$ is selected based on the channel polarization as follows. The synthesized virtual channels are first ordered from the best to worst in terms of their reliability, using for example their Bhattacharyya parameter \cite{arikan09}. Let $O$ be the corresponding ordered set of indices.

\smallskip  For a polar code encoding $K$ information bits, the information set $\cl{I}$ consists of the indexes corresponding to the $K$ most reliable virtual channels (equivalently, the  $N-K$ least reliable virtual channels are frozen).  The usefulness of this construction will become apparent in relation to the successive cancellation decoding, discussed in the next section. 
%For the moment, we simply mention here that Polar codes are known to \emph{achieve the channel capacity}. This means that for sufficiently large $N$, it is possible to choose an information set of size $|\cl{I}|$ arbitrary close of $N I(W)$, where $I(W)$ is the channel's mutual information, while ensuring an arbitrary small error probability under  successive cancellation decoding. 

\subsection{Decoding} 
\makeatletter
\def\@currentlabel{\thesubsection}
\makeatother
\label{subsec:class_dec}

Classical polar codes come equipped with an efficient successive cancellation (SC) decoding algorithm. SC decoding takes advantage of the polar code construction, by estimating inputs $u_1,\dots,u_N$ sequentially. For $i=1,\dots,N$, SC decoding outputs the \emph{maximum a posteriori estimate} $\hat{u}_i$ of $u_i$, conditional on the observed output $\bm{y}$ and previous estimates $\hat{\bm{u}}_{1:i-1} = (\hat{u}_1,\dots,\hat{u}_{i-1})$. Precisely, we have

%\footnote{Here, we use the same notation for a channel $W^{(i)}$ and the corresponding conditional probabilities between channel output and input.}

\begin{equation}
\hat{u}_i := \left\{ \begin{array}{ll}
u_i, & \mbox{if } i\in\cl{F} \\
\displaystyle\argmax_{u\in\{0,1\}} W^{(i)}(\bm{y}, \hat{\bm{u}}_{1:i-1} \mid u), & \mbox{if } i\in\cl{I}
\end{array}\right.
\end{equation}
If the information set contains indexes of \emph{good virtual channels} (close to the perfect channel), the maximum a posteriori estimate $\hat{u}_i$ is equal to the input $u_i$ with high probability. This does not happen for \emph{bad virtual channels} (close to the useless channel), which must then be frozen, so that the corresponding inputs are known to both the encoder and decoder. It is worth  noticing that the maximum a posteriori decoding of the virtual channels can be performed in an efficient way, by using a message passing algorithm that takes advantage of the recursive structure of the polar code~\cite{arikan2009channel}. Overall, the complexity of the SC decoding scales as $O(N\log(N))$. 

\smallskip Finally, we note that Polar codes under SC decoding are known to \emph{achieve the channel capacity}~\cite{arikan2009channel}. This means that for sufficiently large $N$, it is possible to choose an information set of size $|\cl{I}|$ arbitrary close of $N I(W)$, where $I(W)$ is the channel's mutual information, while ensuring an arbitrary small error probability under SC decoding.

\section{Quantum Polar Codes} 
\label{sec:q-pole-code}

\subsection{Encoding}
The encoding of quantum information for CSS quantum polar codes is briefly given in the main text under the paragraph ``CSS Quantum Polar Codes". In this section, we provide more details. 

\smallskip Recall from the main text that the quantum polar transform $Q_N$ corresponds to the recursive action of the quantum $\cnot$ gate on an $N$-qubit quantum state. The quantum $\cnot$ gate induces the reversible $\xor$ gate in the Pauli $Z$ and Pauli $X$ bases. Precisely, $\cnot_{2 \to 1}$ acts as $\xor_{2 \to 1}$ in the Pauli $Z$ basis,  while it acts as $\xor_{1 \to 2}$ in the Pauli $X$ basis. Hence,  it follows that $Q_N$  acts as the classical polar transform $P_N$ in the Pauli $Z$ basis, while it acts as the reverse classical polar transform (\emph{i.e.}, with inverted target and control bits) in the Pauli $X$ basis.

\smallskip  It can be seen that the reversed classical polar transform is described  by $P_N^\top$, where $P_N^\top$ is the transpose of $P_N$ \cite{renes2011efficient}. Therefore, we have the following, for $\bm{u} \in \{0, 1\}^N$,
\begin{align}
Q_N \ket{\bm{u}} = \ket{P_N\bm{u}}, \label{eq:q-pol-enc-x} \\
Q_N \ket{\bm{\oline{u}}} = \ket{\oline{P_N^\top \bm{u}}}. \label{eq:q-pol-enc-z} 
\end{align}

%\smallskip The encoding of quantum information for CSS quantum polar codes is given in the main text paragraph ``CSS Quantum Polar Codes".

%\smallskip For CSS quantum polar codes, the quantum information is encoded as follows (Fig.~\ref{fig:qpolar_N8}). For a subset of positions $\mathcal{Z} \subseteq \mathcal{S}$, the input quantum state is frozen to a known Pauli $Z$ basis  state $\ket{\bm{u}}_\mathcal{Z}$, where $\bm{u} \in \{0, 1\}^{|\mathcal{Z}|}$, and for another subset $\mathcal{X} \subseteq \mathcal{S}$, with $\mathcal{Z} \cap \mathcal{X} = \emptyset$, it is frozen to a known Pauli $X$ basis state $\ket{\bm{\oline{v}}}_\mathcal{X} $, where $\bm{v} \in \{0, 1\}^{|\mathcal{X}|}$. 
%
%\smallskip  The remaining subset $\mathcal{I} := \mathcal{S}\setminus (\mathcal{X} \cup \mathcal{Z})$ is used to encode a quantum state $\ket{\phi}_\cl{I}$.  Hence, the quantum state $\ket{\bm{u}}_\mathcal{Z} \otimes \ket{\phi}_\mathcal{I} \otimes \ket{\bm{\oline{v}}}_\mathcal{X}$ is given as input to the polar transform $Q_N$. The encoded logical code state, denoted by  $\ket{\widetilde{\phi}}_\mathcal{S}$, is given by 
%%
%\begin{equation}
%\ket{\widetilde{\phi}}_\mathcal{S} = Q_N (\ket{\bm{u}}_\mathcal{Z} \otimes \ket{\phi}_\mathcal{I}   \otimes \ket{\bm{\oline{v}}}_\mathcal{X}).  \label{eq:in-qpol}
%\end{equation}

 In the following, we denote by $\qp(N,\mathcal{Z},\mathcal{X},\ket{\bm{u}}_\mathcal{Z} ,\ket{\bm{\oline{v}}}_\mathcal{X})$, the quantum polar code on $N$ qubits, with frozen sets $\mathcal{Z}$ and $\mathcal{X}$, corresponding to the Pauli $Z$ and Pauli $X$ bases, respectively, and  with  corresponding frozen quantum states $\ket{\bm{u}}_\mathcal{Z}$ and $\ket{\bm{\oline{v}}}_\mathcal{X}$. It induces two classical polar codes, one in $Z$ basis, with frozen set $\cl{Z}$, and one in $X$ basis, with frozen set $\cl{X}$, where the latter is defined by a reversed polar transform, as illustrated in Fig.~\ref{fig:qpolar_N8_Zbasis_Xbasis}. 
 
 %Note that the reversed classical polar transform is described  by $P_N^\top$, where $P_N^\top$ is the transpose of $P_N$ \cite{renes2011efficient}. 

\smallskip Let $\pi$ denote the \emph{reverse order permutation} of $\cl{S} = \{1,\dots,N\}$, defined by $\pi(i) = N+1-i$. Then, the classical polar code induced in the $Z$ basis is $\mathcal{P}(N,\cl{Z}, \bm{u})$, while the classical polar code induced in the $X$ basis is $\mathcal{P}(N, \pi(\cl{X}), \pi(\bm{v}))$, where the vector $\pi(\bm{v})$ is defined by permuting the entries of $\bm{v}$ according to $\pi$, that is, $\pi(\bm{v})_i := v_{\pi(i)}$, $\forall i\in \cl{X}$.

\begin{figure}[!b]
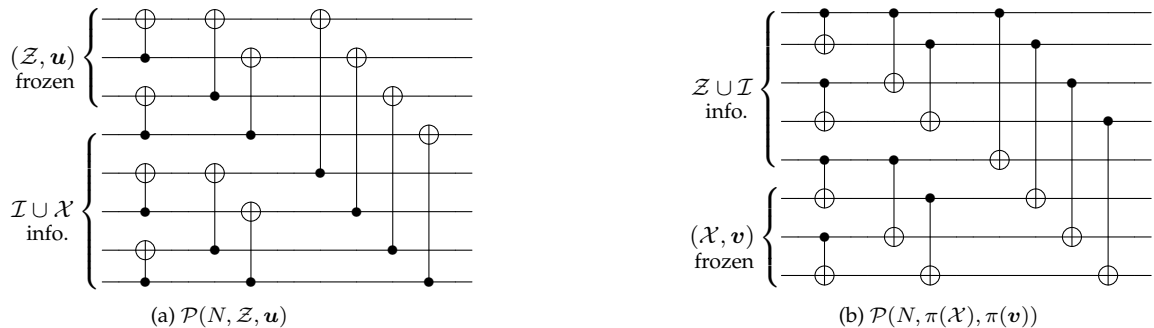

\centering
\begin{subfigure}[b]{.48\textwidth}
\captionsetup{justification=centering}
\,\hfill\hfill\hspace*{8mm}\input{qpolar_N8_Zbasis}\hfill\,
\caption{$\mathcal{P}(N,\cl{Z}, \bm{u})$}
\label{fig:qpolar_N8_Zbasis}
\end{subfigure}\hfill%
\begin{subfigure}[b]{.48\textwidth}
\captionsetup{justification=centering}
\,\hfill\hfill\hspace*{-8mm}\input{qpolar_N8_Xbasis}\hfill\,
\caption{$\mathcal{P}(N, \pi(\cl{X}), \pi(\bm{v}))$}

\label{fig:qpolar_N8_Xbasis}
\end{subfigure}
\caption{Classical polar codes in (a) $Z$ basis, and (b) $X$ basis, induced by the CSS quantum polar code $\qp(N,\mathcal{Z},\mathcal{X},\ket{\bm{u}}_\mathcal{Z} ,\ket{\bm{\oline{v}}}_\mathcal{X})$. The permutation $\pi$ in (b) is the reverse order permutation. (Note that in this example the two codes are actually the same.)}
\label{fig:qpolar_N8_Zbasis_Xbasis}
\end{figure}

\smallskip The $X$ and $Z$ type stabilizer generators of $\qp(N, \mathcal{Z},\mathcal{X},\ket{\bm{u}}_\mathcal{Z},\ket{\oline{\bm{v}}}_\mathcal{X})$ are given by the  lemma below.

\begin{lemma} \label{lem:stab-gen}
Let $\ee_i \in \{0,1\}^N$ be the binary vector, with $0$ everywhere except $1$ at the $i$-th position. Then, the stabilizer group of $\qp(N, \mathcal{Z},\mathcal{X},\ket{\bm{u}}_\mathcal{Z},\ket{\oline{\bm{v}}}_\mathcal{X})$ is generated by the following $X$ and $Z$ type operators
\begin{align}
(-1)^{\bm{v}_i}  X^{P_N \eee_i} , \forall i \in \mathcal{X}. \label{eq:gen-x-1} \\
(-1)^{\bm{u}_i}  Z^{P_N^\top \eee_i} , \forall i \in \mathcal{Z}. \label{eq:gen-z-1}
\end{align}

\end{lemma}

\begin{proof}
Note that the input of the polar transform, that is, the quantum state $\ket{\bm{u}}_\mathcal{Z} \otimes \ket{\phi}_\mathcal{I}   \otimes \ket{\bm{\oline{v}}}_\mathcal{X}$, is stabilized by the following Pauli operators
\begin{align}
(-1)^{v_i} X^{\eee_i}, \forall i \in \mathcal{X},  \\
(-1)^{u_i} Z^{\eee_i}, \forall i \in \mathcal{Z}. 
\end{align}
Therefore, the encoded quantum state $\ket{\widetilde{\phi}}_\mathcal{S} = Q_N(\ket{\bm{u}}_\mathcal{Z} \otimes \ket{\phi}_\mathcal{I}   \otimes \ket{\bm{\oline{v}}}_\mathcal{X})$ is stabilized by the following Pauli operators,
\begin{align}
(-1)^{\bm{v}_i} Q_N X^{\eee_i} Q_N, \forall i \in \mathcal{X},\label{eq:gen-x} \\
(-1)^{\bm{u}_i} Q_N Z^{\eee_i} Q_N, \forall i \in \mathcal{Z}. \label{eq:gen-z}
\end{align}
The operators in~(\ref{eq:gen-x})-(\ref{eq:gen-z}) are a generating set of the stabilizer group of  $\qp(N,\mathcal{Z},\mathcal{X},\ket{\bm{u}}_\mathcal{Z} ,\ket{\bm{\oline{v}}}_\mathcal{X})$. These generators can be written in terms of classical polar transforms $P_N$ and $P_N^\top$ as follows. 

\smallskip \noindent We first consider the case of two qubits, and observe that the sandwiching actions of the $\cnot$ gate on Pauli operators gives, for $\bm{u} \in \{0,1\}^2$,
\begin{align}
\cnot_{2 \to 1} X^{\bm{u}} \cnot_{2 \to 1} = X^{P_2 \bm{u}},  \label{eq:cnot_sand-x} \\
\cnot_{2 \to 1} Z^{\bm{u}} \cnot_{2 \to 1} = Z^{P_2^\top \bm{u}}. \label{eq:cnot_sand-z}
\end{align}
Using~(\ref{eq:cnot_sand-x})-(\ref{eq:cnot_sand-z}) and (\ref{eq:rec_pol}), the sandwiching action of $Q_N$ is described by the classical polar transforms $P_N$ and $P_N^\top$, respectively, on $X$ and $Z$ type operators. Hence, the stabilizer generators from~(\ref{eq:gen-x})-(\ref{eq:gen-z}) can be written as given in~(\ref{eq:gen-x-1})-(\ref{eq:gen-z-1})
\end{proof}

\smallskip Note, from (\ref{eq:gen-x-1}), that the indicator vector of $X$-type generators are given by the columns of $P_N$, corresponding to the set $\mathcal{X}$. Further, from (\ref{eq:gen-z-1}), the indicator vectors of the $Z$-type generators are given by the columns of $P_N^\top$ ($i.e$, the rows of $P_N$), corresponding to the set $\mathcal{Z}$.

\smallskip Similarly to Lemma \ref{lem:stab-gen}, logical Pauli operators can be determined by passing $X$ and $Z$ operators through the polar transform, as shown in the lemma below.

\begin{lemma}
 Let  $\widetilde{X}_i$  and $\widetilde{Z}_i$ be the logical $X$ and $Z$ operators, corresponding to the encoded qubit at the position $i \in \mathcal{I}$. Then, 
\begin{align}
\widetilde{X}_i =   X^{P_N \eee_i}, \label{eq:log-x} \\
\widetilde{Z}_i = Z^{P_N^\top \eee_i}. \label{eq:log-z} 
\end{align}
\end{lemma}

\subsection{Construction} 
\makeatletter
\def\@currentlabel{\thesubsection}
\makeatother
\label{sec:construction-css}

The \emph{construction} of a CSS quantum polar code refers to determining the frozen sets $\mathcal{Z}$ and $\mathcal{X}$ (thus, the information set $\mathcal{I}$), which exploits the classical polarization in $Z$ and $X$ bases, respectively~\cite{renes2011efficient}. We shall assume a Pauli channel $\cl{W}$, with qubit input, given by 
\begin{equation}
\cl{W}(\rho) = p_I\rho + p_X X\rho X + p_Y Y\rho Y + p_Z Z\rho Z, \label{eq:pauli_chanel}
\end{equation} 
where $\rho$ denotes the density matrix of the qubit (mixed) state, and  $p_I, p_X, p_Y, p_Z \in [0,1]$ are probability values, summing to $1$.

\smallskip In $Z$ basis, only $X$ errors matter. Thus, the $Z$-basis induced channel, denoted $W_Z$, captures the effect of $X$-type errors on the quantum state, and is given by 
\begin{equation}
W_Z(\ket{u}\bra{u}) = (p_I+p_Z)\ket{u}\bra{u} + (p_X + p_Y) X\ket{u}\bra{u} X, \text{ for } u\in\{0,1\}. \label{eq:pauli_chanel_Zbasis}
\end{equation} 
Hence, $W_Z$ is a classical binary symmetric channel (BSC), with error probability $p_X + p_Y$. Similarly, the $X$-basis induced channel, denoted $W_X$, captures the effect of $Z$-type errors on the quantum state, and is given by 
\begin{equation}
W_X(\ket{\overline{u}}\bra{\overline{u}}) = (p_I+p_X)\ket{\overline{u}}\bra{\overline{u}} + (p_Z + p_Y) Z\ket{\overline{u}}\bra{\overline{u}} Z, \text{ for } u \in \{0,1\}. \label{eq:pauli_chanel_Xbasis}
\end{equation} 
Hence, $W_X$ is a classical BSC, with error probability $p_X + p_Y$. 
% If the CSS quantum polar code is constructed to decode $W_Z$ and $W_X$ channels, correlations between $X$ and $Z$ errors are ignored. 

\smallskip To construct the CSS quantum polar code,  we exploit the classical polarization of $W_Z$ and $W_{X}$ channels (See Section~\ref{subsec:class_dec}). Note that this construction \emph{ignores the correlations} between $X$ and $Z$ (we will explain how these correlations can be captured a little later, below).
The frozen set $\mathcal{Z}$ is determined by the classical polarization of the  $W_Z$ channel, under the classical polar transform $P_N$,  while the frozen set $\mathcal{X}$ is determined by the classical polarization of the  $W_{X}$, under the reversed classical polar transform $P_N^\top$. The remaining information set $\cl{I}$ corresponds to virtual channels that are \emph{good}  in both $Z$ and $X$ bases. The frozen set $\mathcal{Z}$ corresponds to virtual channels that are \emph{bad} in $Z$ basis. These channels may be either bad or good in $X$ basis, it does not matter, since we need not decode $Z$ errors on corresponding inputs (such errors correspond to $Z$-type generators, thus they act trivially on the code space). Similarly, the frozen set $\mathcal{X}$ corresponds to virtual channels that are bad in $X$ basis, and which may be either bad or good in $Z$ basis. Precisely, positions in $\cl{X}$ that do not impact the decoding of subsequent positions in $\cl{I}$ (if any) may correspond to bad virtual channels in $Z$ basis. For instance, this may be the case for some  positions in $\cl{X}$ that come after the last position in $\cl{I}$, \emph{i.e.}, some  $j\in \cl{X}$, such that $i < j$, $\forall i\in \mathcal{I}$. See also Fig.~\ref{fig:qpolar_N8_Zbasis}.

\smallskip To \emph{capture the correlations} between $X$ and $Z$ errors, one of the two channels, say $W_X$, has to be \emph{extended}~\cite{renes2011efficient}. The extended channel, denoted by $W_{X'}$, is thus defined by the conditional probability of a $Z$ error, given the $X$ error. Precisely,
 \begin{align}
W_{X'}(\ket{\overline{u}}\bra{\overline{u}}) = 
   &\ (p_I + p_Z) \ket{0}\bra{0}_X \otimes 
   \left( \frac{p_I}{p_I + p_Z}\ket{\overline{u}}\bra{\overline{u}} + \frac{p_Z}{p_I + p_Z} Z\ket{\overline{u}}\bra{\overline{u}} Z \right) + \nonumber\\
   &\ (p_X + p_Y) \ket{1}\bra{1}_X \otimes 
   \left( \frac{p_X}{p_X + p_Y}\ket{\overline{u}}\bra{\overline{u}} + \frac{p_Y}{p_X + p_Y} Z\ket{\overline{u}}\bra{\overline{u}} Z \right),
\label{eq:pauli_chanel_Xbasis_ext}
\end{align} 
where $\{\ket{0}_X, \ket{1}_X\}$ is an orthogonal basis of an auxiliary system, indicating whether an $X$ error happened or not. Put differently, $W_{X'}$ is a classical mixture of two BSCs, the first with error probability $p_Z/(p_I+p_Z)$ (when no $X$ error happened), and the second with error probability $p_Y/(p_X+p_Y)$ (when an $X$ error happened). 
%The first BSC happens with probability $p_I+p_Z$, while the second one with probability $p_X+p_Y$. 

\smallskip The construction of the CSS quantum polar code, taking into account the correlations between $X$ and $Z$ errors, is done in the same way as above, while replacing the $W_X$ channel by its extended version $W_{X'}$.

\subsection{Steane Error Correction} 
\makeatletter
\def\@currentlabel{\thesubsection}
\makeatother
\label{subsec:steane-css}

We describe here the Steane  error correction procedure~\cite{steane1997active, steane2002fast} (see also~\cite[Section~4.4]{gottesman2010introduction}), applied to CSS quantum polar codes.  Throughout this section, we consider an encoded state $\ket{\widetilde{\phi}}_\mathcal{S} = Q_N\left( \ket{\bm{u}}_{\cl{Z}} \otimes \ket{\phi}_{\cl{I}} \otimes \ket{\overline{\bm{v}}}_{\cl{X}} \right)$ of the quantum polar code $\qp(N,\mathcal{Z},\mathcal{X},\ket{\bm{u}}_\mathcal{Z} ,\ket{\bm{\oline{v}}}_\mathcal{X})$, that we want to protect against Pauli errors.

\smallskip Steane's error correction procedure consists of the following steps.
\begin{itemize}
\item[$(1)$] An ancilla system $\mathcal{S}'$ is prepared in either the logical all-$\ket{+}$ state or the logical all-$\ket{0}$ state of $\qp(N,\mathcal{Z},\mathcal{X},\ket{\bm{u}}_\mathcal{Z} ,\ket{\bm{\oline{v}}}_\mathcal{X})$. The former state is prepared for $X$-error correction, while the latter of $Z$-error correction. 

\item[$(2)$] A transverse $\cnot$ gate is applied between $\mathcal{S}$ (original)  and  $\mathcal{S}'$ (ancilla) systems, in such a way that either the $X$ or the $Z$ errors  on $\mathcal{S}$ are copied to $\mathcal{S}'$, while the two systems remain separated. 

\item[$(3)$] The ancilla system is measured, outputting a random codeword of a classical polar code (in either $Z$ or $X$ basis), corrupted by the error that has been copied to $\mathcal{S}'$.

\item[$(4)$] A classical SC decoding is applied to determine the error (possibly, the corresponding corrective operation is applied on the $\mathcal{S}$ system).
\end{itemize}

\smallskip It is easily seen that steps $(2)$ and $(3)$ are fault-tolerant, as they consist only of transverse gates and single qubit measurements. 

%The fault-tolerant preparation of the ancilla system $\mathcal{S}'$ (step $(1)$) will make the object of Section \ref{sec:prep-m} (for the particular case of $\pone$ codes, introduced in Section \ref{sec:q1code}).

\smallskip By a slight abuse of language, we shall refer to steps $(1)$-$(3)$ above as $\emph{syndrome extraction}$ \footnote{Strictly speaking, the classical information obtained at step $(3)$ is not an error syndrome, but an error corrupted version of a random codeword (of course, if needed, one may classically compute the corresponding error syndrome).}.  In case of \emph{ideal syndrome extraction}, steps  $(1)$-$(3)$  are assumed to be error free, and the error corrected at step $(4)$ is the one preexisting on system $\mathcal{S}$, before the syndrome is extracted.
%, that is, they do not generate additional errors, neither on $\mathcal{S}$ nor on $\mathcal{S}'$. 
In case of \emph{noisy syndrome extraction}, steps  $(1)$-$(3)$ may generate additional errors on systems $\mathcal{S}$ and $\mathcal{S}'$. We will analyze the impact of these errors, together with providing the details of Steane's error correction applied to either $X$ or $Z$ errors, on the two subsections below.

\smallskip We will consider an ancilla system  $\mathcal{S}' =  \{1', \dots, N'\}$ and subsets $\mathcal{Z}', \mathcal{I}', \mathcal{X}' \subseteq \mathcal{S}'$. They are the counterparts of the subsets  $\mathcal{Z}, \mathcal{I}, \mathcal{X} \subseteq \mathcal{S}$, in the sense that $i' \in \mathcal{Z}' \Leftrightarrow i \in \mathcal{Z}$, and similarly for  $\mathcal{I}'$ and $\mathcal{X}'$.

\paragraph*{$X$-Error Correction.} 

To extract the syndrome for $X$ errors, the ancilla system $\mathcal{S}' =   \mathcal{Z}'\cup \mathcal{I}' \cup \mathcal{X}'$ must be  prepared in a logical $X$ basis state. Usually, the ancilla state is taken to be the logical all-$\ket{+}$ state, obtained by encoding the all-$\ket{+}$ state on system $\mathcal{I}'$. Here, we consider a slightly more general logical $X$ basis state, as follows,
\begin{equation}
\ket{\widetilde {\oline{\bm{w'}}}}_{\mathcal{S}'} = Q_N ( \ket{\bm{u'}}_{\mathcal{Z}'} \otimes \ket{\oline{\bm{w'}}}_{\mathcal{I}'} \otimes \ket{\oline{\bm{v'}}}_{\mathcal{X}'}),
\end{equation}
 where $ \bm{u'}, \bm{w'},\bm{v'}$ are known. Note that frozen values $\bm{u'}$ and $\bm{v'}$  may be different from frozen values $\bm{u}$ and $\bm{v}$. The reason we consider the above logical state is that our preparation procedure for logical polar code states is measurement based, for which  $ \bm{u'}, \bm{w'},\bm{v'}$ are determined based on random outcomes of measurements therein.

\begin{figure}[!t]
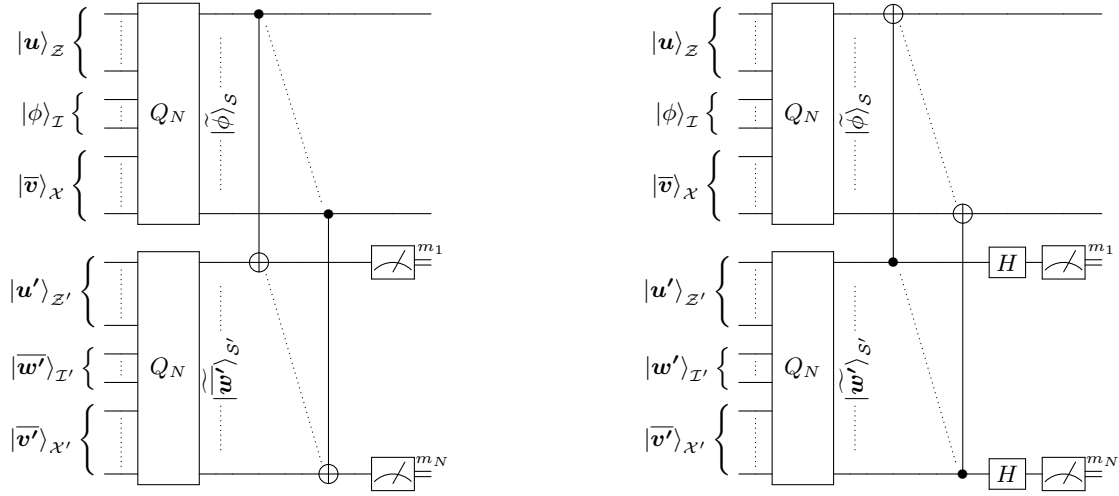

\begin{subfigure}{.48\linewidth}
\,\hfill \,\hfill\hspace*{6mm}\input{steane_Xerror}\hfill\,
\caption{$X$-error correction: $X$ errors are copied from the original  (top) to the ancilla (bottom) system, which is measured in $Z$ basis.}
\label{fig:steane_Xerror}
\end{subfigure}%
\hfill%
\begin{subfigure}{.48\linewidth}
\,\hfill\input{steane_Zerror}\hfill\,
\caption{$Z$-error correction: $Z$ errors are copied from the original  (top) to the ancilla (bottom) system, which is  measured in $X$ basis.}
\label{fig:steane_Zerror}
\end{subfigure}
\captionsetup{justification=centering}
\caption{Steane's error correction procedure.}
\label{fig:steane_error_correction}
\end{figure}

\smallskip Steane's $X$-error syndrome extraction procedure is depicted in Fig.~\ref{fig:steane_Xerror}. After preparing the state $\ket{\widetilde {\oline{\bm{w'}}}}_{\mathcal{S}'} $ on system ${\mathcal{S}'}$, the transversal $\cnot_{\mathcal{S} \to \mathcal{S}'}$ gate is applied on  corresponding qubits of systems $\mathcal{S}$ and  $\mathcal{S}'$. Then, each qubit in the ancilla system $\mathcal{S}'$ is measured in the Pauli $Z$ basis. We  denote by $\bm{m} = (m_1,\dots,m_N)$ the classical outputs of these measurements. 

\smallskip The following lemma gives the state of the $\mathcal{S}$ system, as well as the measurement result $\bm{m}$, after the syndrome extraction procedure. We consider errors $\bm{e}_X$ and $\bm{e'}_X$ that have happened on systems $\mathcal{S}$ and $\mathcal{S}'$, respectively. We  assume  that errors $\bm{e}_X$ and $\bm{e'}_X$  have happened before the transversal $\cnot_{\mathcal{S} \to \mathcal{S}'}$ gate is applied, while the transversal $\cnot_{\mathcal{S} \to \mathcal{S}'}$ gate, as well as the measurement operations, are error free.  We refer to these errors as \emph{preparation errors}. This assumption is made for simplicity only. Indeed, it is not too difficult to see that errors generated on the $\mathcal{S}'$ system, by either the transversal $\cnot_{\mathcal{S} \to \mathcal{S}'}$ gate or the measurement operations, can actually be incorporated to the error $\bm{e'}_X$. Errors generated by the transversal $\cnot_{\mathcal{S} \to \mathcal{S}'}$ gate on the $\mathcal{S}$ system go undetected, but they may be corrected during the next round of error correction.

\begin{lemma} \label{lem:steane}
Let $\bm{e}_X$ and $\bm{e'}_X$ be preparation errors that have happened on systems $\mathcal{S}$ and $\mathcal{S}'$, respectively. Then, after the Steane's $X$-error syndrome extraction, the state of the system $\mathcal{S}$ is given by
\begin{equation} \label{eq:steane_state}
X^{\bm{e}_X} \ket{\widetilde{Z^{\bm{w'}}\phi}}_\mathcal{S} =  X^{\bm{e}_X} Q_N (\ket{\bm{u}}_{\mathcal{Z}} \otimes Z^{\bm{w'}}\ket{\phi}_{\mathcal{I}} \otimes  \ket{\bm{\oline{v \oplus v'}}}_{\mathcal{X}}).
\end{equation}
Further, the measurement outcome is a noisy codeword of the classical polar code $\mathcal{P}(N, \cl{Z}, \bm{u} \oplus \bm{u'})$, with frozen set $\cl{Z}$ and frozen vector $\bm{u} \oplus \bm{u'}$, as follows,
\begin{equation} \label{eq:steane-syn}
\bm{m} = P_N( \bm{u} \oplus \bm{u'}, \bm{a'}, \bm{x'}) \oplus \bm{e}_X  \oplus \bm{e'}_X \in \{0,1\}^N, 
\end{equation}
where $\bm{a'} \in \{0, 1\}^{|\mathcal{I}|}$ and $\bm{x'} \in \{0, 1\}^{|\mathcal{X}|}$ are random vectors. 
\end{lemma}

\begin{proof}
Let $\displaystyle \ket{\phi}_{\cl{I}} = \sum_{\bm{a} \in \{0, 1\}^{|\mathcal{I}|}} \phi_{\bm{a}} \ket{\bm{a}}_{\cl{I}}$. Note also that~$\ket{\overline{\bm{v}}}_{\cl{X}} = \sum_{\bm{x} \in \{0, 1\}^{|\mathcal{X}|}} (-1)^{\bm{v}\cdot \bm{x}} \ket{\bm{x}}_{\cl{X}}$ (up to a normalization factor, which will be omitted in the sequel). Then, the noisy logical state $X^{\bm{e}_X} \ket{\widetilde{\phi}}_\mathcal{S}$ can be written as
\begin{align}
X^{\bm{e}_X} \ket{\widetilde{\phi}}_\mathcal{S} &= X^{\bm{e}_X}  Q_N\left( \ket{\bm{u}}_{\cl{Z}} \otimes \ket{\phi}_{\cl{I}} \otimes \ket{\overline{\bm{v}}}_{\cl{X}} \right) \\
 &= \sum_{\bm{a}, \bm{x}} \phi_{\bm{a}} (-1)^{\bm{v}\cdot \bm{x}} \ket{P_N(\bm{u},\bm{a},\bm{x}) \oplus \bm{e}_X }.
\end{align}
Similarly, the noisy ancilla state can be written as
\begin{align}
X^{\bm{e'}_X}\ket{\widetilde {\oline{\bm{w'}}}}_{\mathcal{S}'} &= X^{\bm{e'}_X}Q_N (  \ket{\bm{u'}}_{\mathcal{Z}'} \otimes \ket{\oline{\bm{w'}}}_{\mathcal{I}'}  \otimes \ket{\oline{\bm{v'}}}_{\mathcal{X}'}) \\
& = \sum_{\bm{a'}, \bm{x'}} (-1)^{\bm{w'}\cdot\bm{a'}  + \bm{v'}\cdot \bm{x'}} \ket{P_N(\bm{u'}, \bm{a'}, \bm{x'}) \oplus \bm{e}'_X }.
\end{align}
After the transverse  $\cnot_{\mathcal{S} \to \mathcal{S}'}$ gate is applied, we get the following state on the bipartite $\mathcal{S} \mathcal{S}'$ system,
\begin{align}
\ket{\theta}_{\mathcal{S} \mathcal{S}'} &:= \cnot_{\mathcal{S} \to \mathcal{S}'} 
\left( X^{\bm{e}_X} \ket{\widetilde{\phi}}_\mathcal{S} \otimes  X^{\bm{e'}_X}\ket{\widetilde {\oline{\bm{w'}}}}_{\mathcal{S}'} \right)\\
&\hspace*{1mm}= \sum_{\bm{a}, \bm{x}} \phi_{\bm{a}} (-1)^{\bm{v}\cdot \bm{x}} \ket{P_N(\bm{u}, \bm{a},  \bm{x}) \oplus \bm{e}_X} \otimes \sum_{\bm{a'}, \bm{x'}} (-1)^{\bm{w'}\cdot\bm{a'}  + \bm{v'}\cdot \bm{x'}}   
\ket{P_N(\bm{u}\oplus\bm{u'}, \bm{a}\oplus \bm{a'}, \bm{x}\oplus \bm{x'}) \oplus \bm{e}_X \oplus \bm{e'}_X }  \\
&\hspace*{1mm}= \sum_{\bm{a}, \bm{x}} \phi_{\bm{a}} (-1)^{\bm{w'}\cdot\bm{a} + (\bm{v}\oplus \bm{v'})\cdot \bm{x}} \ket{P_N(\bm{u}, \bm{a},  \bm{x}) \oplus \bm{e}_x}  
\otimes \sum_{\bm{a'}, \bm{x'}} (-1)^{\bm{w'}\cdot\bm{a'}  + \bm{v'}\cdot \bm{x'}}  \ket{P_N(\bm{u}\oplus\bm{u'},  \bm{a'},  \bm{x'}) \oplus \bm{e}_X \oplus \bm{e'}_X }.
\end{align}
where for the last equality, we use  variable changes  $\bm{a'} \leftarrow  \bm{a'} \oplus \bm{a}$, and $ \bm{x'} \leftarrow \bm{x'} \oplus \bm{x}$. It can be seen that the error on the system $\mathcal{S}$ has propagated to the system $\mathcal{S}'$, and $\ket{\theta}_{\mathcal{S} \mathcal{S}'}$ is a product state that can be rewritten as
\begin{align}
\ket{\theta}_{\mathcal{S} \mathcal{S}'} &= X^{\bm{e}_X} Q_N (\ket{\bm{u}}_{\mathcal{Z}} \otimes Z^{\bm{w'}}\ket{\phi}_{\mathcal{I}} \otimes  \ket{\bm{\oline{v \oplus v'}}}_{\mathcal{X}})  \otimes X^{\bm{e}_X \oplus \bm{e'}_X}Q_N (\ket{\bm{u'} \oplus \bm{u}}_{\mathcal{Z}'} \otimes \ket{\oline{\bm{w'}}}_{\mathcal{I}'} \otimes  \ket{\oline{\bm{v'}}}_{\mathcal{X}'}). \label{eq:theta-f}
\end{align}
Hence, the partial state of the system $\mathcal{S}$  is the same as in (\ref{eq:steane_state}). Further, measuring the qubits of the ancilla system in the Pauli $Z$ basis, we get $\bm{m} = P_N( \bm{u} \oplus \bm{u'}, \bm{a'}, \bm{x'}) \oplus \bm{e}_X  \oplus \bm{e'}_X \in \{0,1\}^N$, 
 for some random vectors $\bm{a'} \in \{0, 1\}^{|\mathcal{I}|}$ and $\bm{x'} \in \{0, 1\}^{|\mathcal{X}|}$. 
\end{proof}
\smallskip Two observations are in place here.

\smallskip First, from~(\ref{eq:steane_state}), it follows that the frozen vector corresponding to $\mathcal{X} \subset \mathcal{S}$ (original system) has changed to $\bm{v} \oplus \bm{v'}$, after the Steane's procedure. Further, the logical $Z$ operator corresponding to $Z^{\bm{w'}}$ gets applied on $\mathcal{S}$. However, since we know $\bm{w'}$, we can reverse this logical operation using (\ref{eq:log-z}). Hence, the logical information encoded in $\mathcal{S}$ has not been altered, due to the syndrome extraction.

\smallskip The second observation concerns the SC decoding, which takes as input the measurement outcome $\bm{m} = P_N( \bm{u} \oplus \bm{u'}, \bm{a'}, \bm{x'}) \oplus \bm{e}_X  \oplus \bm{e'}_X $, and the frozen vector $\bm{u} \oplus \bm{u'}$. It  produces an estimate of the information vector $(\bm{a'}, \bm{x'})$, from which we can produce an estimate of the total error $\bm{e}_X  \oplus \bm{e'}_X$.  The vector $\bm{a'}$ may be correctly estimated, owning to the fact that it corresponds to good virtual channels. Some of the $\bm{x'}$ positions may be incorrectly estimated, but this does not matter, as the induced logical error corresponds to an $X$-type stabilizer operator, acting trivially on the code space.  

%\footnote{Corresponding to the unfrozen virtual channels, see Fig.~\ref{fig:qpolar_N8_Zbasis}. Note that the SC decoding  estimates the inputs of the classical polar transform $P_N$, given a noisy version of the corresponding codeword.} 

\smallskip  Hence, assuming the vector $\bm{a'}$ is decoded correctly, we also get the correct value of the total error  $\bm{e}_X  \oplus \bm{e'}_X$ (up to an $X$-type stabilizer operator). We then correct the $\cl{S}$ system by applying $X^{\bm{e}_X  \oplus \bm{e'}_X}$, which will leave the error $X^{\bm{e'}_X}$ (original error on $\cl{S}'$) on $\cl{S}$ after correction. This leftover error, may hopefully be corrected in the next round of correction, where we may similarly add another error from the ancilla system. However, note that the advantage of error correction is that it does not allow errors to accumulate and the only left error on the encoded state is due to the last round of error correction (which may be  corrected, when the encoded logical state is eventually measured). Hence, we may stabilize logical qubits against noise, by doing error correction repeatedly.

It is worth noticing that for topological (or some families of quantum LDPC codes),  fault tolerant error correction needs the decoding operation to be applied on a time window, composed of several consecutive syndrome extractions, e.g., \cite{fowler2012surface}. For Steane's fault tolerant error correction, decoding is simply applied on each extracted ``syndrome'' (recall the ``syndrome'' in this case is actually a noisy codeword).

\paragraph*{$Z$-Error Correction.}  The decoding of $Z$ errors can be done similarly to the case of $X$ errors as follows. To extract the syndrome for $Z$ errors (see Fig.~\ref{fig:steane_Zerror}), one needs an ancilla system $\mathcal{S}'$ prepared in a logical $Z$ basis  state,
\begin{equation}
\ket{\widetilde {\bm{w'}}}_{\mathcal{S}'} = Q_N ( \ket{\bm{u'}}_{\mathcal{Z}'} \otimes \ket{\bm{w'}}_{\mathcal{I}'} \otimes \ket{\oline{\bm{v'}}}_{\mathcal{X}'}),
\end{equation}

\smallskip After preparing the  state $\ket{\widetilde {\bm{w'}}}_{\mathcal{S}'} $ on ancilla system ${\mathcal{S}'}$, the transversal $\cnot_{\mathcal{S}' \to \mathcal{S}}$ gate is applied on corresponding qubits of systems $\mathcal{S}$ and  $\mathcal{S}'$. Then, each qubit in the ancilla system $\mathcal{S}'$ is measured in the Pauli $X$ basis. The measurement output $\bm{m}$ is a noisy codeword of the classical polar code $\mathcal{P}(N, \pi(\cl{X}), \pi(\bm{v}))$, induced in $X$ basis (Fig.~\ref{fig:qpolar_N8_Xbasis}).

\begin{lemma} \label{lem:steane_Z}
Let $\bm{e}_Z$ and $\bm{e'}_Z$ be preparation errors that has happened on systems $\mathcal{S}$ and $\mathcal{S}'$, respectively. Then, after the Steane's $Z$-error syndrome extraction, the state of the system $\mathcal{S}$ is given by
\begin{equation} \label{eq:steane_state_Z}
Z^{\bm{e}_Z} \ket{\widetilde{X^{\bm{w'}}\phi}}_\mathcal{S} =  Z^{\bm{e}_Z} Q_N (\ket{\bm{u} \oplus  \bm{u'}}_{\mathcal{Z}} \otimes X^{\bm{w'}}\ket{\phi}_{\mathcal{I}} \otimes  \ket{\bm{\oline{v}}}_{\mathcal{X}}).
\end{equation}
Further, the measurement outcome is a noisy codeword of the classical polar code $\mathcal{P}(N, \pi(\cl{X}), \pi(\bm{v}))$, with frozen set $\cl{X}$ and frozen vector $ \pi(\bm{v} \oplus \bm{v'})$, as follows,
\begin{equation} \label{eq:steane-syn-Z}
\bm{m} = P_N^\top( \bm{z'}, \bm{a'}, \bm{v} \oplus \bm{v'}) \oplus \bm{e}_Z \oplus \bm{e'}_Z,
\end{equation}
where $\bm{a'} \in \{0, 1\}^{|\mathcal{I}|}$ and $\bm{x'} \in \{0, 1\}^{|\mathcal{X}|}$ are random vectors. 
\end{lemma}

\smallskip The proof of Lemma \ref{lem:steane_Z} is similar to that of Lemma \ref{lem:steane}, by expanding quantum states of systems $\mathcal{S}$ and $\mathcal{S}'$ in the Pauli $X$ basis. Finally, based on the frozen vector $\bm{v} \oplus \bm{v'}$ and the noisy codeword $\bm{m}$ in (\ref{eq:steane-syn-Z}), the SC decoder generates an estimate of $\bm{a'}$, and in turn we can obtain an estimate of the error $\bm{e}_Z \oplus \bm{e'}_Z$ (up to a $Z$-type stabilizer operator).

\section{$\pone$ Codes: Quantum Polar Codes Encoding One Qubit} 
\label{sec:p-theorem-1}

\subsection{Shor $\pone$ Codes (Proof of Theorem $1$ From the Main Text)}

To prove Theorem $1$, we will use the decomposition  $Q_{N} = (I_{2^k} \otimes Q_{2^{n-k}})(Q_{2^k} \otimes I_{2^{n-k}})$, for $ 0\leq k \leq n$. This decomposition is illustrated in Fig.~\ref{fig:qpolar_decomposition}, where we  have  $2^{n-k}$ parallel $Q_{2^k}$ blocks, followed by $2^{k}$ parallel  $Q_{2^{n-k}}$ blocks. If one considers the quantum system $\cl{S}$ as a  vector of $N$ qubits, the decomposition illustrated in Fig.~\ref{fig:qpolar_decomposition} is equivalent to reshaping $\cl{S}$ as a $2^k \times 2^{n-k}$ matrix of qubits, with columns filled in by consecutive qubits from the original vector, then applying $Q_{2^k}$ on each column, and $Q_{2^{n-k}}$ on each row.

For the logical state $\ket{\widetilde{0}}_\mathcal{S}$, the first $i=2^k$ inputs of the quantum polar transform $Q_N$ are equal to $\ket{0}$, while the remaining $2^n - 2^k$ inputs are equal to $\ket{+}$ (see Fig.~\ref{fig:qpolar_decomposition}). Hence, each of the  $2^{n-k}$ parallel $Q_{2^k}$ unitaries in  Fig.~\ref{fig:qpolar_decomposition} acts trivially on its input state. Therefore, the input of the $r$-th $Q_{2^{n-k}}$ unitary, where $1 \leq r \leq 2^k$, is the quantum state $\ket{0}_{r,1} \otimes_{c=2}^{2^{n-k}} \ket{+}_{r,c}$. It follows that (omitting normalization factors),
\begin{align}
&\ket{\widetilde{0}}_\mathcal{S} 
   = \otimes_{r=1}^{2^k} Q_{2^n-k} \left( \ket{0}_{r,1} \otimes_{c=2}^{2^{n-k}} \ket{+}_{r,c} \right) \\
   &= \otimes_{r=1}^{2^k} Q_{2^n-k} \left( (\ket{+}_{r,1} + \ket{-}_{r,1}) \otimes_{c=2}^{2^{n-k}} \ket{+}_{r,c} \right) \\
   &= \otimes_{r=1}^{2^k} \Big( Q_{2^n-k} \left( \otimes_{c=1}^{2^{n-k}} \ket{+}_{r,c} \right) + Q_{2^n-k} \left( \ket{-}_{r,1} \otimes_{c=2}^{2^{n-k}} \ket{+}_{r,c} \right) \Big) \\
   &= \otimes_{r=1}^{2^k} \left( \otimes_{c=1}^{2^{n-k}} \ket{+}_{r,c} + \otimes_{c=1}^{2^{n-k}} \ket{-}_{r,c} \right)
\end{align}
For the logical state $\ket{\widetilde{1}}_\mathcal{S}$, the inputs of $Q_N$ is $\left( \ket{1}_{r,1} \otimes_{c=2}^{2^{n-k}} \ket{+}_{r,c} \right)$. Hence, using $\ket{1}_{r,1} = \ket{+}_{r,1} - \ket{-}_{r,1}$, it follows similarly that $\ket{\widetilde{1}}_\mathcal{S} =  \otimes_{r=1}^{2^k} \left( \otimes_{c=1}^{2^{n-k}} \ket{+}_{r,c} - \otimes_{c=1}^{2^{n-k}} \ket{-}_{r,c} \right)$.

\begin{figure}[!b]
\,\hfill\input{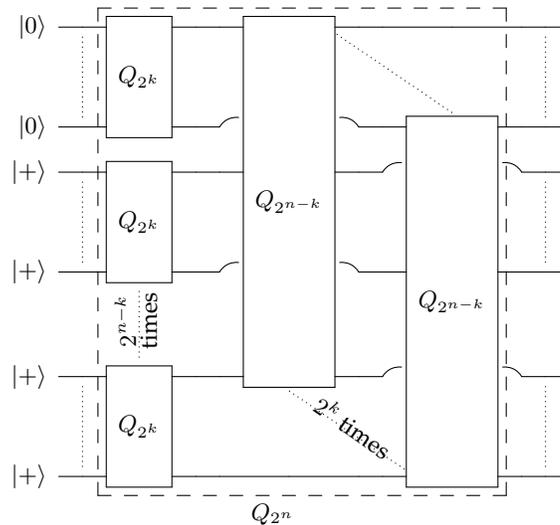}\hfill\,
\caption{ \small Quantum polar transform decomposition, using  $ Q_{2^n} = (I_{2^k} \otimes Q_{2^{n-k}})(Q_{2^k} \otimes I_{2^{n-k}})$. Bent wires go under the blocks they cross. In case input qubits are prepared as shown on the left, the encoded  state is the logical $\ket{0}$ state of a Shor code.}
\label{fig:qpolar_decomposition}
\end{figure}

%Hence, for the first (top) $Q_{2^n}$ unitary in Fig.~\ref{fig:qpolar_decomposition}, we get the state $Q_{2^n} (\ket{0\dots 01}) = \ket{1\dots 11}$. Therefore, the input of the $r$-th $Q_{2^{n-k}}$ unitary, where $1 \leq r \leq 2^k$, is the quantum state $\ket{1}_{r,1} \otimes_{c=2}^{2^{n-k}} \ket{+}_{r,c}$\,, and~(\ref{eq:polar-shor-log1}) follows as  above, using $\ket{1}_{r,1} = \ket{+}_{r,1} - \ket{-}_{r,1}$ (omitting normalization factors).

\subsection{Construction of General $\pone$ Codes} 
\makeatletter
\def\@currentlabel{\thesubsection}
\makeatother
\label{sec:construction-pone}

%The \emph{construction} of a $\pone$ code refers to the choice of the information position $i$, which determines how  well the code protects the encoded quantum information. Hence, the position~$i$ should be chosen in a way to optimize the code's decoding performance, depending on the specific noise model (\emph{i.e.}, quantum channel). For $\pone$ codes, any information position  $i\in\cl{S}$ may be chosen. For the subfamily of Shor-$\pone$ codes, the choice is restricted to positions $i = 2^k, 0\leq k \leq n$. 

%\medskip \noindent It's worth mentioning here that the best information position $i$ suggested by our numerical simulation is not the one corresponding to Shor $\pone$ codes, in general (see Section \ref{sec:num-res} below). Further, the performance gap between the best position $i$ and the best position in the subfamily of Shor $\pone$ codes can be quite significant, as illustrated in Section~\ref{sec:num-res}.

%To determine the information position, recall from the previous section that  only  errors on position $i$ need be decoded. Let $W_Z$ and $W_X$ be the classical channels induced in $Z$ and $X$ bases, respectively (Section~\ref{sec:construction-css}).
 
In this section, we provide more details about the construction of $\pone$ codes, and further numerically estimate its performance on Pauli and erasure channels.

\smallskip For $\pone$ codes, SC decoder only needs to decode the virtual channels corresponding to the information position $i$ for both $X$ and $Z$ error channels $W_Z$ and $W_X$, respectively. Let $P_e\left(W_Z^{(i)}\right)$ and $P_e\left(W_X^{(\pi(i))}\right)$ denote the error probability of the respective virtual channels, corresponding to the information position $i$. Recall that the error probability of a classical channel, is the probability of the maximum a posteriori estimate of the channel input, conditional on the observed channel output, not being equal to the actual input. Hence, the logical error rate of the $\pone$ code, with respect to the information position $i$, is given by
\begin{equation} \label{eq:q-logical-r}
	 P_e^L(i) = 1-\left( 1-P_e\left(W_Z^{(i)}\right) \right) \left(1-P_e\left(W_X^{(\pi(i))}\right) \right).
\end{equation}
The information position $i$ should be chosen so as to minimize the  corresponding logical error rate in~(\ref{eq:q-logical-r}). Precisely, we have
\begin{align}
 i &= \argmin_{j =  1,\dots,N} P_e^L(j), \text{ for } \pone \text{ codes}. \label{eq:i-for-pone}\\
 i &= \argmin_{j = 2^0,\dots, 2^n} P_e^L(j), \text{ for Shor-}\pone \text{ codes}. \label{eq:i-for-shor}
\end{align}

\subsection{Numerical Results} 
Here, we provide numerical results for the construction of $\pone$ and Shor-$\pone$ codes, for the quantum depolarizing and  quantum erasure channels.

\smallskip The quantum depolarizing channel with physical error probability $p$, is a Pauli channel as in~(\ref{eq:pauli_chanel}), with $P_I = 1-p$, and $p_X = p_Y = p_Z = p/3$. We use density evolution~\cite{tal2013construct} to estimate the error probability of virtual channels, \emph{i.e.}, $P_e\left(W_Z^{(j)}\right)$ and $P_e\left(W_X^{(\pi(j))}\right)$, $j=1,\dots,N$. The information positions for the $\pone$ and Shor-$\pone$ codes are then determined according to~(\ref{eq:i-for-pone}) and~(\ref{eq:i-for-shor}). Moreover, we consider the two constructions (or decoding strategies) given in Section \ref{sec:construction-css}, that is, either using or ignoring the correlations between $X$ and $Z$ errors. Here, ignoring correlations between $X$ and $Z$ errors may seem unfounded. We will later provide the rationale for this in Section \ref{subsec:prep-sim}, under the paragraph ``Prepared Codes''.  

%We dont take into account the correlation as  X and Z logical error rates are determined independently, considering logical states 0 and +, respectively (section 5.3). 

\smallskip  The quantum erasure channel erases the input qubit, with some probability $\varepsilon$, or transmits it perfectly, with probability $1-\varepsilon$. When a qubit is erased, it is replaced by a totally mixed state. Further, the channel also outputs a classical flag, which indicates whether the qubit  has been erased ($\ket{1}_E$) or not ($\ket{0}_E$). Hence, it can be represented as a quantum operation as follows,
\begin{equation}
\mathcal{W}_E(\rho) = (1- \varepsilon) \ket{0} \bra{0}_E \otimes \rho + \varepsilon \ket{1} \bra{1}_E \otimes \frac{\ident}{2}.
\end{equation}
It is easily seen that $\mathcal{W}_E$ acts as a classical erasure channel with erasure probability $\varepsilon$ in both Pauli $Z$ and Pauli $X$ basis. Hence, the induced channels $W_X$ and $W_Z$ are classical erasure channels, with erasure probability $\varepsilon$. The erasure probability of virtual channels ($P_e\left(W_Z^{(j)}\right)$ and $P_e\left(W_X^{(\pi(j))}\right)$, $j=1,\dots,N$) can be computed analytically~\cite{arikan2009channel}. The information positions for the $\pone$ and Shor-$\pone$ codes are then determined according to~(\ref{eq:i-for-pone}) and~(\ref{eq:i-for-shor}). For this channel, the two construction strategies (using or ignoring correlations) are easily seen to be equivalent. 

\smallskip Considering the quantum depolarizing channel, Fig.~\ref{fig:depolar_lerr_rates} shows the logical error rate $P_e^L(i)$ vs. the physical error probability $p$, for the $\pone$ and Shor-$\pone$ codes, with information position determined according to  to~(\ref{eq:i-for-pone}) and~(\ref{eq:i-for-shor}). Note that for fixed codelength $N=2^n$, the information position may vary depending on the physical error probability value. We consider the two decoding strategies mentioned above, namely, either using or ignoring the correlations between $X$ and $Z$ errors. It can be observed that using correlations yield (slightly) better decoding performance. Moreover, in both cases, we observe that the logical error rate of the  $\pone$ code is in general lower than that of the  Shor code, and the gap is increasing with increasing codelength (\emph{i.e.}, number of polarization steps $n$). We emphasize that the gap in the decoding performance is due to the channel polarization phenomenon and not to the minimum distance ($\pone$ and Shor-$\pone$ codes have actually the same minimum distance, see below).

\begin{figure}[!t]
\centering
\begin{subfigure}{\textwidth}
\includegraphics[width=.49\textwidth]{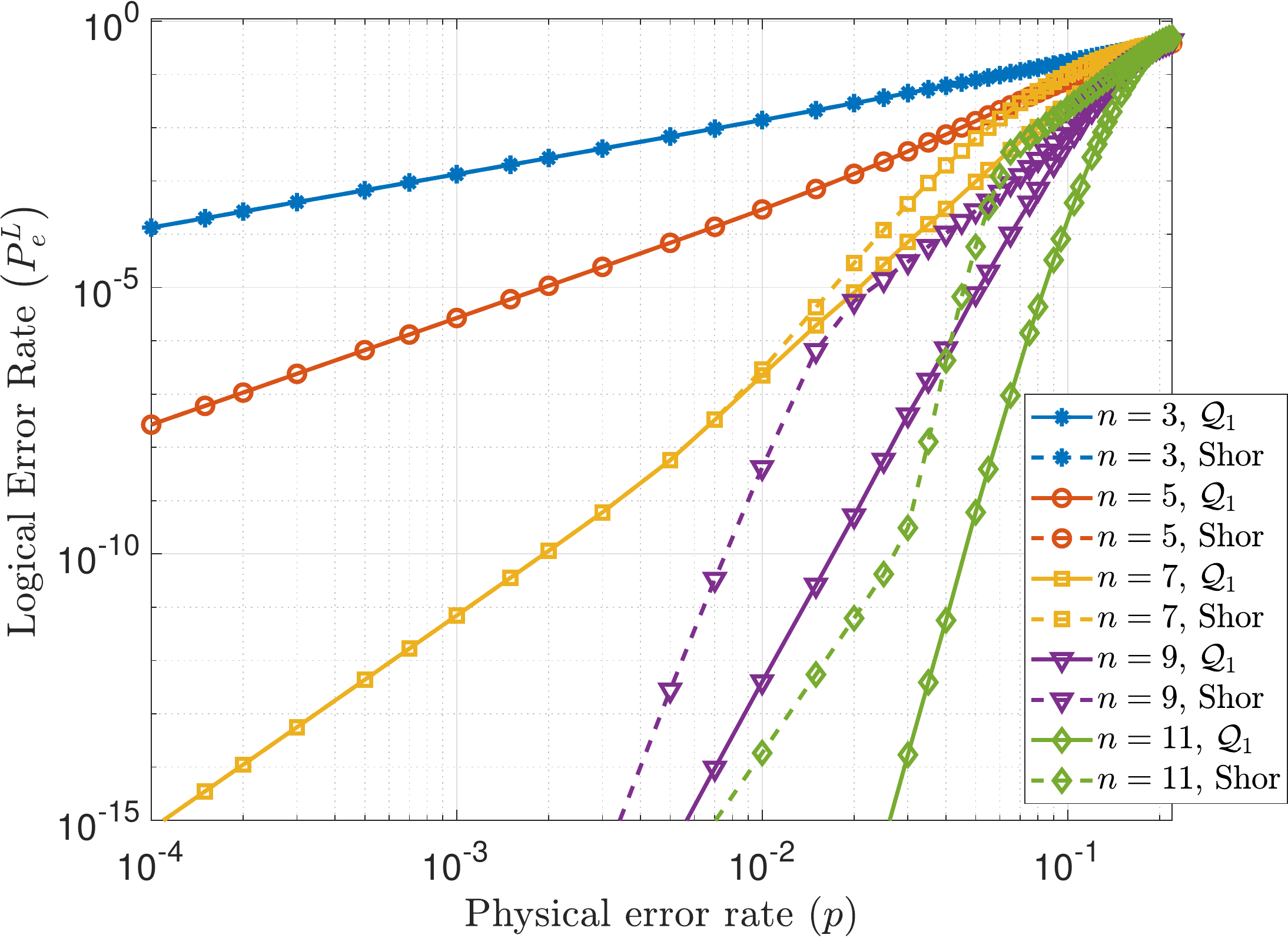}\hfill%
\includegraphics[width=.49\textwidth]{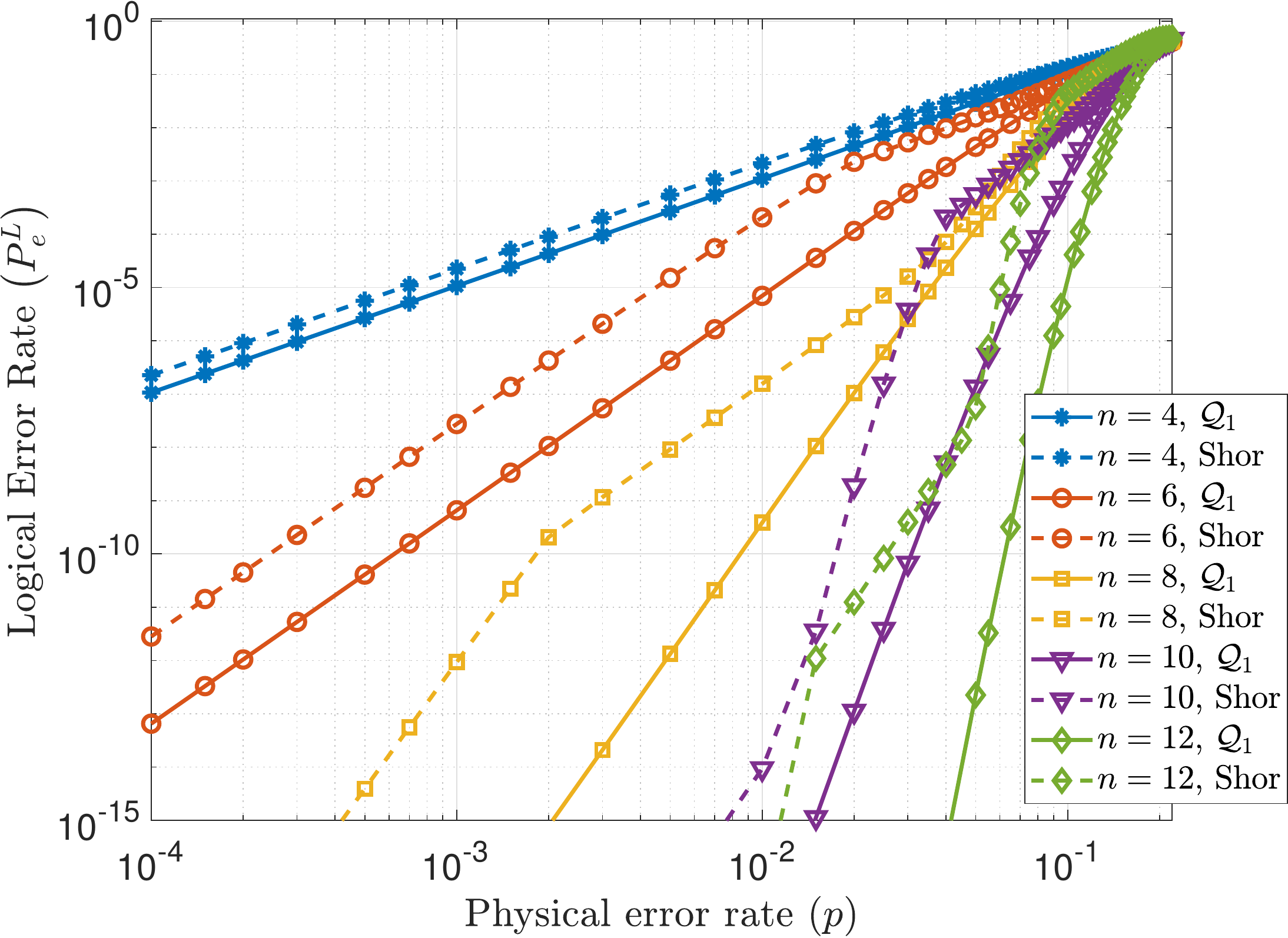}
\caption{Using correlations between $X$ and $Z$ errors. Left: $n$ odd (for $n=3,5,7$,  $\pone$ and Shor curves virtually coincide). Right: $n$ even.}
\label{fig:depolar_corre_xz_dec}
\end{subfigure}

\vspace*{3mm}
\begin{subfigure}{\textwidth}
\includegraphics[width=.49\textwidth]{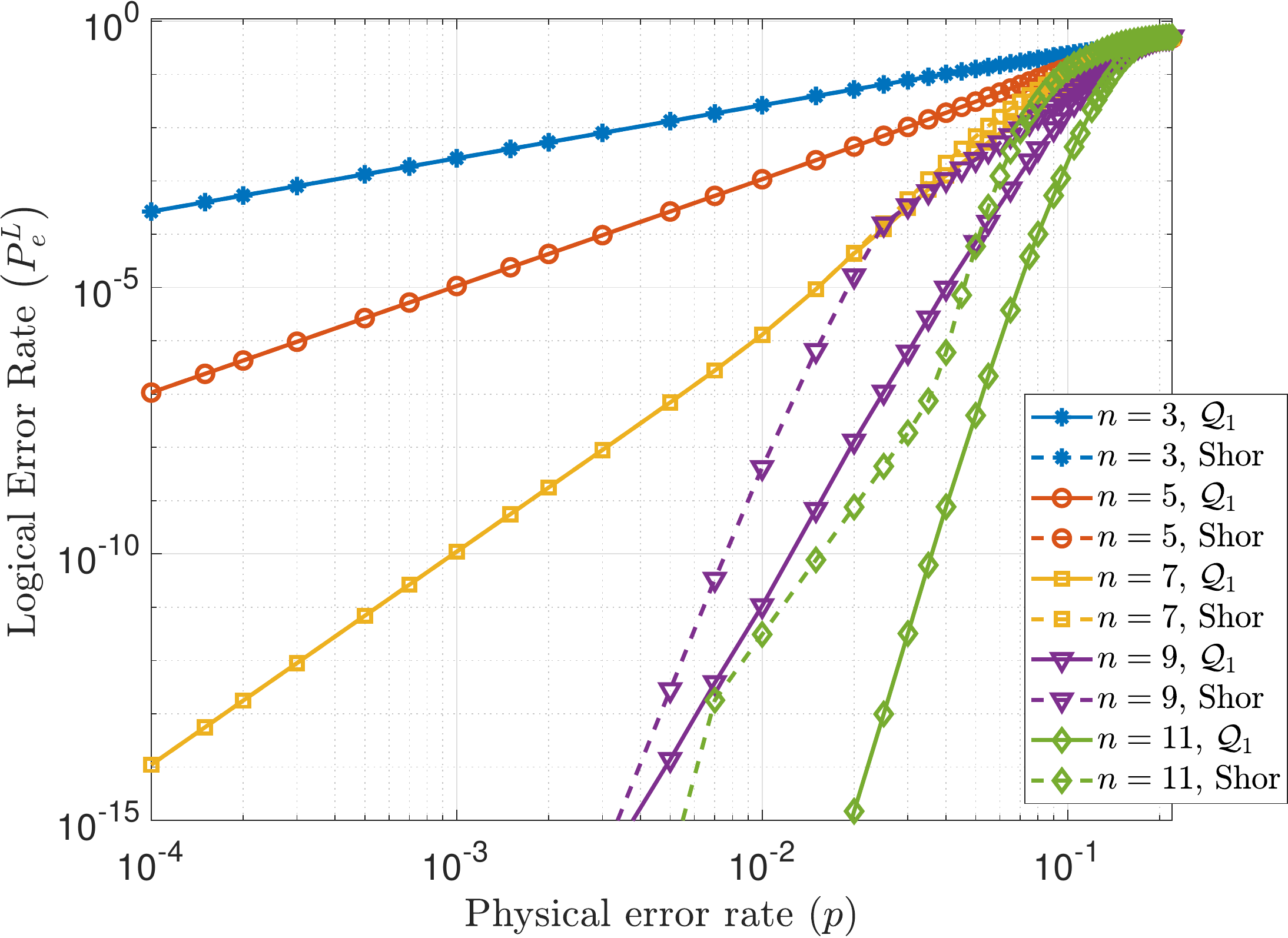}\hfill%
\includegraphics[width=.49\textwidth]{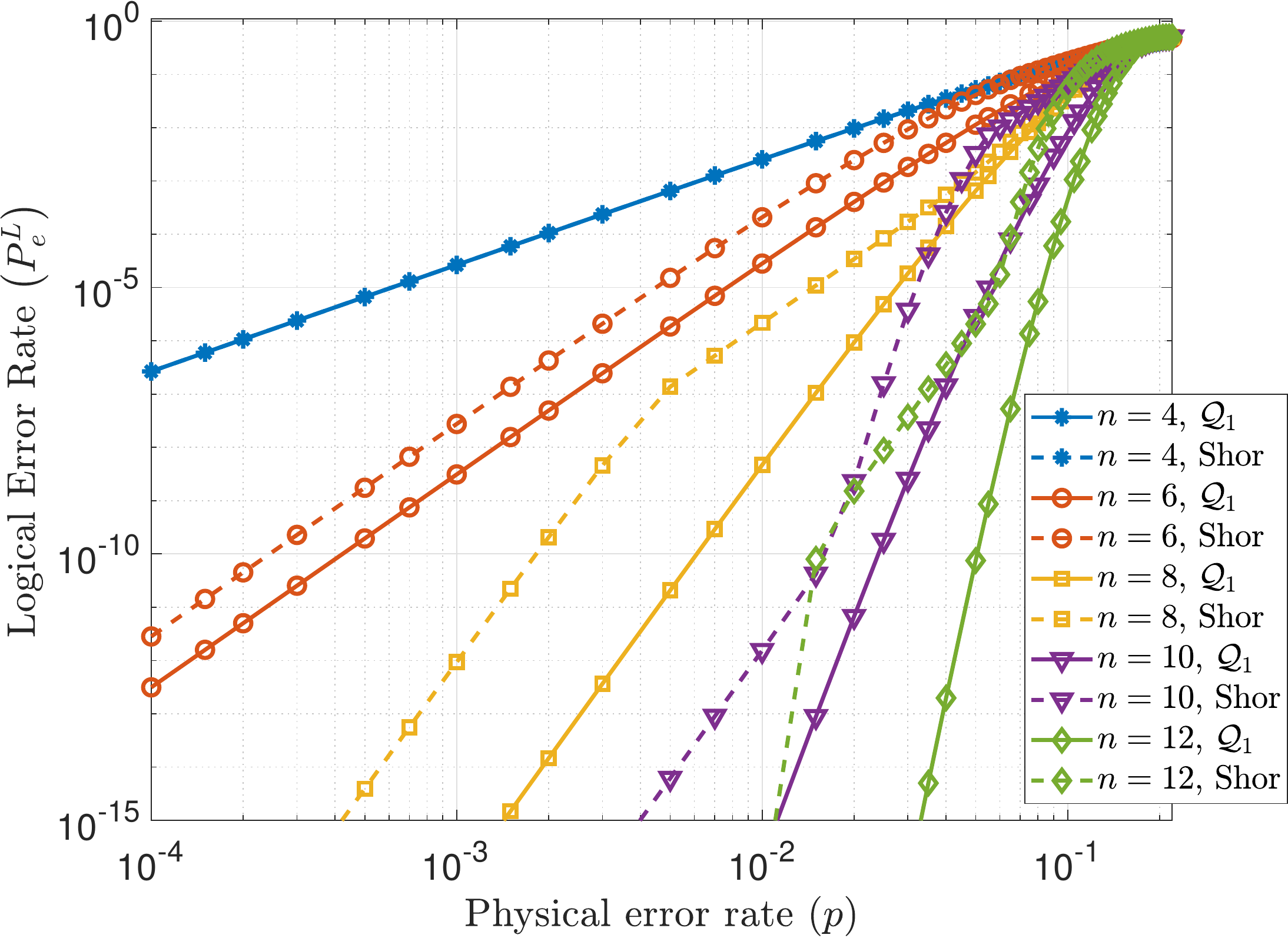}
\caption{Ignoring correlations between $X$ and $Z$ errors. Left: $n$ odd (for $n=3,5,7$,  $\pone$ and Shor curves virtually coincide). Right: $n$ even (for $n=4$,  $\pone$ and Shor curves virtually coincide).}
\label{fig:depolar_indep_xz_dec}
\end{subfigure}
\captionsetup{justification=centering}
\caption{Logical error rate of $\pone$ and Shor codes, for the depolarizing channel.}
\label{fig:depolar_lerr_rates}
\end{figure}

\smallskip As just mentioned, the information position $i$ may vary, depending on the physical error probability $p$.  However,  numerical results suggests that as $p$ goes to zero, the  information position $i$ reaches a stable (constant) value. These values are reported in Table~\ref{tab:b-indx}, for both decoding strategies (either using or ignoring correlations), and different values of $n$.  We  also report in Table~\ref{tab:b-indx} similar results for the quantum erasure channel (we omit the logical error rate curves for the quantum erasure channel, since they are of the same nature as the logical error rate curves in Fig.~\ref{fig:depolar_lerr_rates}).  It can be observed that for  $\pone$ codes, the information position values reported in Table~\ref{tab:b-indx} depend on  the noise model (depolarizing or erasure channel), as well as the decoding strategy (using or ignoring correlations). 

\smallskip Finally, we note that for a given $n$ value, all $\pone$ and Shor codes with information positions given in Table~\ref{tab:b-indx} have the same minimum distance, which is reported on the last column of the table. The reported minimum distance is the minimum weight of logical $X$ and $Z$ operators, and can be computed by using~(\ref{eq:log-x}) and~(\ref{eq:log-z}). As mentioned in the main text, the superior performance of $\pone$ codes owes to the SC decoding, which is able to decode beyond the minimum distance of the code, by effectively exploiting the  channel polarization property of polar codes  (or, to a certain extent, the degree of freedom provided by the full flexibility on the choice of the information index).

\begin{table}[!t]
%\centering
\captionsetup{justification=centering}
\caption[Best information positions  for low error probabilities]{Best information positions  for low error probabilities$^{(*)}$}
\label{tab:b-indx}
\begin{minipage}{.95\linewidth}
{\small\,\hfill%
%\begin{tabular}{|c||*{6}{X{15mm}|}|c|}
\begin{tabular}{|c||X{15mm}|X{15mm}|X{15mm}|X{15mm}||X{15mm}|X{15mm}||c|}
\hline 
       & \multicolumn{2}{c|}{Depolarizing Channel}   &  \multicolumn{2}{c||}{Depolarizing Channel} &  \multicolumn{2}{c||}{\multirow{2}{*}{Erasure Channel}}  &  \\
Levels of & \multicolumn{2}{c|}{using correlations}     &  \multicolumn{2}{c||}{ignoring correlations} &   \multicolumn{2}{c||}{\,} & \\
\cline{2-7}
Recursion ($n$) & \multicolumn{6}{c||}{Information position $i$ for best $\pone$ and Shor-$\pone$ codes} & Min \\
 \cline{2-7}
                & $\pone$    & Shor    &  $\pone$    & Shor    & $\pone$    & Shor  & dist.  \\
\hline
3   & 4    & 4      & 4    & 4     & 2   & 2     & 2  \\[-2mm]
    & \pstab{6.8}{-2} & \pstab{2.0}{-1} & \pstab{1.0}{-3} & \pstab{1.0}{-3} & \pstab{5.0}{-1} & \pstab{5.0}{-1} & \\
\hline
4   & 13   & 4      & 7    & 4     & 7   & 4     & 4  \\[-2mm]
    & \pstab{2.0}{-1} & \pstab{2.0}{-1} & \pstab{9.0}{-4} & \pstab{2.0}{-1} & \pstab{5.0}{-1} & \pstab{5.0}{-1} & \\
\hline
5   & 8    & 8      & 8    & 8     & 4   & 4     & 4  \\[-2mm]
    & \pstab{2.0}{-1} & \pstab{2.0}{-1} & \pstab{2.0}{-1} & \pstab{2.0}{-1} & \pstab{2.0}{-1} & \pstab{5.0}{-1} & \\
\hline
6   & 50  & 8      & 23   & 8     & 23  & 8     & 8  \\[-2mm]
    & \pstab{1.6}{-1} & \pstab{1.8}{-2} & \pstab{1.4}{-1} & \pstab{8.8}{-2} & \pstab{5.0}{-1} & \pstab{4.0}{-2} & \\
\hline
7   & 16   & 16     & 16   & 16    & 8   & 8     & 8  \\[-2mm]
    & \pstab{8.0}{-3} & \pstab{2.0}{-1} & \pstab{1.8}{-2} & \pstab{2.0}{-1} & \pstab{2.6}{-2} & \pstab{1.9}{-1} & \\
\hline
8   & 199  & 16     & 91   & 16    & 87  & 16    & 16 \\[-2mm]
    & \pstab{1.4}{-2} & \pstab{2.0}{-3} & \pstab{6.2}{-2} & \pstab{4.0}{-3} & \pstab{2.6}{-1} & \pstab{6.0}{-3} & \\
\hline
9   & 32   & 32     & 32   & 32    & 16  & 16    & 16 \\[-2mm]
    & \pstab{1.0}{-3} & \pstab{1.6}{-2} & \pstab{2.0}{-3} & \pstab{2.4}{-2} & \pstab{5.0}{-3} & \pstab{4.4}{-2} & \\
\hline
10  & 806  & 32     & 363  & 32    & 343 & 32    & 32 \\[-2mm]
    & \pstab{6.0}{-5} & \pstab{3.0}{-4} & \pstab{6.2}{-2} & \pstab{6.0}{-4} & \pstab{9.0}{-2} & \pstab{1.0}{-3} & \\
\hline
11  & 96   & 64     & 96   & 64    & 32  & 32    & 32 \\[-2mm]
    & \pstab{3.0}{-4} & \pstab{4.0}{-3} & \pstab{6.4}{-4} & \pstab{6.3}{-3} & \pstab{1.0}{-3} & \pstab{1.4}{-2} & \\
\hline
12  & 3222 & 64     & 1451 & 64    & 1367 & 64   & 64 \\[-2mm]
    & \pstab{6.0}{-5} & \pstab{6.0}{-5} & \pstab{6.6}{-2} & \pstab{1.0}{-4} & \pstab{4.4}{-2} & \pstab{2.4}{-4} & \\ 
\hline
\end{tabular}}\hfill\,%

%\renewcommand*\footnoterule{}
%\vspace*{1em}
\begin{justify}
{\small  $^{(*)}$For the depolarizing channel, we consider physical error rates $p \in [10^{-5}, 2\!\times\!10^{-1}]$ (note that the coherent information of the channel vanishes for $p\approx 0.1893$). Reported information position values are constant for physical error rates $p \in [10^{-5}, p_0]$. The value of $p_0$ is reported between round brackets, under the value of $i$ (tiny font size). Similarly, for the quantum erasure channel, we consider channel erasure probability values $\varepsilon \in [10^{-5}, 5\!\times\!10^{-1}]$. Reported information position values are constant for channel erasure probabilities $\varepsilon \in [10^{-5}, \varepsilon_0]$, where the  value of $\varepsilon_0$ is reported between round brackets, under the value of $i$.}
\end{justify}
\end{minipage}
\end{table}

\newpage
\section{Measurement-Based Preparation of Logical $\pone$ Code States} \label{sec:prep_lemma}

\subsection{Measurement-Based Preparation Without Noise (Proof of Lemma $1$ From the Main Text)} 
\makeatletter
\def\@currentlabel{\thesubsection}
\makeatother
\label{sec:prep_lemma_noiseless} 

 We consider the standard ``phase kickback trick'' implementation of Pauli $X \otimes X$ and $Z \otimes Z$ measurements, using ancilla qubits as in Fig.~\ref{fig:mZZ_mXX_notation_and_circuit}.

\begin{figure}[H] 
\centering
\begin{subfigure}[t]{.9\linewidth}
\begin{subfigure}[t]{.4\linewidth}
\,\hfill\input{measureZZ_notation}\,
\end{subfigure}\hfill%
\raisebox{-1.4em}{$\equiv$}\hfill%
\begin{subfigure}[t]{.5\linewidth}
\,\qquad\input{measureZZ_circuit}\hfill\,
\end{subfigure}
\caption{Pauli $Z \otimes Z$ measurement: shorthand notation (left) and quantum circuit implementing the Pauli $Z \otimes Z$ measurement (right).}
\label{fig:mZZ_notation_and_circuit}
\vspace*{3mm}
\end{subfigure}
\begin{subfigure}[t]{.9\linewidth}
\begin{subfigure}[t]{.4\linewidth}
\,\hfill\input{measureXX_notation}\,
\end{subfigure}\hfill%
\raisebox{-1.4em}{$\equiv$}\hfill%
\begin{subfigure}[t]{.5\linewidth}
\,\qquad\input{measureXX_circuit}\hfill\,
\end{subfigure}
\caption{Pauli $X \otimes X$ measurement: shorthand notation (left) and quantum circuit implementing the Pauli $X \otimes X$ measurement (right).}
\label{fig:mXX_notation_and_circuit}
\end{subfigure}
\captionsetup{justification=centering}
\caption{Two-qubit Pauli measurements: shorthand notation and quantum circuits implementing the measurements.}
\label{fig:mZZ_mXX_notation_and_circuit}
\end{figure}

% shows the shorthand notation that will be used in the rest of the paper  for Pauli $Z\otimes Z$ and Pauli $X\otimes X$ measurements, as well as the quantum circuits implementing these measurements.

\subsection*{Case 1: Preparation Using Pauli $Z \otimes Z$ Measurements} \label{case:1}
%Let $k' \eqdef k - \frac{N}{2}$, and $\cl{Z}' \eqdef \{1, \dots, k'\}$ and $\cl{X}' \eqdef \{k', \dots, \frac{N}{2}\}$. We assume that we have been given the following two equivalent $\pone$ code states of length $\frac{N}{2}$
%%
%\begin{align}
%\ket{q_{\frac{N}{2}}^1}_{\mathcal{S}_1} \eqdef Q_{\frac{N}{2}} \ket{\bm{u_1}, \oline{\bm{v_1}}}_{\mathcal{S}_1}, \label{eq:q-N/2-1} \\
%\ket{q_{\frac{N}{2}}^2}_{\mathcal{S}_2} \eqdef Q_{\frac{N}{2}} \ket{\bm{u_2}, \oline{\bm{v_2}}}_{\mathcal{S}_2}, \label{eq:q-N/2-2}
%\end{align}
%where $\bm{u_1}, \bm{u_2} \in \{0, 1\}^{k'}$ and $\bm{v_1}, \bm{v_2} \in \{0, 1\}^{\frac{N}{2}-k'}$ are known. 

In this case, our procedure consists of performing transversal Pauli $Z \otimes Z$ measurements on corresponding qubits of systems $\mathcal{S}_1 = \{1, \dots, K/2\}$, and $\mathcal{S}_2 = \{K/2+ 1, \dots, K\}$, prepared in $\pone$ code states $\ket{q_{\frac{K}{2}}^1}_{\mathcal{S}_1} \eqdef Q_{\frac{K}{2}} \ket{\bm{u_1}, \oline{\bm{v}}_{\bm{1}}}_{\mathcal{S}_1}$ and $\ket{q_{\frac{K}{2}}^2}_{\mathcal{S}_2} \eqdef Q_{\frac{K}{2}} \ket{\bm{u_2}, \oline{\bm{v}}_{\bm{2}}}_{\mathcal{S}_2}$, respectively, as illustrated in Fig.~\ref{fig:qpolarprep_mZZ}.

\begin{figure}[H]
\centering
\,\hfill \hspace*{8em}\input{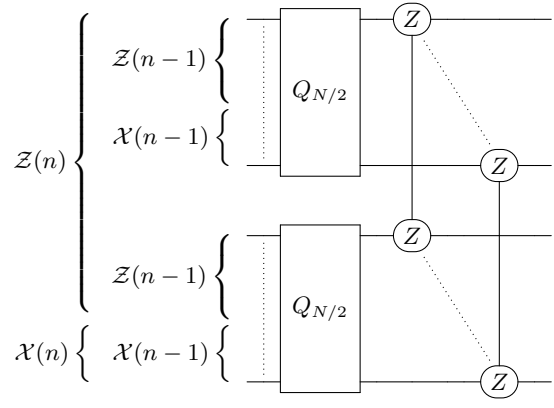}\hfill\,
\captionsetup{justification=centering}
\caption{Preparation using Pauli $Z\otimes Z$ measurements ($i(k) = i(k-1)+\frac{K}{2} \in \{\frac{K}{2}+1, \dots, K\}$).}
\label{fig:qpolarprep_mZZ}
\end{figure}

Expanding the quantum state $\ket{\oline{\bm{v}}_{\bm{1}}}_{\cl{X}(n-1)}$ in the Pauli $Z$ basis, and using~(\ref{eq:q-pol-enc-x}), we get, up to a normalization factor,
\begin{equation}
\ket{q_{\frac{K}{2}}^1}_{\mathcal{S}_1} = \sum_{\bm{x_1} \in \{0,1\}^{\frac{K}{2}-i(k-1)}} (-1)^{\bm{v_1} \cdot \bm{x_1}}  \ket{	P_{\frac{K}{2}}(\bm{u_1}, \bm{x_1})}_{\mathcal{S}_1}. \label{eq:N/2-comp}
\end{equation}
%We may write in the Pauli $Z$ basis, up to a normalization factor,
%%
%\begin{align}
%\ket{q_{\frac{N}{2}}^1}_{\mathcal{S}_1} &= Q_{\frac{N}{2}} \left( \ket{\bm{u_1}, \oline{\bm{v}}_1}_{\mathcal{S}_1}\right), \nonumber \\
%&= \sum_{\bm{x_1} \in \{0,1\}^{\frac{N}{2}-i(n)}} (-1)^{\bm{v_1} \cdot \bm{x_1}} Q_{\frac{N}{2}} \left(\ket{\bm{u_1}, \bm{x_1}}_{\mathcal{S}_1}\right), \nonumber \\
%&= \sum_{\bm{x_1} \in \{0,1\}^{\frac{N}{2}-i(n)}} (-1)^{\bm{v_1} \cdot \bm{x_1}}  \ket{	P_{\frac{N}{2}}(\bm{u_1}, \bm{x_1})}_{\mathcal{S}_1}, \label{eq:N/2-comp}
%\end{align}
%where in the second equality we have expanded the quantum state $\ket{\oline{\bm{v}}_1}_{\cl{X}(n-1)}$ in the Pauli $Z$ basis, and the third equality follows from~(\ref{eq:q-pol-enc-x}).

Similarly, we may also expand the quantum state $\ket{q_{K/2}^2}_{\mathcal{S}_2}$ in the Pauli $Z$ basis. Further, we consider the circuit in Fig.~\ref{fig:mZZ_notation_and_circuit} to perform  transversal Pauli $Z \otimes Z$ measurements, which are done in the following two steps (below, we denote by $\mathcal{S}_3$ the ancilla system $A$ from Fig.~\ref{fig:mZZ_notation_and_circuit}).

\begin{list}{}{\setlength{\labelwidth}{2em}\setlength{\leftmargin}{2em}\setlength{\listparindent}{0em}}
\item[$(1)$] 	We first take an $(K/2)$-qubit ancilla state $\ket{\bm{0}}_{\mathcal{S}_3}$, and then apply transversal $\cnot$ gates, $\cnot_{\mathcal{S}_1 \to \mathcal{S}_3}$ and  $\cnot_{\mathcal{S}_2 \to \mathcal{S}_3}$.  This gives the following joint quantum state on $\mathcal{S}_1 \cup \mathcal{S}_2 \cup \mathcal{S}_3$.
\begin{equation}
\ket{\eta}_{\mathcal{S}_1\mathcal{S}_2\mathcal{S}_3} \ = \!\!\! \sum_{\bm{x_1}, \bm{x_2} \in \{0,1\}^{\frac{K}{2}-i(k-1)}} (-1)^{\bm{v_1} \cdot \bm{x_1} + \bm{v_2} \cdot \bm{x_2}}   \ket{P_{\frac{K}{2}}(\bm{u_1}, \bm{x_1})}_{\mathcal{S}_1}  \ket{P_{\frac{K}{2}}(\bm{u_2}, \bm{x_2})}_{\mathcal{S}_2} 
  \ket{P_{\frac{K}{2}}(\bm{u_1} \oplus \bm{u_2}, \bm{x_1} \oplus \bm{x_2})}_{\mathcal{S}_3}. \label{eq:state-intm}
\end{equation}
\item[$(2)$] 	Then, we measure each qubit in the ancilla system $\mathcal{S}_3$ in the Pauli $Z$ basis. From~(\ref{eq:state-intm}), the measurement outcome gives a binary vector of length $K/2$ as follows,
\begin{equation}
 \bm{m} =  P_\frac{K}{2}(\bm{u'} , \bm{x}) \in \{0,1\}^{\frac{K}{2}},\label{eq:m-out-1}
\end{equation}
where $\bm{u'} = \bm{u_1} \oplus \bm{u_2} \in \{0, 1\}^{i(k)}$, and $\bm{x} \in \{0,1\}^{\frac{K}{2}-i(k)}$ is a random vector. Further, from~(\ref{eq:state-intm}) and~(\ref{eq:m-out-1}), the state of the joint system $\mathcal{S}_1 \cup \mathcal{S}_2$ after the measurements is as follows
\begin{equation}
\ket{\eta'}_{\mathcal{S}_1\mathcal{S}_2} = \sum_{ \substack{\bm{x_1}, \bm{x_2} \in \{0,1\}^{\frac{K}{2}-i(k-1)} \\ \bm{x_1} \oplus \bm{x_2} = \bm{x}}} (-1)^{\bm{v_1} \cdot \bm{x_1} + \bm{v_2} \cdot \bm{x_2}}  \ket{P_{\frac{K}{2}}(\bm{u_1}, \bm{x_1})}_{\mathcal{S}_1} \ket{P_{\frac{K}{2}}(\bm{u_2}, \bm{x_2})}_{\mathcal{S}_2}
\end{equation}
\end{list}

It can be seen as follows that the quantum state $\ket{\eta'}_{\mathcal{S}_1\mathcal{S}_2}$ is the $\pone$ code state $\ket{q_{K}}_{\mathcal{S}_1\mathcal{S}_2}$ as in \emph{Case 1} from the main text,
\begin{align}
\ket{\eta'}_{\mathcal{S}_1\mathcal{S}_2} &= \sum_{ \substack{\bm{x_1}, \bm{x_2} \in \{0,1\}^{\frac{K}{2}-i(k-1)} \\ \bm{x_1} \oplus \bm{x_2} = \bm{x}}} (-1)^{\bm{v_1} \cdot \bm{x_1} + \bm{v_2} \cdot \bm{x_2}}  \ket{P_{\frac{K}{2}}(\bm{u_1}, \bm{x_1})}_{\mathcal{S}_1} \ket{P_{\frac{K}{2}}(\bm{u_2}, \bm{x_2})}_{\mathcal{S}_2}  \nonumber \\
&= \sum_{ \bm{x_2}} (-1)^{\bm{v_1} \cdot (\bm{x} + \bm{x_2}) + \bm{v_2} \cdot \bm{x_2}}  \ket{P_{\frac{K}{2}}(\bm{u'} \oplus \bm{u_2}, \bm{x} \oplus \bm{x_2})}_{\mathcal{S}_1} \ket{P_{\frac{K}{2}}(\bm{u_2}, \bm{x_2})}_{\mathcal{S}_2}  \nonumber \\
&= (-1)^{\bm{v_1} \cdot \bm{x}} \sum_{ \bm{x_2}} (-1)^{(\bm{v_1}  + \bm{v_2}) \cdot \bm{x_2}}  \ket{P_{K}(\bm{u'}, \bm{x}, \bm{u_2}, \bm{x_2})}_{\mathcal{S}_1\mathcal{S}_2}  \nonumber \\
&=  Q_K \ket{\bm{u'}, \bm{x}, \bm{u_2}, \oline{\bm{v_1} \oplus \bm{v_2}}}_{\mathcal{S}_1\mathcal{S}_2} 
\end{align}
where in the second equality, we have used $\bm{u_1} = \bm{u'} \oplus \bm{u_2}$ and $\bm{x_1} = \bm{x} \oplus \bm{x_2}$,  and  in the third equality, we have used $P_K (\bm{a}, \bm{b}) = (P_{\frac{K}{2}}(\bm{a} \oplus \bm{b} ),  P_{\frac{K}{2}}(\bm{b}))$, $\bm{a}, \bm{b} \in \{0, 1\}^{\frac{K}{2}}$, using the recursion of the classical polar transform given in Fig.~\ref{fig:cpolar_recursion}. 

%Note that $\ket{\eta'}_{\mathcal{S}_1\mathcal{S}_2}$ is same as $\ket{q}_{\mathcal{S}_1\mathcal{S}_2}$ in (\ref{eq:joint-state}), hence, we have prepared the desired $\pone$ code state.

\medskip Hence, after the Pauli $Z \otimes Z$ measurements, we have prepared the $\pone$ code state on the joint system $\mathcal{S}= \mathcal{S}_1 \cup \mathcal{S}_2$,
\begin{equation}
\ket{q_K}_\mathcal{S} = Q_K \ket{\bm{u'}, \bm{x}, \bm{u_2}, \oline{\bm{v_1} \oplus \bm{v_2}}}_{\mathcal{S}}. \label{eq:noiseles-state-k}
\end{equation} 
Finally, from (\ref{eq:m-out-1}), we have that $P_\frac{K}{2}(\bm{m}) = (\bm{u'}, \bm{x}) \in \{0, 1\}^{\frac{K}{2}}$ (using $P_K^2 = I$). Hence, $\bm{x}$ is determined by, $\bm{x} = P_{\frac{K}{2}}(\bm{m})\lvert_{\mathcal{X}(n-1)}$, as desired.

\subsection*{Case 2: Preparation Using Pauli $X \otimes X$ Measurements} \label{case:2}
%Let $\cl{Z}' \eqdef \{1, \dots, k\}$ and $\cl{X}' \eqdef \{k, \dots, \frac{N}{2}\}$. We assume that we have been given the following two equivalent $\pone$ code states of length $\frac{N}{2}$
%%
%\begin{align}
%\ket{q_{\frac{N}{2}}^1}_{\mathcal{S}_1} \eqdef Q_{\frac{N}{2}} \ket{\bm{u_1}, \oline{\bm{v_1}}}_{\mathcal{S}_1}, \label{eq:q-N/2-1-x} \\
%\ket{q_{\frac{N}{2}}^2}_{\mathcal{S}_2} \eqdef Q_{\frac{N}{2}} \ket{\bm{u_2}, \oline{\bm{v_2}}}_{\mathcal{S}_2}, \label{eq:q-N/2-2-x}
%\end{align}
%where $\bm{u_1}, \bm{u_2} \in \{0, 1\}^{k}$ and $\bm{v_1}, \bm{v_2} \in \{0, 1\}^{\frac{N}{2}-k}$ are known. 

In this case, our procedure consists in performing transversal Pauli $X \otimes X$ measurements on corresponding qubits of systems $\mathcal{S}_1$ and $\mathcal{S}_2$, prepared in $\pone$ code states $\ket{q_{\frac{K}{2}}^1}_{\mathcal{S}_1} \eqdef Q_{\frac{K}{2}} \ket{\bm{u_1}, \oline{\bm{v}}_{\bm{1}}}_{\mathcal{S}_1}$ and $\ket{q_{\frac{K}{2}}^2}_{\mathcal{S}_2} \eqdef Q_{\frac{K}{2}} \ket{\bm{u_2}, \oline{\bm{v}}_{\bm{2}}}_{\mathcal{S}_2}$, as illustrated in Fig.~\ref{fig:qpolarprep_mXX}.

\begin{figure}[H]
\centering
\,\hfill \hspace*{8em}\input{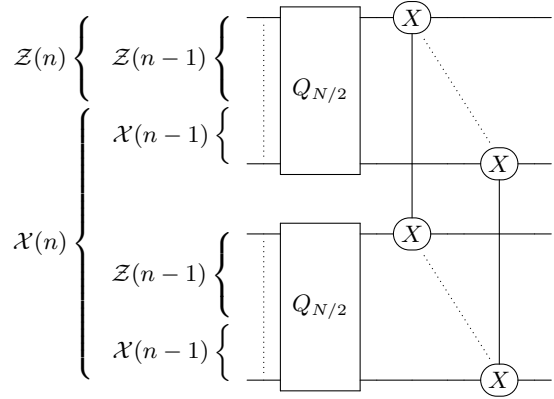}\hfill\,
\captionsetup{justification=centering}
\caption{Preparation using Pauli $X\otimes X$ measurements ($i(k) = i(k-1) \in \{1,\dots,\frac{K}{2}\}$).}
\label{fig:qpolarprep_mXX}
\end{figure}
 We skip the proof here as it is similar to the \emph{Case 1}. By expanding $\ket{q_{\frac{K}{2}}^1}_{\mathcal{S}_1}$ and $\ket{q_{\frac{K}{2}}^2}_{\mathcal{S}_2}$, in the Pauli $X$ basis instead and using (\ref{eq:q-pol-enc-z}) it can be seen that the measurement outcome of Pauli $X \otimes X$ measurements is given by,
\begin{equation}
 \bm{m} = P_\frac{K}{2}^\top(\bm{z}, \bm{v'}) \in \{0, 1\}^{\frac{K}{2}}, \label{eq:m-out-x}
\end{equation}
where $\bm{z} \in \{0, 1\}^{i(k-1)}$ is a random vector, and $\bm{v'} = \bm{v_1} \oplus \bm{v_2} \in \{0, 1\}^{\frac{K}{2}-i(k-1)}$. Further, after measurements, the state on $\mathcal{S} = \mathcal{S}_1 \cup \mathcal{S}_2$ is a $\pone$ state is given by,  
\begin{equation}
\ket{q_K}_\mathcal{S} = \ket{\bm{u_1} \oplus \bm{u_2}, \overline{(\bm{v_1},  \bm{z},  \bm{v'})}}, \label{eq:joint-state-x} 
\end{equation}
where $\bm{z}$ can be determined from (\ref{eq:m-out-x}) as, $\bm{z} = P_{\frac{K}{2}}^\top (\bm{m})\lvert_{\mathcal{Z}(k-1)}$.

\subsection{Fault-tolerant Measurement-Based Preparation with Noise} 
\makeatletter
\def\@currentlabel{\thesubsection}
\makeatother
\label{sec:prep_lemma-noisy}

In this section, supposing a circuit level noise model with Pauli errors, we detail the measurement based preparation with the error detection from Procedure $2$ in the main text, for both Case $1$ and Case $2$. Then, we provide a proof of fault tolerance according to Theorem $3$ from the main text.

%and provide the proof of fault tolerance according to Theorem $3$ in the main text. We provide Lemmas \ref{lem:prep-n}-\ref{lem:fault-up-x}, which imply Theorem $3$. 
%\medskip \noindent  We now investigate our measurement based procedures 
%\medskip\noindent In this section, we consider our measurement based procedure under the effect of Pauli errors. We prove that errors only accumulate but not propagate, during the procedure, as claimed in the main text just after Fig. \ref{fig:qpolarprep_N8_i3}. We further describe our error detection mechanism for estimating frozen states in detail.

%Further, we consider our error detection mechanism, and a non trivial error with zero syndrome.

\smallskip We suppose that we are given noisy versions of $\pone$ code states of length $K/2$ on $\mathcal{S}_1 = \{1, \dots, K/2\}$, and $\mathcal{S}_2 = \{K/2+ 1, \dots, K\}$,  as follows.
\begin{align}
\ket{q_{\frac{K}{2}}^1}_{\mathcal{S}_1} &= X^{\bm{e}_X^1} Z^{\bm{e}_Z^1}Q_{\frac{K}{2}} \ket{\bm{u_1}, \oline{\bm{v}}_1}_{\mathcal{S}_1}, \label{eq:q-N/2-1-n} \\
\ket{q_{\frac{K}{2}}^2}_{\mathcal{S}_2} &= X^{\bm{e}_X^2} Z^{\bm{e}_Z^2}Q_{\frac{K}{2}} \ket{\bm{u_2}, \oline{\bm{v}}_2}_{\mathcal{S}_2}, \label{eq:q-N/2-2-n}
\end{align}
where   $\bm{u_1}, \bm{u_2} \in \{0, 1\}^{i(k-1)}$,\ \ $\bm{v_1}, \bm{v_2} \in \{0, 1\}^{\frac{K}{2}-i(k-1)}$, and the errors $\bm{e}_X^1, \bm{e}_Z^1, \bm{e}_X^2, \bm{e}_Z^2 \in \{0, 1\}^\frac{K}{2}$ are unknown.

\subsection*{Case 1: Preparation Using Noisy Pauli $Z \otimes Z$ Measurements.} \label{sec:noisy_case1}

%In this case, we perform transversal noisy Pauli $Z \otimes Z$ measurements on the corresponding qubits of systems $\mathcal{S}_1$ and $\mathcal{S}_2$, prepared in noisy polar code states according to (\ref{eq:q-N/2-1-n}) and (\ref{eq:q-N/2-2-n}), respectively. 

We suppose that the following errors happen due to the component failures during transversal Pauli $Z \otimes Z$ measurements (see also Fig. \ref{fig:cnot_err}, where all errors are given; also recall from the previous section, that $\mathcal{S}_3$ denotes the  ancilla system  needed for the ``phase kickback trick'' implementation of  $Z \otimes Z$ measurements, \emph{i.e.}, system $A$ in Fig.~\ref{fig:mZZ_notation_and_circuit}).
%
%We suppose following errors due to  $\cnot_{\mathcal{S}_1 \to \mathcal{S}_3}$.
\begin{list}{}{\setlength{\labelwidth}{2em}\setlength{\leftmargin}{2em}\setlength{\listparindent}{0em}}
\item[$(1)$] Suppose failures during the intialization of $\mathcal{S}_3$ in the Pauli $Z$ basis cause an $X$ error $\bm{e}_X^I \in \{0, 1\}^{K/2}$ on $\mathcal{S}_3$.

\item[$(2)$] Suppose $\cnot$ failures in $\cnot_{\mathcal{S}_1 \to \mathcal{S}_3}$ cause $X, Z$ type errors $\bm{e}_X^{C_1(1)}, \bm{e}_Z^{C_1(1)} \in \{0, 1\}^{K/2}$, respectively, on $\mathcal{S}_1$, and $X, Z$ type errors $\bm{e}^{C_1(3)}_X, \bm{e}^{C_1(3)}_Z \in \{0, 1\}^{K/2}$, respectively,  on  $\mathcal{S}_3$.

\item[$(3)$] Suppose $\cnot$ failures in $\cnot_{\mathcal{S}_2 \to \mathcal{S}_3}$ cause $X, Z$ type errors $\bm{e}_X^{C_2(2)}, \bm{e}_Z^{C_2(2)} \in \{0, 1\}^{K/2}$, respectively, on $\mathcal{S}_2$, and $X, Z$ type errors $\bm{e}^{C_2(3)}_X, \bm{e}^{C_2(3)}_Z \in \{0, 1\}^{K/2}$, respectively,  on  $\mathcal{S}_3$.

\item[$(4)$]  Suppose failures during the measurement of $\mathcal{S}_3$ in the Pauli $Z$ basis cause an $X$ error $\bm{e}_X^M \in \{0, 1\}^{K/2}$  on $\mathcal{S}_3$.

% caused by the failures during the measurement of qubits in  . 

%Note that the following holds,
%\begin{equation}
%\wt(\bm{e}_X^M) = t_k^{M} \label{eq:I} \\
%\end{equation} 

% causing an $X$ error $\bm{e}_X^M \in \{0, 1\}^{K/2}$ in the measurement outcome.
\end{list}
 %$e'_X$ be the $X$ errors due to the first $\cnot$ on the

\begin{figure}[!b]
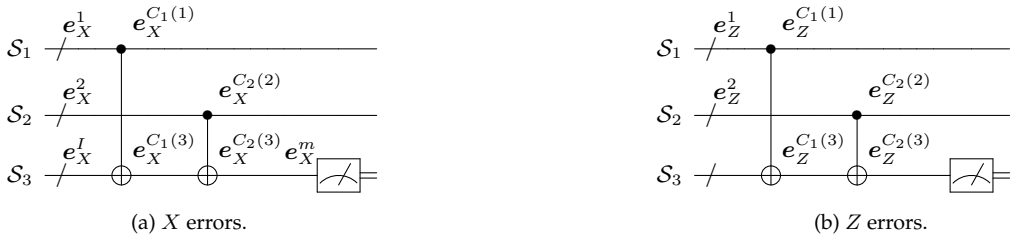

\vspace*{5mm}
\begin{subfigure}[b]{.48\textwidth}
\captionsetup{justification=centering}
\centering
\,\hfill \hspace*{5mm} \input{measure_err}\hfill\,
\caption{$X$ errors.}
\label{fig:cnot_err_x}
\end{subfigure}
~
\begin{subfigure}[b]{.48\textwidth}
\captionsetup{justification=centering}
\centering
\,\hfill \hspace*{-5mm}\input{measurez_err}\hfill\,
\caption{$Z$ errors.}
\label{fig:cnot_err_z}
\end{subfigure}
\caption{The representation of errors during the preparation with Pauli $Z \otimes Z$ measurements. The crossed wires represnt $K/2$-qubit systems $\mathcal{S}_1$, $\mathcal{S}_2$ and $\mathcal{S}_3$, respectively. The $\cnot$ gate between two systems represents the transversal $\cnot$ gates between the systems. The errors are represented where they happen. For example, in the left figure, $\bm{e}_X^1$, $\bm{e}_X^2$, and $\bm{e}_X^I$ represent the initial errors on systems  $\mathcal{S}_1$, $\mathcal{S}_2$ and $\mathcal{S}_3$, respectively. Further,  $\bm{e}_X^{C_1(1)}$ and $\bm{e}_X^{C_1(3)}$ represent the errors caused by the first $\cnot$, applied between $\mathcal{S}_1$ and $\mathcal{S}_3$, respectively. Similarly, the errors due to the second $\cnot$ are represented on systems $\mathcal{S}_2$ and $\mathcal{S}_3$. Finally, $\bm{e}_X^M$ represents the $X$ error due to the measurement on system $\mathcal{S}_3$.}
\label{fig:cnot_err}
\end{figure}

\smallskip The measurement outcome of the transversal Pauli $Z \otimes Z$ measurements and the state of the joint system $\mathcal{S} = \mathcal{S}_1 \cup \mathcal{S}_2$ after measurements is given in Lemma \ref{lem:prep-n}.

\begin{lemma}\label{lem:prep-n}
The measurement outcome of noisy transversal Pauli $Z \otimes Z$ measurements on systems $\mathcal{S}_1$ and $\mathcal{S}_2$, respectively in (\ref{eq:q-N/2-1-n}) and (\ref{eq:q-N/2-2-n}), is equal to
\begin{equation}
 \bm{m}^\prime =   \bm{m} \oplus \bm{e}_X \in \{0, 1\}^{\frac{K}{2}}, \label{eq:m-out-n}
\end{equation}
where $\bm{m}$ is the measurement outcome corresponding to the preparation without noise from (\ref{eq:m-out-1}) and $\bm{e}_X = \bm{e}_X^I \oplus   \bm{e}_X^1  \oplus \bm{e}_X^{C_1(3)} \oplus \bm{e}_X^2 \oplus \bm{e}_X^{C_2(3)} \oplus  \bm{e}_X^M $. Further, after the measurement, we get
\begin{equation}
\ket{q_{K}^\prime}_{\mathcal{S}} = X^{\widetilde{\bm{e}}_X} Z^{\widetilde{\bm{e}}_Z} \ket{q_{K}}_{\mathcal{S}},\label{eq:joint-state-n}
\end{equation}
where $\ket{q_{K}}_{\mathcal{S}}$ is the prepared state without noise from (\ref{eq:noiseles-state-k}) and the errors are as follows,
%
%\footnote{We have previously used the tilde notation for logical code states. Here $\widetilde{\bm{e}}_X = (\bm{e}_X^1, \bm{e}_X^2)$ simply means that $\widetilde{\bm{e}}_X$ is the concatenation of $\bm{e}_X^1$ and $\bm{e}_X^2$ errors, and similarly for $\widetilde{\bm{e}}_Z$.}
%
\begin{align}
\widetilde{\bm{e}}_X &= (\bm{e}_X^1 \oplus \bm{e}^{C_1(1)}_X , \bm{e}_X^2 \oplus \bm{e}_X^{C_2(2)}) \in  \{0, 1\}^{K}. \label{eq:xerr-f} \\  
\widetilde{\bm{e}}_Z &= (\bm{e}_Z^1 \oplus \bm{e}_Z^{C_1(1)} , \bm{e}_Z^2 \oplus \bm{e}^{C_1(3)}_Z \oplus \bm{e}_Z^{C_2(2)} ) \in  \{0, 1\}^{K}. \label{eq:zerr-f}
\end{align}
\end{lemma}

\begin{proof}
Note that  $X$ errors in Fig.~\ref{fig:cnot_err_x} and $Z$ errors in Fig.~\ref{fig:cnot_err_z}  propagate from the LHS to the RHS, as follows: 
\begin{list}{}{\setlength{\labelwidth}{2em}\setlength{\leftmargin}{2em}\setlength{\listparindent}{0em}}
\item[$(1)$] An $X$ error simply passes through the target of a $\cnot$ gate, and it propagates from the control of the $\cnot$ gate to its target. 

\item[$(2)$] An $Z$ error simply passes through the control of a $\cnot$ gate, and it propagates from the target of the $\cnot$ gate to its control. 
\end{list}

\begin{table}[!t]
\centering
\captionsetup{justification=centering}
\caption{Total errors on systems $\mathcal{S}_1$, $\mathcal{S}_2$, and $\mathcal{S}_3$.}
\label{tab:2}
\begin{tabular}{| c | c | c |}
\hline
System & total $X$ error & total $Z$ error \\ 
\hline
$\mathcal{S}_1$ & $\bm{e}_X^1 \oplus \bm{e}_X^{C_1(1)}$ & $\bm{e}_Z^1 \oplus \bm{e}_Z^{C_1(1)}$ \\
\hline
$\mathcal{S}_2$ & $\bm{e}_X^2 \oplus \bm{e}_X^{C_2(2)}$ & $\bm{e}_Z^2  \oplus \bm{e}_Z^{C_1(3)} \oplus \bm{e}_Z^{C_2(2)}$ \\
\hline
$\mathcal{S}_3$ & $ \bm{e}_X^I \oplus   \bm{e}_X^1  \oplus \bm{e}_X^{C_1(3)} \oplus \bm{e}_X^2 \oplus \bm{e}_X^{C_2(3)} \oplus  \bm{e}_X^M  $ & $\bm{e}_Z^{C_1(3)} \oplus \bm{e}_Z^{C_2(3)}$ \\ \hline
\end{tabular}
\end{table}

The $X$ and $Z$ errors on the RHS for $\mathcal{S}_1, \mathcal{S}_2$ and $\mathcal{S}_3$ are given in Table \ref{tab:2}. It follows that the noisy measurement outcome $\bm{m}^\prime$ is equal to the binary sum of the measurement outcome $\bm{m}$ in the noiseless case and the $X$ error on system $\mathcal{S}_3$, therefore, from Table~\ref{tab:2}, we get (\ref{eq:m-out-n}). Further, the prepared state $\ket{q_{K}^\prime}_{\mathcal{S}}$ is equal to the noiseless state $\ket{q_{K}}_{\mathcal{S}}$, with some $X$ and $Z$ errors on it. The $X$ (or $Z$) error on $\mathcal{S}$ is simply the concatenation of the total $X$ (or $Z$) errors on systems $\mathcal{S}_1$ and $\mathcal{S}_2$.  Therefore,  from Table~\ref{tab:2}, we get (\ref{eq:xerr-f}), and (\ref{eq:zerr-f}), for $X$ and $Z$ errors in $\ket{q_{K}^\prime}_{\mathcal{S}}$, respectively.

\end{proof}

We now combine the noisy preparation in Lemma \ref{lem:prep-n} with the error detection gadget according to Procedure $2$ of the main text. We recall here how the error detection gadget works. Consider the error $\bm{e}_X$ in the measurement outcome $\bm{m}^\prime$ in $(\ref{eq:m-out-n})$ so that its syndrome is a zero vector,~$i.e.$,
\begin{equation}
P_{\frac{K}{2}}(\bm{e}_X) \lvert_{\mathcal{Z}(k-1)} = \bm{0}. \label{eq:zero-syn}
\end{equation}
Then the prepared state is accepted and the random vector $\bm{x}$ in (\ref{eq:joint-state-n}) is estimated as 
\begin{equation}
\hat{\bm{x}} = P_{\frac{K}{2}}(\bm{m}^\prime)\lvert_{\mathcal{X}(k-1)}. \label{eq:xval-hat}
\end{equation}
The total error in the prepared state, with respect to the estimate $\hat{\bm{x}}$, is given in Lemma \ref{lem:x-propg}.
\begin{lemma} \label{lem:x-propg}
Consider the error $\bm{e}_X$ in the measurement outcome $\bm{m}$ from (\ref{eq:m-out-n}), so that it satisfies (\ref{eq:zero-syn}), and let $\bm{\hat{x}}$ be defined according to (\ref{eq:xval-hat}).
Then, the state of the joint system $\mathcal{S}$ in (\ref{eq:joint-state-n}) can be written as follows,  with respect to $\bm{\hat{x}}$.
\begin{equation}
\ket{q_{K}^\prime}_{\mathcal{S}}  = X^{\bm{e}_X^f} Z^{\bm{e}_Z^f} Q_K \ket{(\bm{u'}, \bm{\hat{x}}, \bm{u_2}, \oline{\bm{v_1} \oplus \bm{v_2}} )}_{\mathcal{S}}, \label{eq:state-final-zz}
\end{equation}
where $\bm{e}_X^f = \widetilde{\bm{e}}_X + (\bm{e}_X, \bm{0})$ and $\bm{e}_Z^f = \widetilde{\bm{e}}_Z$, so that $\widetilde{\bm{e}}_X$ and $\widetilde{\bm{e}}_Z$ are according to (\ref{eq:xerr-f}) and (\ref{eq:zerr-f}), respectively.

%\item[$(b)$] There exists errors $\bm{e}_X^{\prime f}$, $\bm{e}_Z^{\prime f}$, which are equivalent errors to $\bm{e}_X^{f}$, $\bm{e}_Z^{f}$, respectively, so that they satisfy,
%\begin{align}
%\wt(\bm{e}_X^{\prime f}) &\leq \minn \big( 2 (\wt(\bm{e}_X^2) + t_k^{C_2}(2) ) + A_k, 2 (\wt(\bm{e}_X^1) + t_k^{C_1}(2) ) + A_k \big), \label{eq:final-x-wt}\\
%\wt(\bm{e}_Z^{\prime f}) &\leq \wt(\bm{e}_Z^1)  +  \wt(\bm{e}_Z^2) + t_k^{C_1}(1) +   t_k^{C_2}   \label{eq:final-z-wt}
%\end{align}
%where $A_k = t_X^I + t_k^{C_1}(1) + t_k^{C_2}(1) + t_k^{M}$.

%

\end{lemma}

\begin{proof}

%Consider the preparation, using Pauli $Z \otimes Z$ measurements, according to Lemma \ref{lem:prep-n}. Consider the syndrome of the error term $\bm{e}_X$ in (\ref{eq:m-out-n}) is a zero vector, that is, 
%%
%\begin{equation} 
%P_{\frac{K}{2}}(\bm{e}_X) \lvert_{\mathcal{Z}(k-1)} = P_{\frac{K}{2}}(\bm{m})\lvert_{\mathcal{Z}(k-1)} \,\oplus\, \bm{u'} = \bm{0}. \label{eq:zero_syndrome}
%\end{equation}
%
%As the syndrome of the error $\bm{e}_X$ is a zero vector, no error is detected. Therefore, the prepared state $\ket{q_K}_{\mathcal{S}}$ in (\ref{eq:joint-state-n}) is accepted. Further, the estimate of $\bm{x}$ in (\ref{eq:joint-state-n}) is taken as
%
From (\ref{eq:m-out-1}), (\ref{eq:m-out-n}), and (\ref{eq:xval-hat}), we have that
\begin{equation}
\hat{\bm{x}} = \bm{x} \oplus P_{\frac{K}{2}}(\bm{e}_X)\lvert_{\mathcal{X}(n-1)}. \label{eq:hatx-vs-x}
\end{equation}
We may write $\ket{q_{K}^\prime}_{\mathcal{S}}$ in (\ref{eq:joint-state-n}), with respect to $\bm{\hat{x}}$ as follows,
\begin{align}
\ket{q_{K}^\prime}_{\mathcal{S}} &= X^{\widetilde{\bm{e}}_X} Z^{\widetilde{\bm{e}}_Z} Q_K \ket{(\bm{u'}, \bm{x}, \bm{u_2}, \oline{\bm{v_1} \oplus \bm{v_2}} )} \nonumber\\
&= X^{\widetilde{\bm{e}}_X} Z^{\widetilde{\bm{e}}_Z}  \sum_{\bm{x_2}} (-1)^{(\bm{v_1} \oplus \bm{v_2}) \cdot \bm{x_2}} \ket{P_K(\bm{u'}, \bm{\hat{x}} \oplus P_{\frac{K}{2}}(\bm{e}_X)\lvert_{\mathcal{X}(n-1)} , \bm{u_2}, \bm{x_2} )} \nonumber\\
& = X^{\widetilde{\bm{e}}_X } Z^{\widetilde{\bm{e}}_Z} \sum_{\bm{x_2}} (-1)^{(\bm{v_1} \oplus \bm{v_2}) \cdot \bm{x_2} } X^{(\bm{e}_X, \bm{0})} \ket{P_K(\bm{u'}, \bm{\hat{x}}  , \bm{u_2}, \bm{x_2} )} \nonumber \\
&= X^{\widetilde{\bm{e}}_X \oplus (\bm{e}_X, \bm{0})} Z^{\widetilde{\bm{e}}_Z} Q_K \ket{(\bm{u'}, \bm{\hat{x}}, \bm{u_2}, \oline{\bm{v_1} \oplus \bm{v_2}} )},
\end{align}
where in the second equality, we have expanded the quantum state in the Pauli $Z$ basis and used (\ref{eq:hatx-vs-x}), and in the third equality, we have used  $\bm{e}_X = P_{\frac{K}{2}}(P_{\frac{K}{2}}(\bm{e}_X)\lvert_{\mathcal{Z}(n-1)}, P_{\frac{K}{2}}(\bm{e}_X)\lvert_{\mathcal{X}(n-1)}) = P_{\frac{K}{2}}(\bm{0}, P_{\frac{K}{2}}(\bm{e}_X)\lvert_{\mathcal{X}(n-1)})$.

\end{proof}

\subsection*{Case 2: Preparation Using Noisy Pauli $X \otimes X$ Measurements.}

%In this case, we perform transversal noisy Pauli $X \otimes X$ measurements on the corresponding qubits of systems $\mathcal{S}_1$ and $\mathcal{S}_2$, prepared in noisy polar code states according to (\ref{eq:q-N/2-1-n}) and (\ref{eq:q-N/2-2-n}), respectively. 

The errors during the preparation with Pauli $X \otimes X$ measurements are given in Fig. \ref{fig:cnot_errxx}.  Note that we have $Z$ type initialization and measurement errors on the ancilla system $\mathcal{S}_3$, which are denoted by $\bm{e}_Z^I, \bm{e}_Z^m \in \{0, 1\}^{K/2}$, respectively. For the errors caused by the $\cnot$ gates, $\cnot_{\mathcal{S}_3 \to \mathcal{S}_1}$ and  $\cnot_{\mathcal{S}_3 \to \mathcal{S}_2}$, we have used the same notation as in the case of Pauli $Z \otimes Z$ measurements.

\begin{figure}[!t]
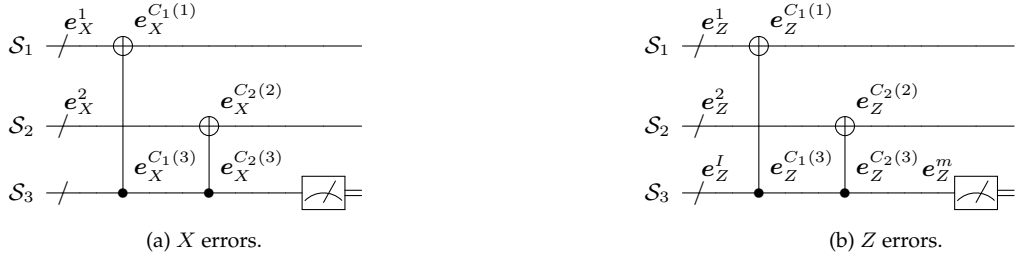

\captionsetup{justification=centering}
\begin{subfigure}[b]{.48\textwidth}
\centering
\,\hfill \input{measurexx_errx}\hfill\,
\caption{$X$ errors.}
\label{fig:cnot_err_xxx}
\end{subfigure}
~
\begin{subfigure}[b]{.48\textwidth}
\centering
\,\hfill \hspace*{-10mm}\input{measurexx_errz}\hfill\,
\caption{$Z$ errors.}
\label{fig:cnot_err_xxz}
\end{subfigure}
\caption{Errors during preparation with Pauli $X \otimes X$ measurements.}
\label{fig:cnot_errxx}
\end{figure}

\smallskip Here, we provide Lemmas \ref{lem:prep-n-x} and \ref{lem:z-propg}, which are analogous to Lemmas \ref{lem:prep-n} and \ref{lem:x-propg}, respectively, and thus their proofs have been omitted.

\begin{lemma}\label{lem:prep-n-x}
The measurement outcome of transversal Pauli $X \otimes X$ measurements on systems $\mathcal{S}_1$ and $\mathcal{S}_2$, respectively in (\ref{eq:q-N/2-1-n}) and (\ref{eq:q-N/2-2-n}), is equal to
\begin{equation}
 \bm{m}^\prime = \bm{m} \oplus \bm{e}_Z \in \{0, 1\}^{\frac{K}{2}}, \label{eq:m-out-n-x}
\end{equation}
where $\bm{m}$ is the measurement outcome  corresponding to the preparation without noise, as in (\ref{eq:m-out-x}), and $\bm{e}_Z = \bm{e}_Z^I \oplus \bm{e}_Z^1 \oplus  \bm{e}_Z^{C_1(3)} \oplus \bm{e}_Z^2 \oplus \bm{e}_Z^{C_2(3)} \oplus \bm{e}_Z^{m} $ . Further, after the measurement, we get,
\begin{equation}
\ket{q_{K}^\prime}_{\mathcal{S}} = X^{\widetilde{\bm{e}}_X} Z^{\widetilde{\bm{e}}_Z} \ket{q_{K}}_{\mathcal{S}},\label{eq:joint-state-n-x}
\end{equation}
where $\ket{q_{K}}_{\mathcal{S}}$ is the prepared state without noise, as in (\ref{eq:joint-state-x}) and  the errors are as follows,
%
%\footnote{We have previously used the tilde notation for logical code states. Here $\widetilde{\bm{e}}_X = (\bm{e}_X^1, \bm{e}_X^2)$ simply means that $\widetilde{\bm{e}}_X$ is the concatenation of $\bm{e}_X^1$ and $\bm{e}_X^2$ errors, and similarly for $\widetilde{\bm{e}}_Z$.}
%
\begin{align}
\tilde{\bm{e}}_X &= (\bm{e}_X^1 \oplus \bm{e}^{C_1(1)}_X , \bm{e}_X^2 \oplus \bm{e}^{C_1(3)}_X  \oplus \bm{e}_X^{C_2(2)}) \in  \{0, 1\}^{K}. \label{eq:xerr-f-1} \\  
\tilde{\bm{e}}_Z &= (\bm{e}_Z^1 \oplus \bm{e}_Z^{C_1(1)} , \bm{e}_Z^2 \oplus \bm{e}_Z^{C_2(2)}  ) \in  \{0, 1\}^{K}. \label{eq:zerr-f-1}
\end{align}
\end{lemma}

\medskip We now combine the noisy preparation in Lemma \ref{lem:prep-n-x} with the error detection gadget according to Procedure $2$ from the main text.  Recall that when no error is detected, the prepared state is accepted and the $\bm{z}$ in (\ref{eq:joint-state-n-x}) is estimated as 
\begin{equation}
\bm{\hat{z}} = P_{\frac{K}{2}}^\top (\bm{m}^\prime)\lvert_{\mathcal{Z}(k-1)}. \label{eq:zval-hat}
\end{equation}

\begin{lemma} \label{lem:z-propg}
Consider the error $\bm{e}_Z$ in the measurement outcome $\bm{m}$ from (\ref{eq:m-out-n-x}), so that its syndrome is a zero vector, $i.e.$, $P_{\frac{N}{2}}^\top(\bm{e}_Z) \lvert_{\mathcal{X}(n-1)} = 0$, and let $\bm{\hat{z}}$ be defined as in (\ref{eq:zval-hat}).
Then, the state of the joint system $\mathcal{S}$ in (\ref{eq:joint-state-n-x}) can be written as follows,  with respect to $\bm{\hat{z}}$.
\begin{equation}
\ket{q_{K}^\prime}_\mathcal{S}  = X^{\bm{e}_X^f} Z^{\widetilde{\bm{e}}_Z^f} Q_K \ket{\bm{u_1} \oplus \bm{u_2}, \oline{\bm{v_1}, \bm{\hat{z}}, \bm{v'}}}_\mathcal{S}, \label{eq:state-final-xx}
\end{equation}
where $\bm{e}_X^f = \widetilde{\bm{e}}_X$ and $\bm{e}_Z^f = \widetilde{\bm{e}}_Z \oplus (\bm{0}, \bm{e}_Z)$, so that $\widetilde{\bm{e}}_X$ and $\widetilde{\bm{e}}_Z$ are according to (\ref{eq:xerr-f}) and (\ref{eq:zerr-f}), respectively, and $\bm{e}_Z$ is the measurement error in (\ref{eq:m-out-n-x}).
\end{lemma}

\subsection*{Proof of Theorem $3$ From the Main Text}
We provide a proof of Theorem $3$, using the mathematical induction. The base case, $i.e.$, the zeroth level of recursion, corresponds to the initialization of $N$ qubits in a Pauli $Z$ basis state. Let $T_0$ be the number of failures in initialization. As a failure in initialization produces only a single qubit $X$ error, therefore, we have that $\wt(\bm{e}^f_X) = T_0$, for the error $\bm{e}^f_X \in \{0, 1\}^N$ after the zeroth level of recursion. Hence, Theorem $3$ holds for the zeroth level of recursion.

\smallskip \noindent We show in Lemma \ref{lem:fault-up} below that if Theorem $3$ is true for the $(k-1)^\text{th}, 1 \leq k \leq n$, level of recursion, it remains true for the $k^\text{th}$ level of recursion, therefore, implying that Theorem $3$ holds for any $1 \leq k \leq n$.

\bigskip Recall that at the $k^\text{th}$ level of recursion, we prepare $2^{n-k}$ copies of $\ket{q_K}$, where each copy is prepared by applying transversal two-qubit Pauli measurements on two copies of $\ket{q_{K/2}}$. It is enough to show that Theorem $3$ holds for a given preparation instance of $\ket{q_K}$. Hence, we consider $\ket{q_{K/2}}_{\mathcal{S}_1}$ and $\ket{q_{K/2}}_{\mathcal{S}_2}$ from (\ref{eq:q-N/2-1-n}) and (\ref{eq:q-N/2-2-n}), and for the sake of brevity, we suppose that Pauli $Z \otimes Z$ measurements are applied at the $k^\text{th}$ level of recursion. Similar results to the ones that will be proven here can be obtained for the case of Pauli $X \otimes X$ measurements.

\smallskip  We suppose that $t_k$ component failures happen during the implementation of transversal Pauli $Z \otimes Z$  measurements on $\mathcal{S}_1$ and $\mathcal{S}_2$, using an ancilla system $\mathcal{S}_3$ (see also Fig. \ref{fig:cnot_err}).  We decompose $t_k$ as follows,
\begin{equation}
t_k = t_k^I + t_k^{C_1} + t_k^{C_2} + t_k^{M}, \label{eq:fail-k}
\end{equation}
where $t_k^I$ is the number of failures during the initialization of $\mathcal{S}_3$ in the Pauli $Z$ basis, $t_k^{C_i}$ for $i \in \{1, 2\}$ is the number of $\cnot$ failures in $\cnot_{\mathcal{S}_i \to \mathcal{S}_3}$, and $t_k^{M}$ is the number of failures during the Pauli $Z$ measurement on $\mathcal{S}_3$. To relate with Fig. \ref{fig:cnot_err}, note that the errors $\bm{e}_X^I$ and $\bm{e}_X^M$ are produced by $t_k^I$ and $t_k^{M}$-faults, respectively. Further, the errors $\bm{e}_X^{C_1(1)}, \bm{e}_X^{C_1(3)}, \bm{e}_Z^{C_1(1)}, \bm{e}_Z^{C_1(3)}$ are produced by $t_k^{C_1}$-faults, and finally  the errors $\bm{e}_X^{C_2(2)}, \bm{e}_X^{C_2(3)}, \bm{e}_Z^{C_2(2)}, \bm{e}_Z^{C_2(3)}$ are produced by $t_k^{C_2}$-faults.

\smallskip  We will also need the following notation:
\begin{list}{}{\setlength{\labelwidth}{2em}\setlength{\leftmargin}{2em}\setlength{\listparindent}{0em}}
\vspace*{-2mm}
\item[$(i)$] For $i \in \{1, 2\}$, let $t_k^{C_i}(X)$ be the number of $\cnot$ failures in $\cnot_{\mathcal{S}_i \to \mathcal{S}_3}$ that produce a single qubit $X$ error on the corresponding outputs, $i.e.$, the number of $\cnot$ failures, where one of the following errors $X \otimes I, X \otimes Z, I \otimes X, Z \otimes X, Y \otimes I, Y \otimes I, Y \otimes Z, I \otimes Y, Z \otimes Y$ is produced. Similarly, let $t_k^{C_i}(Z)$ be the number of $\cnot$ failures in $\cnot_{\mathcal{S}_i \to \mathcal{S}_3}$ that produce a single qubit $Z$ error.

\item[$(ii)$]  For $i \in \{1, 2\}$, let $t_k^{C_i}(XX)$ be the number of $\cnot$ failures in  $\cnot_{\mathcal{S}_i \to \mathcal{S}_3}$ that produce a two qubit $X$ error, $i.e.$,  the number of $\cnot$ failures, where one of the following errors $X \otimes X, X \otimes Y, Y \otimes X, Y \otimes Y$ is produced. Similarly, let $t_k^{C_i}(ZZ)$ be the number of $\cnot$ failures in $\cnot_{\mathcal{S}_i \to \mathcal{S}_3}$ that produce a two qubit $Z$ error.
\end{list}
Note that the following inequalities hold trivially for $ i \in \{1,2\}$,
\begin{align}
t_k^{C_i}(X) + t_k^{C_i}(XX) &\leq t_k^{C_i}. \label{eq:fault} \\
t_k^{C_i}(Z) + t_k^{C_i}(ZZ) &\leq t_k^{C_i}. \label{eq:fault-z} 
\end{align}
%Further, we get the following relations between the number of faults and the errors produced by them,
%%
%\begin{align}
%\wt(\bm{e}_X^I) &= t^I  \\
%\wt(\bm{e}_X^M) &= t^m  \\
%\wt(\bm{e}^{C_i(i)}_X \oplus \bm{e}^{C_i(3)}_X ) &= t^{C_i}_k(X), i \in  \{1, 2\} \label{eq:ci} \\
%\wt(\bm{e}^{C_i(i)}_Z \oplus \bm{e}^{C_i(3)}_Z ) &= t^{C_i}_k(Z), i \in  \{1, 2\} \label{eq:ci-z} \\
%\wt(\bm{e}^{C_i(i)}_X) \oplus \wt(\bm{e}^{C_i(3)}_X) &= t^{C_i}_k(X) + 2t^{C_i}_k(XX), i \in  \{1, 2\}  \label{eq:ci-1} \\
%\wt(\bm{e}^{C_i(i)}_Z) \oplus \wt(\bm{e}^{C_i(3)}_Z) &= t^{C_i}_k(Z) + 2t^{C_i}_k(ZZ), i \in  \{1, 2\}  \label{eq:ci-1z}
%\end{align}
%
%\smallskip In Lemma \ref{lem:fault-up}, we upper bound the errors in the final prepared state in (\ref{eq:state-final-zz}) by the number of component failures.
%
%\medskip For the sake of clarity, we suppose that Pauli $Z \otimes Z$ measurements are applied at the $k^{th}$ level of recursion. The similar results can easily be obtained for the case of Pauli $X \otimes X$ measurements. 

In Lemma \ref{lem:inequal-fault}, we provide several inequalities, connecting the number of faults with the weight of the produced errors (see Fig. \ref{fig:cnot_err}). We will use these equalities later in the proof of Lemma \ref{lem:fault-up}.

\begin{lemma} \label{lem:inequal-fault}
The following inequalities hold at the $k^\text{th}$ level of recursion,
\begin{align}
\wt(\bm{e}_X^I) &= t_k^I.  \label{eq:I-x}\\
\wt(\bm{e}_X^M) &= t_k^M.  \label{eq:I-m}\\
\wt(\bm{e}^{C_i(i)}_X \oplus \bm{e}^{C_i(3)}_X ) &= t^{C_i}_k(X), i \in  \{1, 2\}. \label{eq:ci} \\
\wt(\bm{e}^{C_i(i)}_X) + \wt(\bm{e}^{C_i(3)}_X) &= t^{C_i}_k(X) + 2t^{C_i}_k(XX), i \in  \{1, 2\}.  \label{eq:ci-1} \\
\wt(\bm{e}^{C_i(i)}_Z \oplus \bm{e}^{C_i(3)}_Z ) &= t^{C_i}_k(Z), i \in  \{1, 2\}. \label{eq:ci-z} \\
\wt(\bm{e}^{C_i(i)}_Z) + \wt(\bm{e}^{C_i(3)}_Z) &= t^{C_i}_k(Z) + 2t^{C_i}_k(ZZ), i \in  \{1, 2\}.  \label{eq:ci-1z}
\end{align}
\end{lemma} 

\begin{proof}
The equalities in (\ref{eq:I-x}) and (\ref{eq:I-m}) simply follow from the fact a failure in Pauli $Z$ basis initialization or measurement produces a single qubit $X$ error on the output.

%\smallskip Further, the $j^{th}, j \in \{1, \dots, K/2\}$-component of the error vector $\bm{e}^{C_i(i)}_X \oplus \bm{e}^{C_i(3)}_X \in \{0, 1\}^{K/2}$ is equal to one, if and only if the the $j^{th}$-$\cnot$ gate in $\cnot_{\mathcal{S}_i \to \mathcal{S}_3}, i \in \{1,2\}$ produces a single qubit error on the output. Hence, (\ref{eq:ci}) holds. Similarly, it can be seen that (\ref{eq:ci-z}) holds.  

\smallskip Suppose the $j^\text{th}, j \in \{1, \dots, K/2\}$-$\cnot$ gate in $\cnot_{\mathcal{S}_i \to \mathcal{S}_3}$ failed. If it produces a single qubit $X$ error on the output, the $j^\text{th}$ component of either $\bm{e}^{C_i(i)}_X$ or $\bm{e}^{C_i(3)}_X$ is equal to $1$. If it produces a two qubit $X$ error on the output, the $j^\text{th}$ component of both $\bm{e}^{C_i(i)}_X$ and $\bm{e}^{C_i(3)}_X$ is equal to $1$. This observation directly implies (\ref{eq:ci}) and (\ref{eq:ci-1}). Similarly, it can be seen that (\ref{eq:ci-z}) and (\ref{eq:ci-1z}) hold.

%Hence, the weight of $\bm{e}^{C_i(i)}_X \oplus \bm{e}^{C_i(3)}_X$ is equal to the instances, where a single qubit $X$ errors occur, which corresponds to (\ref{eq:ci}). Further, $\wt(\bm{e}^{C_i(i)}_X) + \wt(\bm{e}^{C_i(3)}_X)$

\end{proof}

\begin{lemma} \label{lem:fault-up}
Suppose that  $T_{k-1}^1$ faults occur during the (successful) preparation of $\ket{q_{K/2}}_{\mathcal{S}_1}$ in (\ref{eq:q-N/2-1-n}) and  $T_{k-2}^2$ faults occur during the preparation of $\ket{q_{K/2}}_{\mathcal{S}_2}$ in (\ref{eq:q-N/2-2-n}), using the measurement based procedure incorporated with the error detection.  Further, suppose that the following holds for $i \in \{1,2\}$
\begin{align}
\wt(\bm{e}_X^i) \leq T_{k-1}^i. \label{eq:k/2-fail-1} \\
\wt(\bm{e}_Z^i) \leq T_{k-1}^i. \label{eq:k/2-fail-2}
\end{align}
Consider the prepared state $\ket{q_K^\prime}$ according to (\ref{eq:state-final-zz}). Then, there exist equivalent errors $\bm{e}_X^{\prime f} \equiv \bm{e}_X^{f}$, $\bm{e}_Z^{\prime f} \equiv \bm{e}_Z^{f}$, satisfying,
\begin{align}
\wt(\bm{e}_X^{\prime f}) &\leq T_k, \label{eq:final-x-wt} \\
\wt(\bm{e}_Z^{\prime f}) &\leq T_k, \label{eq:final-z-wt}
\end{align}
where $T_k = T_{k-1}^1 + T_{k-1}^2 + t_k$, where $t_k$ is from (\ref{eq:fail-k}), is the total number of failures after the $k^\text{th}$ level of recursion.
\end{lemma}

\begin{proof}

%We only provide proofs of (\ref{eq:final-x-wt}) and (\ref{eq:final-z-wt}) for $\ket{q_k}$ in (\ref{eq:state-final-zz}), as the proofs for $\ket{q_k}$ in (\ref{eq:state-final-xx}) can be done similarly.

\smallskip We first prove (\ref{eq:final-x-wt}).  Let $\bm{e}_X$ be the measurement error in (\ref{eq:m-out-n}), so that its syndrome is zero according to (\ref{eq:zero-syn}). Then the syndrome of  $(\bm{e}_X, \bm{e}_X)$ is also zero, that is, $P_N(\bm{e}_X, \bm{e}_X) \lvert_{\mathcal{Z}(k)} = \bm{0}$. It can be seen as follows.
\noindent We have that
\begin{align}
(\bm{e}_X, \bm{e}_X) &= ( P_{K/2}(\bm{0}, P_{\frac{K}{2}}(\bm{e}_X)\lvert_{\mathcal{X}(n-1)}), P_{K/2}(\bm{0}, P_{\frac{K}{2}}(\bm{e}_X)\lvert_{\mathcal{X}(n-1)})) \\
&= P_{K}(\bm{0}, \bm{0}, \bm{0}, P_{\frac{K}{2}}(\bm{e}_X)\lvert_{\mathcal{X}(n-1)}) \label{eq:int-lemma}
\end{align}
Using $\lvert \mathcal{X}(n-1) \lvert = \lvert \mathcal{X}(n) \lvert $ and (\ref{eq:int-lemma}), it follows that  $P_N(\bm{e}_X, \bm{e}_X) \lvert_{\mathcal{Z}(k)} = \bm{0}$. Therefore, $(\bm{e}_X, \bm{e}_X)$ gives a stabilizer (up to a sign factor) of the quantum state $\ket{q_K^\prime}$ in (\ref{eq:state-final-zz}), implying that $\bm{e}_X^{\prime f}  \eqdef \widetilde{\bm{e}}_X + (\bm{0}, \bm{e}_X) $ is an equivalent error to $\bm{e}_X^f = \widetilde{\bm{e}}_X + (\bm{e}_X, \bm{0})$ in (\ref{eq:state-final-zz}).

\medskip We now provide upper bounds on $\wt(\bm{e}_X^f)$ and $\wt(\bm{e}_X^{\prime f})$.  

\smallskip \noindent From (\ref{eq:m-out-n}) and (\ref{eq:xerr-f}), we have that
\begin{align}
\bm{e}_X^f & = \big( \bm{e}^2_X \oplus \bm{e}^{I}_X \oplus   \bm{e}^{C_1(1)}_X \oplus \bm{e}^{C_1(3)}_X \oplus  \bm{e}^{C_2(3)}_X \oplus  \bm{e}^{M}_X, \bm{e}^2_X \oplus  \bm{e}^{C_2(2)}_X \big). \label{eq:x-f-err}
\end{align}
% 

%\smallskip \noindent Let $t_k^{C_1}(X)$ be the number of $\cnot$ failures in $\cnot_{\mathcal{S}_1 \to \mathcal{S}_3}$ that produce a single qubit $X$ error on the corresponding output and let $t_k^{C_1}(XX)$ be the failures that produce a two qubit $X$ error. Note that  $t_k^{C_1}(X) + t_k^{C_1}(XX) \leq t_k^{C_1}$. Similarly, let $t_k^{C_2}(X)$ be the number of $\cnot$ failures in $\cnot_{\mathcal{S}_2 \to \mathcal{S}_3}$ that produce a single qubit $X$ error on the corresponding output and let $t_k^{C_2}(XX)$ be the number of failures that produce a two qubit $X$ error.
%
%
%\smallskip \noindent Note that the following holds.
%\begin{align}
%\wt(\bm{e}^{I}_X) &= t^I_k \label{eq:I}\\
%\wt(\bm{e}^{C_1(1)}_X \oplus \bm{e}^{C_1(3)}_X ) &= t^{C_1}_k(X) \label{eq:c1}\\
%\wt(\bm{e}^{C_2(2)}_X) \oplus \wt(\bm{e}^{C_2(3)}_X) &= t^{C_2}_k(X) + 2t^{C_2}_k(XX)  \label{eq:c2}\\
%\wt(\bm{e}^{I}_X) &= t_k^{M} \label{eq:m} \\
%t_k^{C_1}(X) + t_k^{C_1}(XX) &\leq t_k^{C_1} \label{eq:fault-1}\\
%t_k^{C_2}(X) + t_k^{C_2}(XX) &\leq t_k^{C_2} \label{eq:fault-2} 
%\end{align}
%
\noindent Further,
\begin{align}
\wt(\bm{e}_X^f) & \leq 2 \wt(\bm{e}^2_X) + \wt(\bm{e}^{I}_X) + \wt(\bm{e}^{C_1(1)}_X \oplus \bm{e}^{C_1(3)}_X ) + \wt(\bm{e}^{C_2(2)}_X) \oplus \wt(\bm{e}^{C_2(3)}_X) + \wt(\bm{e}^{I}_X) \\
 & \leq 2 \big(\wt(\bm{e}^2_X) + t^{C_2}_k(XX)\big) + t^I_k  + t^{C_1}_k(X) +  t^{C_2}_k(X) + t_k^{M}, \label{eq:wt-xerr-1}
\end{align}
where the first equality follows from (\ref{eq:x-f-err}) and using $\wt(\sum_i \bm{u}_i) \leq \sum_i \wt(\bm{u})_i$, and the second equality follows from (\ref{eq:ci})-(\ref{eq:ci-1}).
Similarly, it can be shown that
\begin{equation}
\wt(\bm{e}_X^{\prime f}) \leq 2 (\wt(\bm{e}^1_X) + t^{C_1}_k(XX)) + t^I_k  + t^{C_1}_k(X) +  t^{C_2}_k(X) + t_k^{M} \label{eq:wt-xerr-2}
\end{equation}
Without loss of generality, we may assume that $\wt(\bm{e}^1_X) + t^{C_1}_k(XX)  \leq \wt(\bm{e}^2_X) + t^{C_2}_k(XX)$. Therefore, we have that
\begin{align}
\wt(\bm{e}_X^{\prime f}) & \leq \wt(\bm{e}^1_X) + t^{C_1}_k(XX) + \wt(\bm{e}^2_X) + t^{C_2}_k(XX) +  t^I_k  + t^{C_1}_k(X) +  t^{C_2}_k(X) + t_k^{M} \\
& \leq  \wt(\bm{e}^1_X) + \wt(\bm{e}^2_X) + t^I_k  + (t^{C_1}_k(X)  + t^{C_1}_k(XX)) +  ( t^{C_2}_k(X) +  t^{C_2}_k(XX)) + t_k^{M} \\
& \leq  T_{k-1}^1 + T_{k-1}^2  +  t^I_k + t_k^{C_1} + t_k^{C_2} + t_k^{M} \\
& \leq T_{k-1}^1 + T_{k-1}^2 + t_k, 
\end{align}
where the third equality follows from (\ref{eq:k/2-fail-1}) and (\ref{eq:fault}), and the fourth equality follows from (\ref{eq:fail-k}).
%Finally, (\ref{eq:final-x-wt}) follows from (\ref{eq:wt-xerr-1}) and (\ref{eq:wt-xerr-2}).
% and let $t_k^{C_1}(ZZ)$ be the number of times they produce a two qubit $Z$ error. Note that  $t_k^{C_1}(Z) + t_k^{C_1}(ZZ) \leq t_k^{C_1}$. Similarly let $t_k^{C_2}(Z)$ be the number of times $\cnot$ failures in $\cnot_{\mathcal{S}_2 \to \mathcal{S}_3}$ produce a single qubit $Z$ error on the corresponding output and let $t_k^{C_2}(ZZ)$ be the number of times they produce a two qubit $Z$ error.
%

\medskip We now prove (\ref{eq:final-z-wt}). Using (\ref{eq:zerr-f}), we get for $\bm{e}_Z^f$ in (\ref{eq:state-final-zz}),
\begin{equation}
\bm{e}_Z^f = (\bm{e}_Z^1 \oplus \bm{e}_Z^{C_1(1)} , \bm{e}_Z^2 \oplus  \bm{e}^{C_1(3)}_Z  \oplus \bm{e}_Z^{C_2(2)}). \label{eq:z-f-err}
\end{equation}
%Note from (\ref{eq:z-f-err}) that a $\cnot$ failure in $\cnot_{\mathcal{S}_1 \to \mathcal{S}_3}$ that causes two qubit $Z$ error on the ouput, produces a stabilizer final error on the prepared state, equal to the corresponding measured Pauli $Z \otimes Z$. Therefore, we only need to consider the $\cnot$ failures in $\cnot_{\mathcal{S}_1 \to \mathcal{S}_3}$, which cause only one Pauli $Z$ errors on the output.  This implies that the following error is equivalent to $\bm{e}_Z^f$ in (\ref{eq:z-f-err}),
%
%\begin{equation}
%\bm{e}_Z^{\prime f} = (\bm{e}_Z^1 \oplus \bm{e}_Z^{C_1(1)} \oplus \bm{e}_Z^{C_1(1)} \odot \bm{e}_Z^{C_1(3)}  , \bm{e}_Z^2 \oplus \bm{e}^{C_1(3)}_Z \oplus  \bm{e}_Z^{C_1(1)} \odot \bm{e}_Z^{C_1(3)}  \oplus \bm{e}_Z^{C_2(2)}), \label{eq:z-f-eq}
%\end{equation}
%%
As $Z \otimes Z$ on the corresponding qubits of $\mathcal{S}_1$ and $\mathcal{S}_2$ is a stabilizer generator, it follows that the following error is equivalent to $\bm{e}_Z^f$ in (\ref{eq:z-f-err}),
\begin{align}
\bm{e}_Z^{\prime f} = \bm{e}_Z^{f} \oplus (\bm{e}_Z^{C_1(3)}, \bm{e}_Z^{C_1(3)})= (\bm{e}_Z^1 \oplus \bm{e}_Z^{C_1(1)} \oplus \bm{e}_Z^{C_1(3)}  , \bm{e}_Z^2  \oplus \bm{e}_Z^{C_2(2)}), \label{eq:z-f-eq}
\end{align}
%where $\bm{e}_Z^{C_1(1)} \odot \bm{e}_Z^{C_1(3)}$ denotes the Hadamard product of $\bm{e}_Z^{C_1(1)}$ and $\bm{e}_Z^{C_1(3)}$. Note that the following holds.
%\begin{align}
%\wt(\bm{e}_Z^{C_1(1)} \oplus \bm{e}_Z^{C_1(1)} \odot \bm{e}_Z^{C_1(3)}) + \wt(\bm{e}_Z^{C_1(3)} \oplus \bm{e}_Z^{C_1(Z)} \odot \bm{e}_Z^{C_1(3)} ) = t_k^{C_1}(Z) \label{eq:1-zfault}\\
%\wt(\bm{e}_Z^{C_2(2)}) \leq t_k^{C_2}, \label{eq:2-zfault}
%\end{align}
%where $t_k^{C_1}(Z)$ denotes the number of times $\cnot$ failures in $\cnot_{\mathcal{S}_1 \to \mathcal{S}_3}$ produce a single qubit $Z$ error on the corresponding output.
%From (\ref{eq:z-f-eq})-(\ref{eq:2-zfault}) and using $\wt(\sum_i \bm{u}_i) \leq \sum_i \wt(\bm{u})_i$, we have that 
%%
We have, 
\begin{align}
\wt(\bm{e}_Z^{\prime f}) &\leq \wt(\bm{e}_Z^1)  +  \wt(\bm{e}_Z^2) + \wt(\bm{e}_Z^{C_1(1)} \oplus \bm{e}_Z^{C_1(3)}) + \wt(\bm{e}_Z^{C_2(2)})   \\
& \leq  T_{k-1}^1 + T_{k-1}^2 + t_k^{C_1}(Z) +   t_k^{C_2}(Z) +  2t_k^{C_2}(ZZ) \\
& \leq  T_{k-1}^1 + T_{k-1}^2 + t_k^{C_1} +  t_k^{C_2} \\
& \leq T_{k-1}^1 + T_{k-1}^2 + t_k,
\end{align}
where the first inequality follows from (\ref{eq:z-f-eq}) and using $\wt(\sum_i \bm{u}_i) \leq \sum_i \wt(\bm{u})_i$, the second inequality follows from (\ref{eq:k/2-fail-2}), and (\ref{eq:ci-z})-(\ref{eq:ci-1z}), the third inequality follows from (\ref{eq:fault-z}), and finally the fourth inequality follows from (\ref{eq:fail-k}).

\end{proof}

%

%\smallskip In Lemma \ref{lem:fault-up-x}, we upper bound the errors in the final prepared state in (\ref{eq:state-final-xx}) by the number of component failures.
%
%\begin{lemma} \label{lem:fault-up-x}
%Suppose that  $T_{k-1}^1$ faults occur during the (successful) preparation of $\ket{q_{K/2}}_{\mathcal{S}_1}$ in (\ref{eq:q-N/2-1-n}) and  $T_{k-2}^2$ faults occur during the preparation of $\ket{q_{K/2}}_{\mathcal{S}_2}$ in (\ref{eq:q-N/2-2-n}).  Further, suppose that the errors in $\ket{q_{K/2}}_{\mathcal{S}_1}$  and $\ket{q_{K/2}}_{\mathcal{S}_2}$, respectively, satisfy the following inequalities for $i \in \{1, 2\}$.
%%
%\begin{align}
%\wt(\bm{e}_X^i) \leq T_{k-1}^i \label{eq:k/2-fail-1-x} \\
%\wt(\bm{e}_Z^i) \leq T_{k-1}^i \label{eq:k/2-fail-2-x}
%\end{align}
%Then, there  exists errors $\bm{e}_X^{\prime f}$, $\bm{e}_Z^{\prime f}$, which are equivalent errors to $\bm{e}_X^{f}$, $\bm{e}_Z^{f}$ in (\ref{eq:state-final-xx}), respectively, so that they satisfy,
%%
%\begin{align}
%\wt(\bm{e}_X^{\prime f}) &\leq T_k, \label{eq:final-x-wt-x} \\
%\wt(\bm{e}_Z^{\prime f}) &\leq T_k, \label{eq:final-z-wt-x}
%\end{align}
%where $T_k = T_{k-1}^1 + T_{k-1}^2 + t_k$, with $t_k$ according to (\ref{eq:fail-k}), is the total number of failures after $k^{th}$ level of recursion.
%\end{lemma}

%

\section{Numerical Results on Fault Tolerant Error Correction}
\makeatletter
\def\@currentlabel{\thesubsection}
\makeatother
 \label{sec:num_sim-prep}

%We now consider our measurement based preparation under the effect of noise. We suppose that $\cnot$ gate and single qubit measurements on physical qubits are faulty, and we adapt the following noise models, which are frequently used in the literature~\cite{}. 

In this section, we provide details about the methods used to generate numerical results regarding the logical error rates of $\pone$ codes, under Steane error correction in conjunction with the proposed measurement based preparation with error detection. In this context, we also provide new simulation results on the rate of successful preparation of $\pone$ code states, for the measurement based preparation with error detection. 

\smallskip We start by presenting first the noise model used to produce errors during the simulation of the measurement-based preparation procedure and the Steane error-correction scheme. Note that we consider the implementation of Pauli $Z \otimes Z$ and Pauli $X \otimes X$ measurements according to circuits in Fig.~\ref{fig:mZZ_notation_and_circuit} and Fig.~\ref{fig:mXX_notation_and_circuit}, respectively. However, for the Pauli $X \otimes X$ measurement in Fig.~\ref{fig:mXX_notation_and_circuit}, we consider the initialization in  Pauli $Z$ basis followed by the Hadamard gate as one operation, corresponding to the initialization in Pauli $X$ basis, and similarly, we consider the last Hadamard gate followed by  the Pauli $X$ measurement as one operation, corresponding to a Pauli $X$ measurement. 

%remove both Hadamard gates on the ancilla qubit $A$, by considering the initialization in a Pauli $X$ basis state, and then Pauli $X$ measurement in the end (thus, .
% Finally, we suppose that the initialization in Pauli $Z$ and $X$ basis states is also done by performing Pauli $Z$ and Pauli $X$ measurements, respectively.

% for initialization of qubits in a Pauli $Z$ basis state. 

\subsection{Noise Model} 
\makeatletter
\def\@currentlabel{\thesubsection}
\makeatother
\label{sec:noise-mod}

We only need noise models for the basic components of the procedure, $i.e$, initialization operations and  single-qubit measurements, in either $Z$ or $X$ basis, and $\cnot$ gates. We consider the following types of errors, corresponding to the circuit based depolarizing noise model from~\cite{fowler2012surface}. The probability parameter $p$ below, is referred to as the \emph{physical error rate}.

\begin{itemize}
\item[$(1)$] The noisy initialization in Pauli $Z$ basis is equal to perfectly initializing a qubit in a Pauli $Z$ basis state, then applying a Pauli $X$ error on the qubit, with probability $p$. 
Similarly, the noisy initialization in Pauli $X$ basis is equal to perfectly initializing a qubit in a Pauli $X$ basis state, then applying a Pauli $Z$ error on the qubit, with probability $p$.

\item[$(2)$] The noisy Pauli $Z$ measurement is equal to first applying a Pauli $X$ error, with probability $p$, on the qubit we want to measure, and then applying  the perfect Pauli $Z$ measurement. Similarly, the noisy Pauli $X$ measurement is equal to first applying a Pauli $Z$ error, with probability $p$, and then applying  the perfect Pauli $X$ measurement.

\item[$(3)$] The noisy $\cnot$ gate is equal to the perfect $\cnot$ followed by a two-qubit depolarizing channel, with error probability $p$. Precisely, after the perfect $\cnot$,  any one of the $15$ two qubit Pauli errors $ I \otimes X, I\otimes Y , I\otimes Z, X\otimes I, X\otimes X, X\otimes Y, X\otimes Z, Y\otimes I, Y\otimes X, Y\otimes Y , Y\otimes Z,  Z\otimes I, Z\otimes X, Z\otimes Y , Z\otimes Z$, may occur with probability $\frac{p}{15}$.

\end{itemize}

%\medskip \noindent In the following, we refer the failure probability $p$ of the Pauli $Z$ and $X$ measurements, and the $\cnot$ gate as the physical error rate.

\begin{figure}[!t] 
\centering
\,\hfill\hspace*{2em}\input{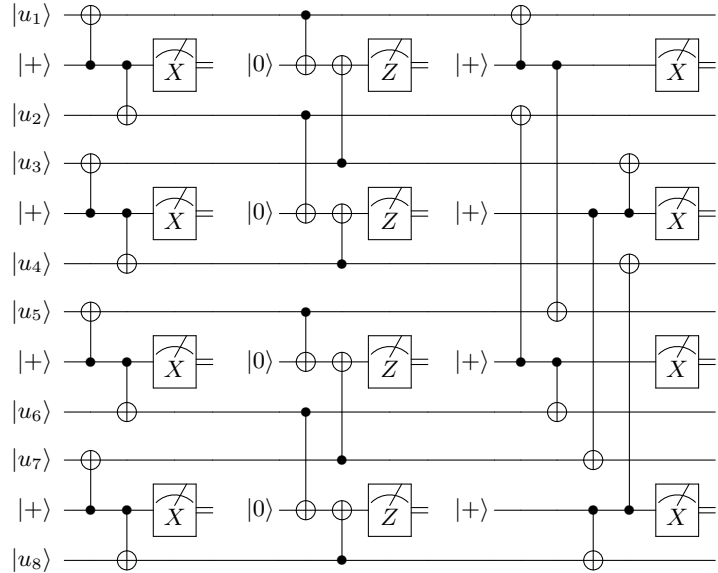}\hfill\,
\caption{Measurement based preparation for the $\pone(N = 8, i(n) = 3)$, with Pauli $Z \otimes Z$ and Pauli $X \otimes X$ measurements implemented according to Fig.~\ref{fig:mZZ_notation_and_circuit} and Fig.~\ref{fig:mXX_notation_and_circuit}. Initialization operations are shown as ket states. For single-qubit measurement operations, the Pauli basis is also indicated.  Initialization operations, measurements, and $\cnot$ gates are noisy, according to the assumed circuit level depolarizing noise model. Note that errors are generated only by the above noisy operations (we do not consider errors that might occur while the qubits are idle). Error detection is performed after each recursion level (one round of measurements of the four ancilla qubits), and the procedure is restarted from the beginning if an error is detected. Error detection after the first level of recursion is useless, hence it is not performed.}
\label{fig:qpolarprep_N8_i3_expanded}
\end{figure}

\subsection{Simulation of the Preparation Procedure} 
\makeatletter
\def\@currentlabel{\thesubsection}
\makeatother
\label{subsec:prep-sim}
Consider the $\pone$ code state $\ket{q_N}_\mathcal{S}$ on the $N =2^n$ qubit system $\mathcal{S} = \{1, \dots, N\}$. To prepare $\ket{q_N}_\mathcal{S}$, we first initialize the $N$ qubits in a Pauli $Z$ basis state (using the noisy initialization defined above), and then follow the recursive preparation procedure based on two qubit Pauli measurements. During the recursive procedure, initialization operations, $\cnot$ gates, and single-qubit measurements  are replaced by their noisy versions, and errors generated at some point by  noisy operations are propagated throughout the rest of the procedure.
%\footnote{We do not simulate the preparation circuit, but rather how errors propagate through the circuit.}. 
If an error is detected at any recursion level (according to the error detection method from Procedure $2$ in the main text), then  we discard the whole procedure and restart from the beginning, by initializing the $N$ qubits in a Pauli $Z$ basis state.  See also the example in Fig.~\ref{fig:qpolarprep_N8_i3_expanded}.

\smallskip As qubits are initialized in Pauli $Z$ basis, we may ignore the first $k, 1 \leq k \leq n$ levels of recursion in case Pauli $Z \otimes Z$ measurements are needed to be applied consecutively for the first $k$ levels of recursion correspond to applying Pauli $Z \otimes Z$ measurements. This reduces the number of component in the preparation circuit, hence improving the probability that the preparation succeeds.

\smallskip At the end of recursion, we have the knowledge of the $\ket{\bm{u}}_{\mathcal{Z}(n)}$ and $\ket{\bm{v}}_{\mathcal{X}(n)}$, hence the knowledge of prepared $\pone$ code state. Further, we also know the total $X$ and $Z$ type errors on the prepared state, which we will use for simulating Steane's error correction in the next sections.

\smallskip Further, we determine the preparation rate of the measurement based procedure as follows. We run the preparation procedure $R > 0$ times, and denote by $t$ the number of times the preparation completed (that is, no error has been detected during the preparation procedure). Then, the preparation rate, denoted by $p_{\text{prep}}$, is defined as follows,
\begin{equation}
p_{\text{prep}}= \lim_{R \to \infty} \frac{t}{R}.
\end{equation}

\paragraph*{Prepared Codes.} We consider $\pone$ codes of length $N = 16$ and $N= 64$ qubits. We choose the information position $i$ according to the results in Table~\ref{tab:b-indx} (best positions of the corresponding lengths, for sufficiently low  error rate), assuming a depolarizing noise model and ignoring correlations between $X$ and $Z$ errors. Thus, for $N = 16$, we take the information position $ i = 	7$, and for $N = 64$, we take $i = 23$. 

\smallskip The fact that the $\pone$ code construction ignores correlations between between $X$ and $Z$ errors is due to the Steane error correction. Indeed, $X$ and $Z$ errors are corrected independently, considering ancilla states prepared in either a logical $X$ basis state $\ket{\widetilde {\oline{\bm{w'}}}}_{\mathcal{S}'}$, or a logical $Z$ basis state $\ket{\widetilde{\bm{w'}}}_{\mathcal{S}'}$, respectively. During the $X$-error correction step (Fig.~\ref{fig:steane_Xerror}), $Z$ errors that happened on the ancilla system $\mathcal{S}'$ while preparing $\ket{\widetilde {\oline{\bm{w'}}}}_{\mathcal{S}'}$, are copied to the original system $\mathcal{S}$. Similarly, during the $Z$-error correction step (Fig.~\ref{fig:steane_Zerror}), $X$ errors that happened on the ancilla system $\mathcal{S}'$ while preparing $\ket{\widetilde{\bm{w'}}}_{\mathcal{S}'}$, are copied to $\mathcal{S}$. Clearly, these $X$ and $Z$ errors are decorrelated, since they happened during the preparation of different logical states.

\begin{figure}[!t]
\,\hfill\includegraphics[width = 0.495\linewidth]{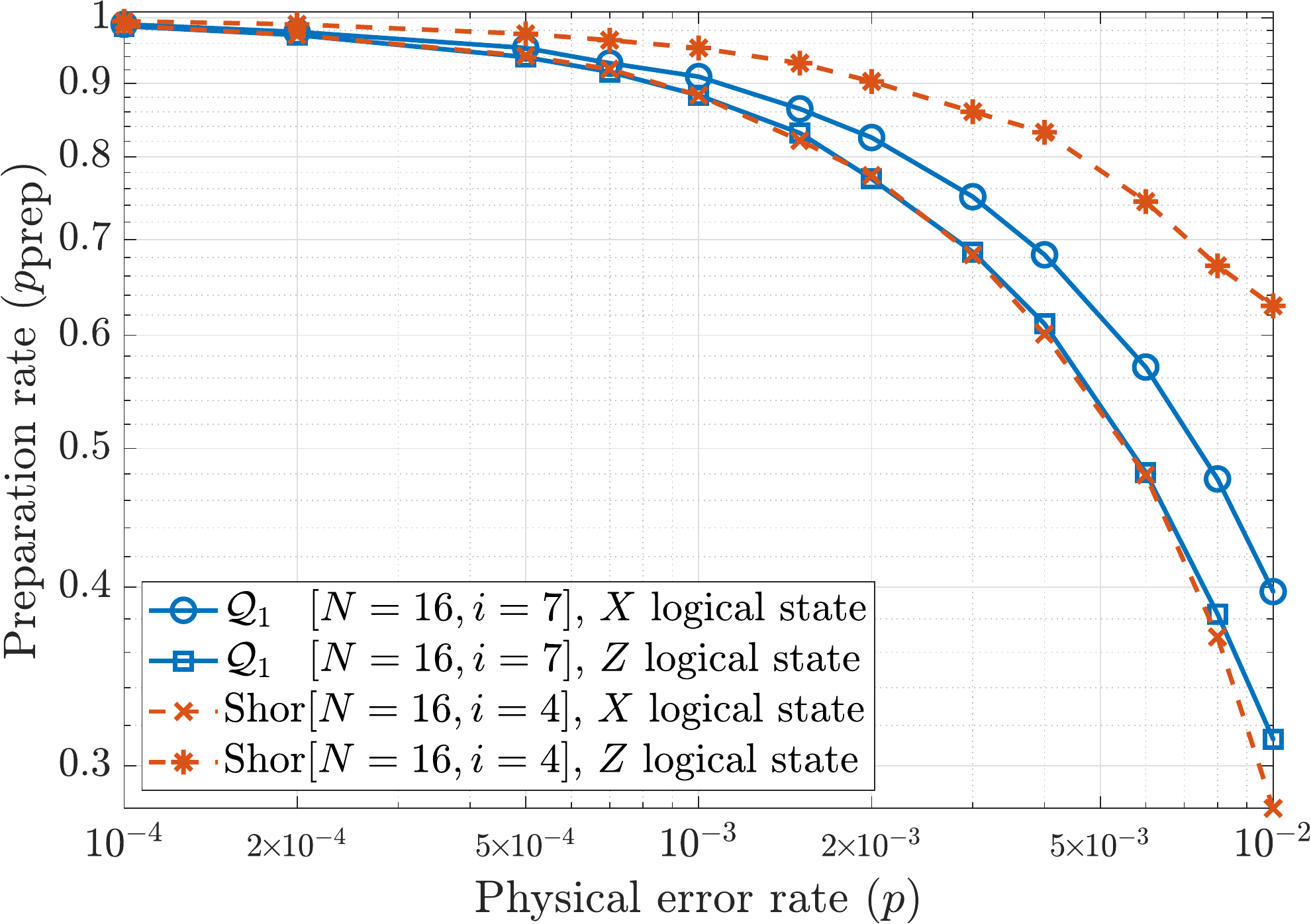} \hfill\,
\,\hfill\includegraphics[width = 0.485 \linewidth]{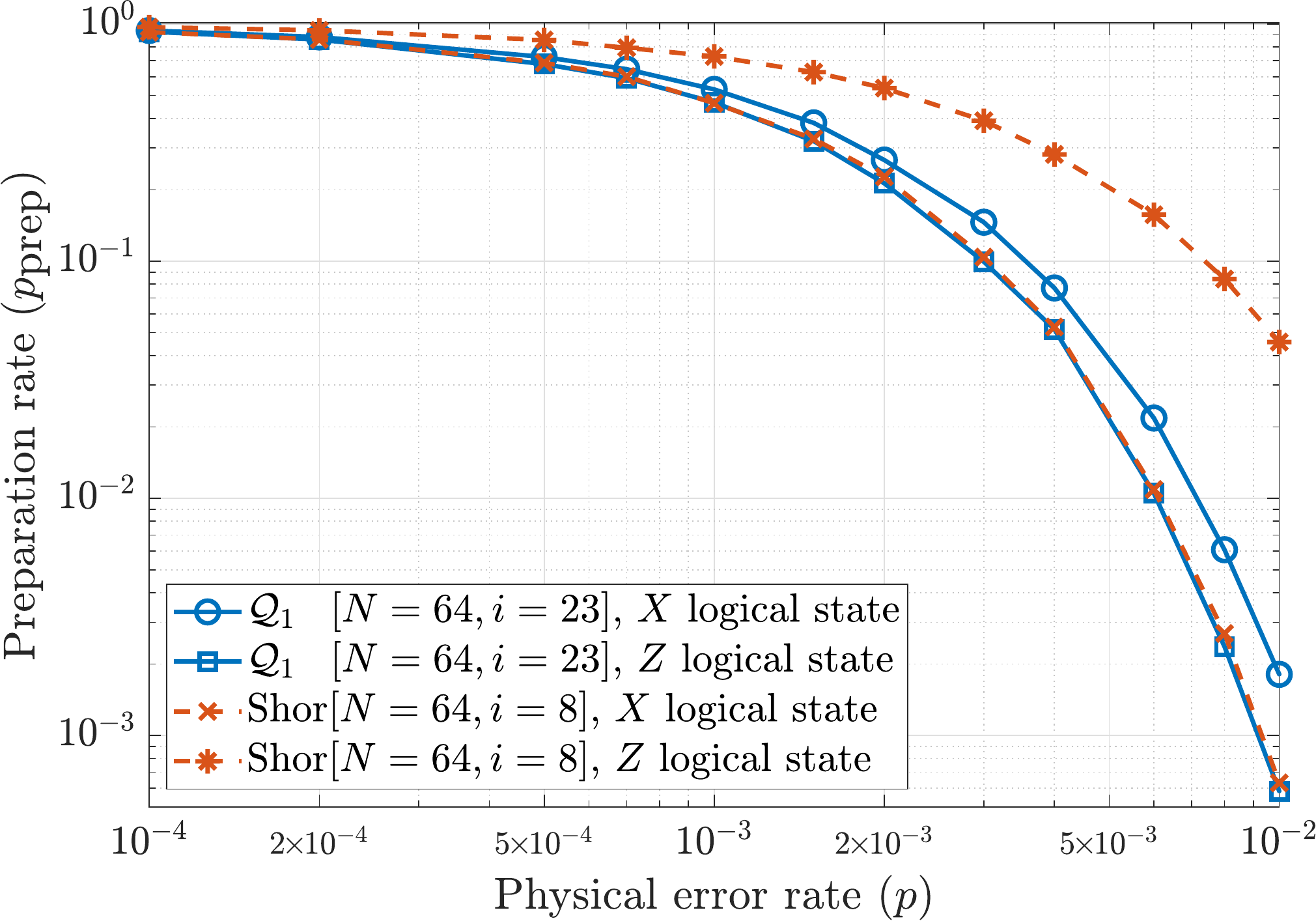} \hfill\,
\captionsetup{justification=centering}
\vspace*{-5mm}\caption{Preparation rate for the logical $X$ and logical $Z$ code states of $\pone$ codes $(N = 16, i = 7)$ and $(N = 64, i = 23)$.}
\label{fig:prep_rate}
\end{figure}

%\smallskip We observe that $p_{\text{prep}}$ is symmetric, in the sense that it is virtually the same for the $X$ and $Z$ logical states (the corresponding curves are superimposed), for both $\pone$ codes. Further, $p_{\text{prep}}$ is non-zero for sufficiently low $p$ and it approaches to $1$ as $p$ goes to zero, for both $\pone$ codes. However, $p_{\text{prep}}$ is much lower for $N = 64, i = 23$, compared to  $N = 16, i = 7$, especially at higher values of $p$. This is expected, as for $N = 64$ it is more likely that we detect an error, hence, discard the preparation, compared to $N = 16$. 

\smallskip Fig.~\ref{fig:prep_rate} shows the preparation rate $p_{\text{prep}}$, for the logical $Z$ and logical $X$ code states of $Q_1$ and Shor codes of length $N = 16, 64$, with respect to the physical error rate $p$. Here, the total number of runs is $R = 10^5$. 
 We observe that $p_{\text{prep}}$ is non-zero for sufficiently low $p$ and it approaches to $1$ as $p$ goes to zero, for all the $\pone$ code states. Further, for $p$ close to zero, $p_{\text{prep}}$ is symmetric, in the sense that it is virtually the same for all code states for a given $N$. For bigger values of $p$, the difference in $p_{\text{prep}}$  for different code states of the same length is explained by the fact that for some code states, we have Pauli $Z \otimes Z$ measurements at the first levels of recursion, which may be ignored, hence, reducing the number of componenents in the preparation circuit, as explained before. For example, for $Z$ logical states of Shor codes $N = 16, i = 4$ and $N = 64, i = 23$, we may ignore the first two and three levels of recursions, respectively, explaining their higher $p_{\text{prep}}$ compared to other code states of the same length. 
Furthermore, $p_{\text{prep}}$ is much lower for codes states of length $N = 64$, compared to that of length $N = 16$, especially at higher values of $p$. This is also explained by the fact that the preparation circuit, in general, consists of much larger number of components for $N = 64$ than $N = 16$.   

\newpage
\subsection{Monte-Carlo Based Estimates of the Logical Error Rates}
\makeatletter
\def\@currentlabel{\thesubsection}
\makeatother
\label{sec:lgr}

In this section, we provide a practical method to estimate logical error rates of $\pone$ codes (given in Fig. $4$ of the main text and see also Fig. \ref{fig:p1_vs_shor}), under the Steane's error correction. We consider a $\pone$ code of length $N$, with information position $i$, and consider the following notation $\mathcal{I} = \{i\}$, $\mathcal{Z} = \{1, \dots, i-1\}$,  $\mathcal{X} = \{i+1, \dots, N\}$, and $\mathcal{S} =   \mathcal{Z} \cup \mathcal{I} \cup \mathcal{X}$.

 %which is used to generate the numerical results on logical error rates of $\pone$ codes given in Fig. $4$ of the main text. 

\paragraph*{Logical $X$ Error Rate.} To determine the logical $X$ error rate, we prepare a logical Pauli $Z$ basis code state  $\ket{\widetilde{w}}_\mathcal{S} =  Q_N ( \ket{\bm{u}}_{\mathcal{Z}} \otimes \ket{w}_{\mathcal{I}} \otimes \ket{\oline{\bm{v}}}_{\mathcal{X}})$,   $w \in \{0, 1\}$, that we want to protect against $X$ errors. Further, we prepare a logical Pauli $X$ basis code state $\ket{\widetilde{\oline{w'}}}_{\mathcal{S}'} = Q_N ( \ket{\bm{u'}}_{\mathcal{Z}'} \otimes \ket{\oline{w'}}_{\mathcal{I}'} \otimes \ket{\oline{\bm{v'}}}_{\mathcal{X}'})$,  $w' \in \{0, 1\}$, to be used as the ancilla system for syndrome extraction (see Fig.~\ref{fig:steane_Xerror}.)

\smallskip The logical $X$ error rate is determined as follows.
\begin{description}[labelindent=\parindent]
\item[\em (1) Preparing $\pone$ code states]  We first simulate the measurement based preparation of $\ket{\widetilde{w}}_\mathcal{S}$ and $\ket{\widetilde{\oline{w'}}}_{\mathcal{S}'}$ states, as described before (in case of error detection, we restart the preparation procedure until it completes). After the preparation completes,  we know the frozen vectors $\bm{u}, \bm{u'}  \in \{0, 1\}^{i-1}$, corresponding to the frozen sets $\mathcal{Z}$ and $\mathcal{Z}'$,  and the frozen vectors $\bm{v}, \bm{v'} \in \{0, 1\}^{N-i}$ corresponding to the frozen sets $\mathcal{X}$ and $\mathcal{X}'$. Further, for $\ket{\widetilde{w}}_\mathcal{S}$, we know the logical $Z$ value $w$, and for $\ket{\widetilde{\oline{w'}}}_{\mathcal{S}'}$ we know the logical $X$ value $w'$. Moreover, we also have the final errors $\bm{e}_X$ and $\bm{e'}_X$ on systems $\mathcal{S}$ and $\mathcal{S}'$, respectively.

\item[\em (2) Generating the syndrome] We generate the syndrome according to Steane's procedure (Fig.~\ref{fig:steane_Xerror}). According to Lemma~\ref{lem:steane}, the syndrome $\bm{m}$ consists of a noisy version of a random codeword of the classical polar code $P(N, \mathcal{Z}, \bm{u} \oplus \bm{u'})$. We generate $\bm{m}$ as the sum of the random codeword from $(2.1)$ and the three error terms from $(2.2)$, below. 

%$\bm{m} = P_N( \bm{u} \oplus \bm{u'}, \bm{a'}, \bm{x'}) \oplus \bm{e}_X  \oplus \bm{e'}_X \in \{0,1\}^N, $
%
%xx To generate a noisy codeword according to $(\ref{eq:steane-syn})$, 
%we first generate a codeword of the classical polar code $P(N, \mathcal{Z}, \bm{u} \oplus \bm{u'})$, and then error terms are added to this codeword, as follows,
%%
\begin{itemize}
\item[$(2.1)$] First, we generate a \textbf{\em random codeword} $P_N( \bm{u} \oplus \bm{u'}, a', \bm{x'})$, by taking  random values $a' \in \{0, 1\}$ and $\bm{x'} \in \{0, 1\}^{N-i}$.

%\footnote{Note here that instead of $\bm{0}$, we may take a random $\bm{x} \in \{0, 1\}^n$.}.

\item[$(2.2)$]  We then add to the generated codeword,  the following 

%\textbf{\em two error terms}\footnote{In addition to the preparation errors (\emph{i.e.}, $\bm{e}_X$ and $\bm{e'}_X$) considered in Lemma~\ref{lem:steane}, we also consider here error generated by the qubitwise $\cnot_{\mathcal{S} \to {\mathcal{S}}'}$ and the Pauli $Z$ measurements, within the Steane's error correction procedure.}

\begin{itemize}
\item[$\bullet$] The first error term is $\bm{e}_X \oplus \bm{e'}_X$, where $\bm{e}_X$ and $\bm{e'}_X$ are given in step $(1)$ (see also~(\ref{eq:steane-syn})).

\item[$\bullet$] The second term corresponds to the $X$ error generated on the system $\mathcal{S}'$ during the implementation of the qubitwise  $\cnot_{\mathcal{S} \to {\mathcal{S}}'}$.

\item[$\bullet$]  The third error term is due to Pauli $Z$ measurements on the system $\mathcal{S}'$. 

\end{itemize}

\end{itemize}

Let $\bm{e}_X^{C}$ and $\bm{e}_X^{C^\prime}$ be the error produced by $\cnot_{\mathcal{S} \to {\mathcal{S}}'}$ on $\mathcal{S}$ and $\mathcal{S}'$, respectively. Further, let $\bm{e}_X^{M^\prime}$ is the error produced by Pauli $Z$ measurements on $\mathcal{S}'$. Then, we have that
%
%Denoting the second error term above by $\bm{e''}_X \eqdef \bm{e}_X^{C(1)} + \bm{e}_X^M$, we may write,
%\begin{equation}
%\bm{m} = P_N( \bm{u} \oplus \bm{u'}, a', \bm{x'}) \oplus \bm{e}_X  \oplus \bm{e'}_X  \oplus \bm{e}_X^{C'} \oplus \bm{e}_X^{M^\prime}. \label{eq:n-cword-dec1}
%\end{equation}
%\begin{equation}
%\bm{m} = P_N( \bm{u} \oplus \bm{u'}, a', \bm{x'}) \oplus \bm{e}_X^\text{tot}, \quad \text{ where } \ \bm{e}_X^\text{tot} := \bm{e}_X  \oplus \bm{e'}_X  \oplus \bm{e}_X^{C'} \oplus \bm{e}_X^{M^\prime}. \label{eq:n-cword-dec1}
%\end{equation}
\begin{equation}
\bm{m} = P_N( \bm{u} \oplus \bm{u'}, a', \bm{x'}) \oplus \bm{e}_X^\text{tot}, \label{eq:n-cword-dec1} 
\end{equation}
where  $\bm{e}_X^\text{tot} := \bm{e}_X  \oplus \bm{e'}_X  \oplus \bm{e}_X^{C'} \oplus \bm{e}_X^{M^\prime}$ is the total $X$-error that happened on the system $\mathcal{S}'$.

Further, we also update the error on system $\mathcal{S}$, by adding $\bm{e}_X^{C}$ to $\bm{e}_X$, \emph{i.e.}, adding the $X$ error that happened on  system $\mathcal{S}$ during the implementation of the qubitwise $\cnot_{\mathcal{S} \to {\mathcal{S}}'}$.

%the total error $\bm{e}_X^\text{tot} := \bm{e}_X  \oplus \bm{e'}_X  \oplus \bm{e}_X^{C'} \oplus \bm{e}_X^{M'}$

%. For this, we compute $\bm{u} \oplus \bm{u'}$ according to step $(1)$, and generate a codeword $P_N(\bm{u} \oplus \bm{u'}, \bm{0})$ \footnote{Note here that instead of $\bm{0}$, we may take a random $\bm{x} \in \{0, 1\}^n$.}. Then, we get the syndrome $\bm{m} = P_N(\bm{u} \oplus \bm{u'}, \bm{0}) \oplus \bm{e}_x \oplus \bm{e}'_x$, where the $\bm{e}_x$ and $\bm{e}'_x$ are according to step $(1)$ in $P(N, \mathcal{Z}, \bm{u} \oplus \bm{u'})$.

%we may take random $\bm{a} \in \{0,1\}$, and $\bm{x}$ (see (\ref{eq:X-syndrom-eq})) and then apply the polar transform on
%that is, $P_N( \bm{u} \oplus \bm{u'}, \bm{a}, \bm{x})$
%Then, we generate noisy version of a random codeword as in (\ref{eq:steane-syn}). Further, in (\ref{eq:steane-syn}), we may take the frozen value $\bm{u} \oplus \bm{u'}$ to be the all zero vector as the decoding performance does not depend on the frozen value. Further, we can also take the the random $(\bm{a}, \bm{x}) \in \{0, 1\}$ to be the all zero vector. Hence, the noisy version.

\item[\em (3) Error correction] Given the frozen value $\bm{u} \oplus \bm{u'}$ from step (1), and the extracted syndrome $\bm{m}$ from step (2), we use SC decoding\footnote{The SC decoding assumes that $e_X^\text{tot}$ is a vector of independent and identically distributed Bernoulli variables, that is $P(e_X^\text{tot}(i) = 1) = p_X^\text{in}$, for all $i = 1,\dots,N$, and for some probability value $p_X^\text{in} \in (0, 1)$, which is referred to as the \emph{input error probability of the decoder}. While the ``belief-propagation'' variant of the SC decoding needs an estimate of $p_X^\text{in}$ (in order to calculate the input log-likelihood ratios of the decoder, $\text{LLR}(i) := (-1)^{m(i)}\log((1-p_X^\text{in})/p_X^\text{in}), \forall i = 1,\dots,N$), here we use the so-called ``min-sum'' approximation of the SC decoding, which does not need such an estimate (in this case, the input log-likelihood ratio values of the decoder may be simply set to $\text{LLR}(i) := (-1)^{m(i)}, \forall i = 1,\dots,N$). \label{foot:SC_decoding}} to get an estimate $\hat{a}' \in \{0, 1\}$ of $a'$, and then generate an estimate of the total error $\bm{e}_X^\text{tot}$ in~(\ref{eq:n-cword-dec1}), as follows
\begin{equation} \label{eq:exguess}
\hat{\bm{e}}_X^\text{tot} = \bm{m} \oplus P_N( \bm{u} \oplus \bm{u'}, \hat{a}', \bm{0}).
\end{equation}
Note that we do not need to estimate $\bm{x'}$ here, since the induced logical error corresponds to an $X$-type stabilizer operator, acting trivially on the code space (see also the discussion after Lemma~\ref{lem:steane}).

We then perform error correction on the state of the system $\mathcal{S}$, and update  the $X$ error on $\mathcal{S}$, by adding the estimated total error $\hat{\bm{e}}_X^\text{tot}$. 
Let $e_{X}^{\text{corc}}$ be the the error on system $\mathcal{S}$ after correction. From (\ref{eq:n-cword-dec1}) and (\ref{eq:exguess}), it follows that $e_{X}^{\text{corc}}$ is equivalent to the following error,
\begin{equation} \label{eq:err-S-after-corc}
e_{X}^{\text{corc}} \equiv  \bm{e}'_X \oplus \bm{e}_X^{C} \oplus \bm{e}_X^{C'} \oplus \bm{e}_X^{M'} \oplus P_N(\bm{0}, a' \oplus \hat{a}', \bm{0}, \bm{0}).
\end{equation}

\item[\em (4) Guessing the logical value] The error correction is successful if we can  successfully recover the logical $Z$ value from the state of the system $\mathcal{S}$, after error correction.

To get the logical $Z$ value, we need to perform single-qubit Pauli $Z$ measurements on system $\mathcal{S}$, and then estimate the logical value from the measurement outcome. 

%If our estimate is equal to $u$, the known logical $Z$ value of the initial state, this means the decoding succeeded. otherwise it failed.

It can be seen that the measurement outcome of Pauli $Z$ measurements gives a noisy version of a random codeword of the classical polar code $P(N, \mathcal{Z}, \bm{u} \oplus \bm{u'})$, as follows, 
\begin{equation} \label{eq:m-out-singleZ}
\bm{m} = P_N(\bm{u} \oplus \bm{u'},w,  \bm{x}) \oplus \bm{e}_X^{\text{corc}} \oplus \bm{e}_X^M,
\end{equation}
where  $w$ is the logical $Z$ value corresponding to the initial state on the system $\mathcal{S}$, $\bm{x}$ is a random vector, $\bm{e}_X^{\text{corc}}$ is the $X$ error on the system $\mathcal{S}$ after the error correction from (\ref{eq:err-S-after-corc}), and $\bm{e}_X^M$ is the error caused by the single-qubit Pauli $Z$ measurements on system $\cl{S}$.

From the frozen vector $\bm{u} \oplus \bm{u'}$ and the noisy codeword $\bm{m}$ in (\ref{eq:m-out-singleZ}), we generate an estimate $\hat{w}$ of $w$, using the SC decoding. If $\hat{w} \neq w$, we report a decoding failure, otherwise the decoding succeeds.
\end{description}

We run the above steps (1)-(4), until we report $f > 0$  decoding failures. Let $R$ be the number of runs for $f$ decoding failures, then  the logical $X$ error rate, denoted by $P_X^L$, is computed as follows,
\begin{equation}
P_X^L = \frac{f}{R}. \label{eq:log_eratex}
\end{equation}
For our numerical simulations, the value of $f$ varies between $50$ to $200$, depending on the physical error rate $p$ of our noise model.

\paragraph*{Logical $Z$ Error Rate.} To determine the logical $Z$ error rate, we consider a logical Pauli $X$ basis code state  $\ket{\widetilde{\oline{w}}}_\mathcal{S} =  Q_N ( \ket{\bm{u}}_{\mathcal{Z}} \otimes \ket{\oline{w}}_{\mathcal{I}} \otimes \ket{\oline{\bm{v}}}_{\mathcal{X}})$,   $w \in \{0, 1\}$, that we want to protect against $Z$ errors. Further, we consider a logical Pauli $Z$ basis code state $\ket{\widetilde{w'}}_{\mathcal{S}'} = Q_N ( \ket{\bm{u'}}_{\mathcal{Z}'} \otimes \ket{w'}_{\mathcal{I}'} \otimes \ket{\oline{\bm{v'}}}_{\mathcal{X}'})$,  $w' \in \{0, 1\}$, to be used as the ancilla system for syndrome extraction (Fig.~\ref{fig:steane_Zerror}.) The logical $Z$ error rate is then determined similarly to the above description for logical $X$ error rate, while inverting $X$ and $Z$ bases, and using Lemma~\ref{lem:steane_Z} instead of Lemma~\ref{lem:steane}. The logical $Z$ error rate is denoted by $P_Z^L$.

\medskip Finally, the {\bfseries\itshape logical error rate} (accounting for both $X$ and $Z$ errors), denoted by $P_e^{L}$, is given as follows,
\begin{equation}
P_e^{L} = P_X^L + P_Z^L - P_X^L P_Z^L.  \label{eq:log_rate}
\end{equation}

\subsection{Density-Evolution Based Estimates of the Logical Error Rates}

To reliably estimate the logical error rate value, the Monte-Carlo method described in Section~\ref{sec:lgr} requires an increasingly large number of simulations as the logical error rate decreases. Hence, as $p$ decreases, this requires increasingly more computational time (or resources) and  becomes practically unfeasible for small values of $p$. In this section, we provide a theoretical method to estimate the logical error rate, based on density evolution~\cite{tal2013construct} (see also Fig.~\ref{fig:p1_vs_shor}).

\smallskip Consider the Steane's error correction for $X$ errors from Section \ref{sec:lgr}. Note that it involves two steps of decoding; the first during Step (3) for error correction and the second during Step (4) for guessing the logical value. Each failure in decoding introduces a logical $X$ error, hence, the logical value is incorrectly determined if and only if one of the two decoder fails.

\smallskip For a given realization of errors (\emph{i.e.}, $e_X^\text{tot}$ in~(\ref{eq:n-cword-dec1}), or $\bm{e}_X^{\text{corc}} \oplus \bm{e}_X^M$ in~(\ref{eq:m-out-singleZ})), the SC decoding works by propagating the corresponding log-likelihood ratio (LLR) values throughout the polar encoding graph (from the right-hand side of the graph, corresponding to the encoded information, to the left-hand side of the  graph, corresponding to the uncoded information; see also footnote~\ref{foot:SC_decoding}). Rather than propagating LLR values for a given sample (realization of errors), the density evolution method propagates their probability density functions, thus averaging over all the sample space. To determine the probability distribution of the input LLRs, one needs to estimate the input error probability of the decoder (\emph{i.e.}, $P(e_X^\text{tot}(i) = 1)$ for the SC decoding in Step~(3), or $P(\bm{e}_X^{\text{corc}}(i) \oplus \bm{e}_X^M(i) = 1)$ for the SC decoding in Step~(4)).

\begin{figure}[!b]
\begin{subfigure}{\textwidth}
\includegraphics[width=0.495\linewidth]{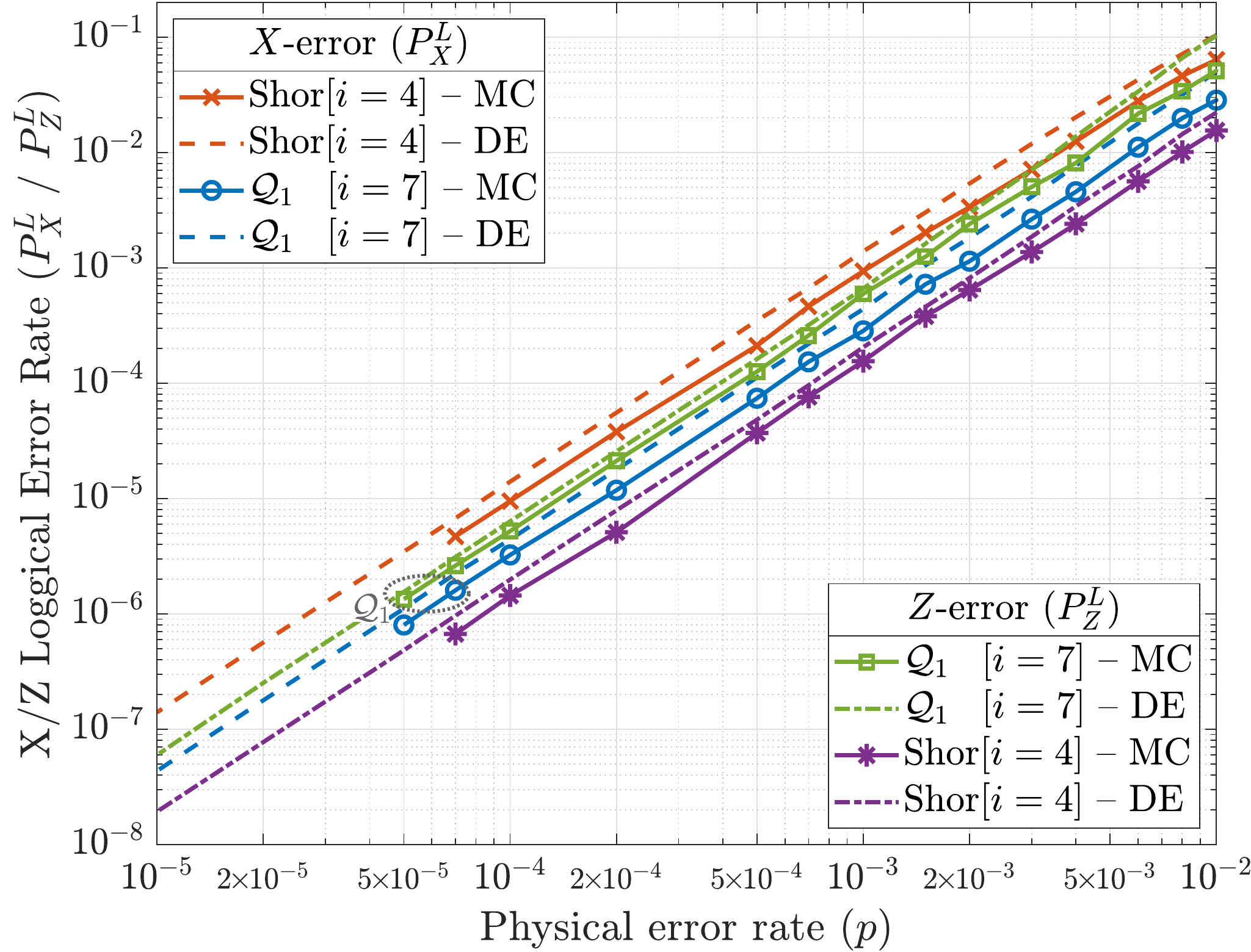}\hfill%
\includegraphics[width=0.495\linewidth]{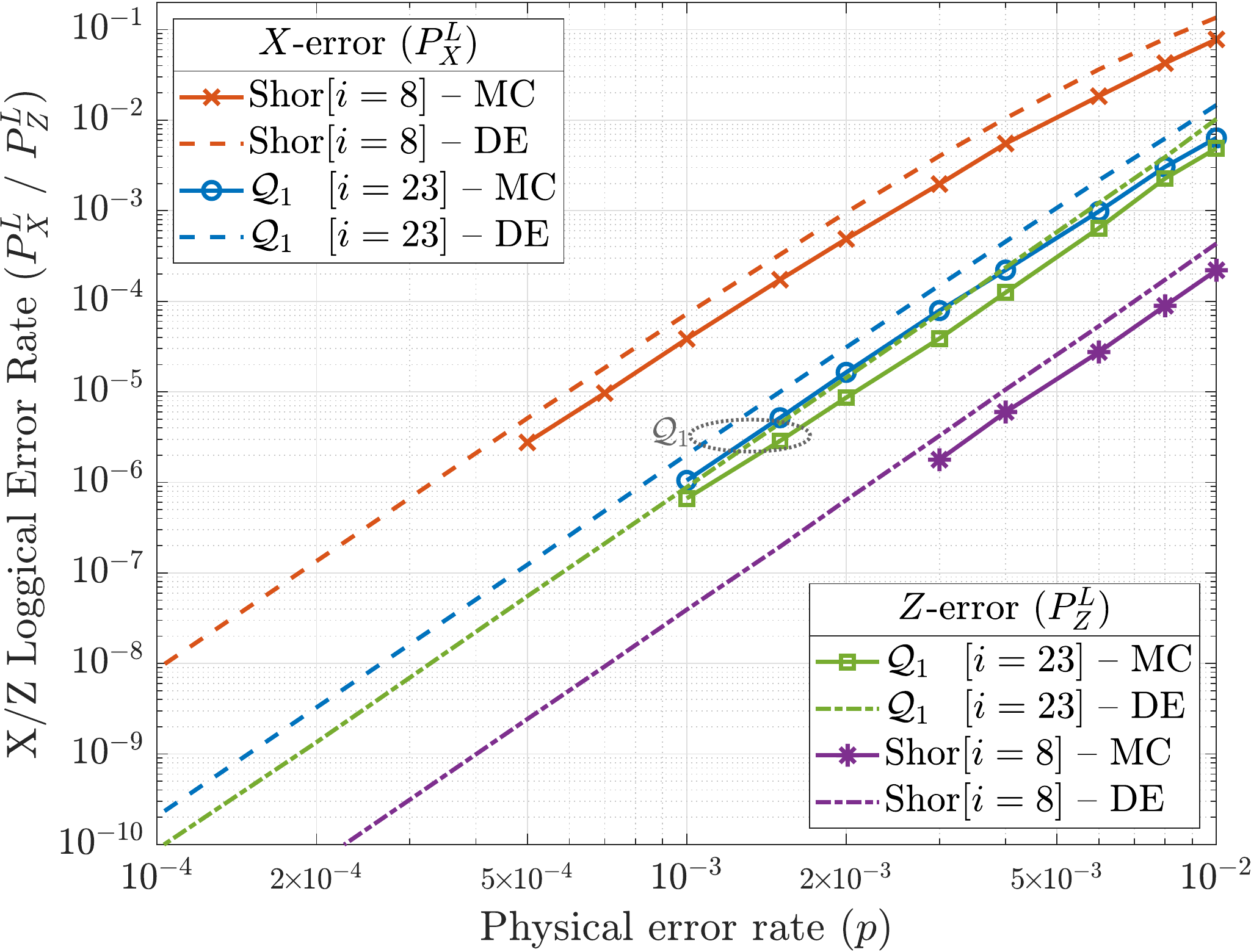}
\captionsetup{justification=centering}
\caption{Monte-Carlo (MC) and Density-evolution (DE) based
estimates of the $X/Z$ logical error rates of $\pone$ and Shor codes of length $N\!=\!16$ and $N\!=\!64$.}
\end{subfigure}

\vspace*{3mm}
\begin{subfigure}{\textwidth}
\,\hfill\includegraphics[width = 0.55\linewidth]{lger_rate_N16_N64.pdf}\hfill\,
\captionsetup{justification=centering}
\caption{Monte-Carlo (MC) and Density-evolution (DE) based
estimates of the logical error rate of $\pone$ and Shor codes of length $N=16$ and $N=64$ (here, the logical error rate $P_e^{L} = P_X^L + P_Z^L - P_X^L P_Z^L$, where $P_X^L$ and $P_Z^L$ are given in (a) above). }
\end{subfigure}
\captionsetup{justification=centering}
\caption{Numerical results for $\pone$ and Shor codes of length $N=16$ and $N=64$.}
\label{fig:p1_vs_shor}
\end{figure}

\smallskip Accordingly, we determine the logical $X$ error rate in the following two steps.

\begin{description}[labelindent=\parindent]
%\smallskip We estimate the failure rates of both decoders, using the density evolution, which takes as input the error probabilities associated with the error term in the corresponding noisy codewords. 

\item[\em (1) Input error probabilities of decoders] To determine the input error probabilities of the decoders in Step $(3)$ and $(4)$, we need to generate statistics for the error terms $e_X^\text{tot}$ in~(\ref{eq:n-cword-dec1}) and $\bm{e}_X^{\text{corc}} \oplus \bm{e}_X^M$ in~(\ref{eq:m-out-singleZ}).  In particular, we need to estimate the output $X$ error probabilities of the $Z$ and $X$ logical states, by numerical simulation. This is done as follows.

For a given physical error rate $p$, we run the measurement based preparation with the error detection for logical $Z$ and $X$ states until we have $R$ successful preparations for each one of them \footnote{In our numerical simulation, we have taken $R = 100/p$.}. Let $\bm{e}_X^r, \bm{e}_X^{ \prime\, r} \in \{0, 1\}^N, 1 \leq r \leq R$, be the $X$ errors corresponding to the $r^\text{th}$ successful preparation of the logical $Z$ and $X$ states, respectively.  Then, the average output $X$ error rate for $Z$ and $X$ logical states are estimated as follows,
\begin{align}
p_X^{ \text{prep}} = \frac{1}{RN} \sum_{r = 1}^R \wt(\bm{e}_X^r). \\
p_X^{ \prime\, \text{prep}} = \frac{1}{RN} \sum_{r = 1}^R \wt(\bm{e}_X^{\prime\, r}).
\end{align}

\smallskip Let $p_X^{\text{in}_1}$ and $p_X^{\text{in}_2}$ be the input error rate of the first and second decoders, respectively. From (\ref{eq:n-cword-dec1}), we have that
%and (\ref{eq:m-out-singleZ}) that,
%
\begin{equation} \label{eq:dec-1-out}
p_X^{\text{in}_1} = 1 - (1- p_X^{ \text{prep}}) (1- p_X^{ \prime\,  \text{prep}}) (1-\frac{8p}{15}) (1- p)
\end{equation}
Further, for the sake of simplicity, we compute the input error rate of the second decoder, assuming that the first decoder has succeeded. Then, from (\ref{eq:err-S-after-corc}) and (\ref{eq:m-out-singleZ}), we have that
\begin{equation} \label{eq:dec-2-out}
p_X^{\text{in}_2}\ = 1 - (1- p_X^{ \prime\,  \text{prep}}) (1-\frac{8p}{15}) (1- p)^2.
\end{equation}
Note that as long as the first decoder succeeds with a probability close to one, (\ref{eq:dec-2-out}) approximates well the input error rate of the second decoder.

\item[\em (2) $X$ logical error rate] After obtaining the input error probabilities $p_X^{\text{in}_1}$ and $p_X^{\text{in}_2}$, we compute output error rates of the decoders\footnote{Precisely, we determine the probability distributions of the input LLRs of the decoder, propagate these probability distributions numerically throughout the polar encoding graph, then take the corresponding output error rates.}, using density evolution. Let $P_X^{\text{out}_1}$ and $P_X^{\text{out}_2}$ be the output error probabilities of the first and second decoders, respectively. Then, we determine the $X$ logical error rate of the Steane's error correction as follows,
\begin{equation}
P_X^L = 1- \big( (1- p_X^{\text{out}_1}) (1- p_X^{\text{out}_2}) + p_X^{\text{out}_1} p_X^{\text{out}_2}\big).
\end{equation} 
\end{description}

We may similarly detemine the $Z$ input error probabilites for the first and second decoder, estimate their output probabilites using density evolution, and then detemine $P_Z^L$, $i.e.$, the $Z$ logical error rate of the Steane's error correction. Finally, using $P_X^L$ and $P_Z^L$, we may determine the logical error rate $P_e^{L}$ as in (\ref{eq:log_rate}). 

\subsection{Numerical Results}
Numerical results for the logical error rates of $\pone(N = 16, i = 7)$ and $\pone(N = 64, i = 23)$ codes, based on either Monte-Carlo simulation or density evolution, are provided in Fig.~\ref{fig:p1_vs_shor}. It can be seen that  density evolution based estimates closely upper bound those obtained though Monte-Cralo simulation, thus  providing a trustworthy extrapolation of the logical error rate values for small values of $p$.

\newpage

\bibliographystyle{unsrt}
\bibliography{biblio_database}

\begin{thebibliography}{10}

\bibitem{preskill1998fault}
John Preskill.
\newblock Fault-tolerant quantum computation.
\newblock In {\em Introduction to quantum computation and information}, pages
  213--269. World Scientific, 1998.

\bibitem{gottesman2010introduction}
Daniel Gottesman.
\newblock An introduction to quantum error correction and fault-tolerant
  quantum computation.
\newblock In {\em Quantum information science and its contributions to
  mathematics, Proceedings of Symposia in Applied Mathematics}, volume~68,
  pages 13--58, 2010.

\bibitem{arikan2009channel}
Erdal Arikan.
\newblock Channel polarization: A method for constructing capacity-achieving
  codes for symmetric binary-input memoryless channels.
\newblock {\em IEEE Transactions on Information Theory}, 55(7):3051--3073,
  2009.

\bibitem{renes2011efficient}
Joseph~M. Renes, Frédéric Dupuis, and Renato Renner.
\newblock Efficient polar coding of quantum information.
\newblock {\em Physical Review Letters}, 109(5):050504, August 2012.

\bibitem{wilde2013polar}
Mark~M. Wilde and Saikat Guha.
\newblock Polar codes for degradable quantum channels.
\newblock {\em IEEE Transactions on Information Theory}, 59(7):4718--4729, July
  2013.

\bibitem{renes2014polar}
Joseph~M. Renes and Mark~M. Wilde.
\newblock Polar codes for private and quantum communication over arbitrary
  channels.
\newblock {\em IEEE Transactions on Information Theory}, 60(6):3090--3103, June
  2014.

\bibitem{dupuis2021polarization}
Fr{\'e}d{\'e}ric Dupuis, Ashutosh Goswami, Mehdi Mhalla, and Valentin Savin.
\newblock Polarization of quantum channels using {Clifford}-based channel
  combining.
\newblock {\em IEEE Transactions on Information Theory}, 67(5):2857--2877,
  2021.
\newblock arXiv:1904.04713.

\bibitem{holmes2020nisq}
Adam Holmes, Mohammad~Reza Jokar, Ghasem Pasandi, Yongshan Ding, Massoud
  Pedram, and Frederic~T Chong.
\newblock {NISQ+}: Boosting quantum computing power by approximating quantum
  error correction.
\newblock {\em arXiv:2004.04794}, 2020.

\bibitem{krishna2018magic}
Anirudh Krishna and Jean-Pierre Tillich.
\newblock Magic state distillation with punctured polar codes.
\newblock {\em arXiv preprint arXiv:1811.03112}, 2018.

\bibitem{steane1997active}
Andrew~M Steane.
\newblock Active stabilization, quantum computation, and quantum state
  synthesis.
\newblock {\em Physical Review Letters}, 78(11):2252, 1997.

\bibitem{steane2002fast}
Andrew~M Steane.
\newblock Fast fault-tolerant filtering of quantum codewords.
\newblock {\em arXiv quant-ph/0202036}, 2002.

\bibitem{shor1995scheme}
Peter~W Shor.
\newblock Scheme for reducing decoherence in quantum computer memory.
\newblock {\em Physical review A}, 52(4):R2493, 1995.

\bibitem{bacon2006operator}
Dave Bacon.
\newblock Operator quantum error-correcting subsystems for self-correcting
  quantum memories.
\newblock {\em Physical Review A}, 73(1):012340, 2006.

\bibitem{tal2013construct}
Ido Tal and Alexander Vardy.
\newblock How to construct polar codes.
\newblock {\em IEEE Transactions on Information Theory}, 59(10):6562--6582,
  2013.

\bibitem{gottesman2014fault}
Daniel Gottesman.
\newblock Fault-tolerant quantum computation with constant overhead.
\newblock {\em Quantum Information \& Computation}, 14(15-16):1338--1372, 2014.

\bibitem{svore2006flow}
Krysta~M Svore, Andrew~W Cross, Isaac~L Chuang, and Alfred~V Aho.
\newblock A flow-map model for analyzing pseudothresholds in fault-tolerant
  quantum computing.
\newblock {\em Quantum Information \& Computation}, 6(3):193--212, 2006.

\bibitem{tomita2014low}
Yu~Tomita and Krysta~M Svore.
\newblock Low-distance surface codes under realistic quantum noise.
\newblock {\em Physical Review A}, 90(6):062320, 2014.

\bibitem{fowler2012surface}
Austin~G Fowler, Matteo Mariantoni, John~M Martinis, and Andrew~N Cleland.
\newblock Surface codes: Towards practical large-scale quantum computation.
\newblock {\em Physical Review A}, 86(3):032324, 2012.

\end{thebibliography}

\end{document}